\documentclass[twocolumn,pra,aps,showpacs]{revtex4}

\usepackage{amsmath}
\usepackage{amssymb}
\usepackage[dvips]{graphicx}
\usepackage{pstricks}
\usepackage{pst-node}
\usepackage{pst-plot}

\newcommand{\qed}{\hspace*{\fill}$\square$}
 \newtheorem{thm}{Theorem}
 \newtheorem{lem}[thm]{Lemma}
 \newtheorem{cor}[thm]{Corolary}
 \newtheorem{defn}[thm]{Definition}
 \newtheorem{prop}[thm]{Proposition}
 \newenvironment{proof}{\noindent \emph{Proof.}}{\qed}

 \newcommand{\C}{\mathbf{C}}

 \newcommand{\GL}[1]{\mathbf{GL}({#1})}
 \newcommand{\funcion}[3]{#1:\,#2\longrightarrow #3}

 \newcommand{\vect}[1]{\boldsymbol{\mathrm{#1}}}
 \newcommand{\set}[2]{ \{\,#1\,|\,#2\,\}}
 \newcommand{\sset}[1]{ \{#1\} }
 \newcommand{\genset}[2]{ (\,#1\,|\,#2\,)}

 \newcommand{\lin}{\mathrm{Lin}\,}
 
 \newcommand{\traza}{\mathrm{Tr}}
 
 \newcommand{\inv}{^{-1}}
 \newcommand{\inver}[1]{\bar #1}

 \newcommand{\prima}{^\prime}
 \newcommand{\primas}{^{\prime\prime}}
 
 \newcommand{\nin}{ \not\in}

 \newcommand{\ket}[1]{|#1\rangle}
 \newcommand{\bra}[1]{\langle #1|}
 \newcommand{\braket}[2]{\langle #1|#2\rangle}

 \newcommand{\vket}[1]{|\vect{#1}\rangle}

 \newcommand{\spc}{\hspace{0.5mm}}

 \newcommand{\D} {\mathcal D}
 \newcommand{\Hilb} {\mathcal H}
 \newcommand{\Alg} {\mathcal A}
 \newcommand{\Ribb} {\mathcal F}
 \newcommand{\Ribbclosed} {\mathcal K}
 \newcommand{\Conj}[1] {(#1)_{\mathrm{cj}}}
 \newcommand{\Irr}[1] {(#1)_{\mathrm {ir}}}
 \newcommand{\Ind}[1] {{#1}_{\mathrm {Ind}}}
 \newcommand{\Norm}[1] {\mathbf{N}_{#1}}
 \newcommand{\comp} {\bar}

 \newcommand{\complong}[1] {\overline {#1}}
 \newcommand{\alledges} {E_{\mathrm{ext}}}
 \newcommand{\Edual} {E^\bigtriangleup}
 \newcommand{\Edirect} {E^\bigtriangledown}
 \newcommand{\Edualcomp} {{\comp E}^\bigtriangleup}
 \newcommand{\Edirectcomp} {{\comp E}^\bigtriangledown}

 \newcommand{\xpctd}[1]{\langle#1\rangle}%_{\mathrm{GS}}}
\newcommand{\Ribbopen} {\mathcal J}

\begin{document}

\title[Short Title]{
A Family of Non-Abelian Kitaev Models on a Lattice: \\
Topological Condensation and confinement}

\author{H. Bombin and M.A. Martin-Delgado}
\affiliation{
Departamento de F\'{\i}sica Te\'orica I, Universidad Complutense,
28040. Madrid, Spain.
}

\begin{abstract}
We study a family of non-Abelian topological models in a lattice
that arise by modifying the Kitaev model through the introduction of
single-qudit terms. The effect of these terms amounts to a reduction
of the discrete gauge symmetry with respect to the original systems,
which corresponds to a generalized mechanism of explicit symmetry
breaking. The topological order is either partially lost or
completely destroyed throughout the various models. The new systems
display condensation and confinement of the topological charges
present in the standard non-Abelian Kitaev models, which we study in
terms of ribbon operator algebras.

\end{abstract}


\maketitle

\section{Introduction}
\label{sect_I}

The subject of topological orders poses new challenges in the
understanding of new phases of  matter due to novel effects in
quantum many-body physics \cite{wenbook04}. There is by now a good
deal of examples in condensed matter, like in fractional Hall effect
systems \cite{wenniu90}, \cite{wen90}, \cite{wen92}, \cite{wen91},
\cite{frohlichkerler91}, short range RVB (Resonating Valence Bond)
models \cite{roksharkivelson88}, \cite{readchakraborty89},
\cite{moessnersondhi01}, \cite{ardonne_etal04}. or in quantum spin
liquids \cite{kalmeyerlaughlin87}, \cite{wenwilczekzee89},
\cite{readsachdev91}, \cite{wen91}, \cite{senthilfisher00},
\cite{wen02}, \cite{sachdevparks02}, \cite{balentsfishergirvin02}.
There exists also exactly solvable models \cite{kitaev},
\cite{levinwen03},
\cite{levinwen05}, \cite{topodistill}, \cite{tetraUQC}, that are
paradigmatic examples for exhibiting topological properties that can
be addressed in full detail since the whole spectrum of those models
is known. Although topological orders typically arise in the quantum
physics of two spatial dimensions, it is possible to construct
exactly solvable models in three spatial dimensions and beyond
\cite{topo3D}. There is yet another field in which topological
orders appear naturally. It corresponds to discrete gauge theories
that arise as a consequence of a spontaneous symmetry breaking
mechanism from a continuous gauge group down to a discrete gauge
group \cite{Bais_80}, \cite{Bais_81}, \cite{Krauss_Wilczek_89},
\cite{Preskill_Krauss_90}, \cite{Wild_Bais_98}, \cite{BSS_02a},
\cite{BSS_02b}. In these two-dimensional topological quantum field
theories, the standard algebraic language to describe the residual
gauge invariant properties of the excitations is that of
quasitriangular Hopf algebras (quantum groups) \cite{SB_01}.

At the same time, quantum systems with topological order
provide new expectations for finding alternative ways of robust quantum
computation \cite{kitaev},  \cite{dennis_etal02},
\cite{bravyikitaev98}.
In fact, there are several forms to set up schemes for topological
quantum computation, some of them based on the braiding of quasiparticles
\cite{kitaev},  \cite{Ogburn99}, \cite{freedman_etal00a},
\cite{freedman_etal00b}, \cite{freedman_etal01}, \cite{mochon_04},
\cite{georgiev_06}, \cite{rmp_topo_07}, some of them based solely
on the topological entangled properties of the degenerate ground states,
without selective addressing of the physical qubits and without resorting
to braiding of excitations \cite{topodistill}, \cite{tetraUQC},
and others based on cluster states \cite{RHG_07}.

Topological orders can be thought of as new forms of long range
entanglement and they are at the crossroads of condensed matter and
quantum information \cite{rmp},  \cite{homologicalerror},
\cite{surfaceVsColor}, \cite{martindelgado04},
\cite{Dur_Briegel_07}, \cite{rico_briegel_07}, \cite{korepin_07},
 \cite{mosaic_spin_models_07}, \cite{lidar_06}.
Some forms of hidden topological orders in
quantum spin chains can be detected with string order parameters,
which in turn can be interpreted in the light of
quantum information techniques,
and their long-range entanglement detected with them \cite{martindelgado04}
using matrix product states from condensed matter.

There are experimental proposals based on optical lattices
\cite{duandemlerlukin03}, \cite{zoller05} to implement models with
Abelian topological orders \cite{exp_prop_07}, and in particular,
the study of the string order parameter mentioned above can also be proposed
by means of these techniques \cite{martindelgado04b}. There are also
proposals for non-Abelian models based on Josephson junction
arrays \cite{Ioffe_02}, \cite{doucot_04}, \cite{AKTB_07},
in addition to the largely studied case of the fractional quantum Hall effect
\cite{rmp_topo_07}.

One of the emblematic examples of exactly solvable models to study topological
orders on a lattice is the Kitaev model \cite{kitaev}, both in its Abelian and
non-Abelian versions. It captures the algebraic properties exhibited by the
discrete gauge theories mentioned above. In addition,
it provides us with an explicit realization of a Hamiltonian on a lattice,
with the bonus that it allows for a model of topological quantum computation.

Comparatively, there are much less works on the non-Abelian Kitaev model
than in the Abelian case (toric code). This is due, to some extent,
to the additional mathematical technical difficulties presented by the
non-Abelian case which is traditionally introduced with the language
of quasi-triangular Hopf algebras and their representations
\cite{DPR_90}, \cite{Bais_Driel_Wild_92}.
Here we have made an effort to explain its contents in full detail
and clarity with simpler algebraic tools based on group theory and
their representations. Our goal is twofold: to make the model
more accessible to a broader audience with a previous knowledge on
the Abelian toric code, and to use that simpler presentation as a
starting point for considering more general models.

In this paper we introduce a family of non-Abelian topological
models on a lattice, such that the standard Kitaev model corresponds
to a particular case. More specifically, we study a two-parameter
family labeled by a pair of subgroups $N\subset M\subset G$, $N$
normal in $G$, where $G$ is a discrete non-Abelian gauge group. The
particular case $N=1$, $M=G$ correspond to the original Kitaev
models. The Hamiltonians of the family, denoted $H_G^{N,M}$, are
explicitly constructed in eq. \eqref{Hamiltoniano_NM}. The standard
vertex ('electric') operators are modified according to the subgroup
$M$, while the face ('magnetic') operators change in accordance with
$N$. In addition, there are new terms entering in the Hamiltonians
which act on the edges of the lattice. Since there is a qudit
attached to each edge these are single-qudit terms. Depending on the
choice of the pair of subgroups $(M,N)$ with respect to $G$, the
non-Abelian discrete gauge group of the whole Hamiltonian
$H_G^{N,M}$ may range from $G$ down to the trivial group when $M=N$.
This is so because the gauge group for these models turns out to be
given by $G'=M/N$. Therefore, the new family of non-Abelian models
provides us with a mechanism of explicit symmetry breaking of an
original Hamiltonian with large discrete gauge symmetry group. In
other words, this mechanism can also be seen as a symmetry-reduction
mechanism, since we may have still a smaller gauge symmetry present
in the Hamiltonian.

The new edge terms do not commute with the vertex and face terms of
the original Hamiltonian, but this can be compensated by slightly
changing these vertex and face terms. This change corresponds to
studying the regimen in which the single-qudit terms have a higher
coupling constant. Choosing the models this way, we can study their
ground state and also the charge condensation phenomena. At least in
some cases, single-qudit terms can be understood as a mechanism for
introducing string tension, or more appropriately 'ribbon tension',
to some of the quasiparticle excitations which thus get confined. In
those cases, a complete characterization of the charge types and
domain wall fluxes will be given.

In order to facilitate both the exposition of the results and the
readability of the manuscript, throughout the main text we will be
giving the main constructions and results omitting many auxiliary
details or proofs. However, all these can be found in a well-ordered
form in a complete set of appendices.

We hereby summarize briefly some of our main results:

\noindent i/ we introduce a family of Hamiltonians defined on
two-dimensional spatial lattices of arbitrary topology which exhibit
a variety of discrete non-Abelian gauge group symmetry and
topological orders;

\noindent ii/ the ground state of the models can be exactly given
and characterized in terms of open a boundary ribbon operators. In
many interesting cases the spectrum of excitations can be
characterized accordingly;

\noindent iii/ the new models show condensation and confinement of
the charges in the original models with Hamiltonian $H_G$;

\noindent iv/ in order to facilitate and complement the study of the
family of models, we have carried out a thorough clarification of
the main properties of the standard non-Abelian Kitaev model. In
particular:

\noindent iv.a/ The ribbon operator algebra is introduced in an
intrinsic way, with the motivation to find operators that describe
excitations.

\noindent iv.b/ We study in detail and generalize the concept of
ribbon. In particular, closed ribbons and a related algebra are
defined, and their transformation properties described.

\noindent iv.c/ The vertex and face operators that appear in the
Hamiltonian are related to elementary closed ribbon operators,
showing that everything in the models can be translated to the
language of ribbons.

\noindent iv.d/ We give a detailed account of two-particle states,
giving explicitly a basis for the states that clarifies the meaning
of the labels for topological charge.

\noindent v/ A description of the ground state in terms of boundary
ribbon operators is given.

This paper is organized as follows: in Sect.\ref{sect_II} we treat
the standard non-abelian Kitaev model. We start explaining the terms
appearing in the Hamiltonian and go on characterizing the ground
state and quasiparticle excitations by means of closed ribbon
operators. We also present an explicit characterization of the
topological charges of the model and study when single-quasiparticle
states are possible. In Sect.\ref{sect_III} we motivate the new
family of non-Abelian model Hamiltonians and present their generic
properties. Then, we show how these models exhibit topological
condensation and confinement described by domain walls. To this end
we make use of closed and open ribbon operator algebras.
Sect.\ref{sect_conclusions} is devoted to conclusions.

Appendices deserve special attention since they contain the detailed
and basic explanations of all the constructions used throughout the
text. Specifically, Appendix \ref{appendix_A} contains a brief
summary of representation theory for group algebras, their centers
and induced characters. In Appendix \ref{appendix_B} we perform an
extensive treatment of ribbon operators, which are necessary to
describe the whole spectrum of the models. We define ribbons as
geometrical objects and then construct and characterize a series of
ribbon operator algebras. In Appendix \ref{appendix_C} we study the
relationship between certain ribbon transformations and the action
of ribbon operator algebras on suitable subspaces, which is a key
ingredient in describing the topological properties of the models.
In Appendix \ref{appendix_D} we give some details about the local
degrees of freedom that appear in the Hilbert space of two-particle
excitations. In Appendix \ref{apendice_excitaciones_solitarias} we
explain why single-quasiparticle states exist in non-abelian models
on surfaces of nontrivial topology. Finally, in Appendix
\ref{apendice_condensacion} we show several results needed for
condensation and ground state characterization.

\section{Non-Abelian Kitaev  Model}
\label{sect_II}

\subsection{Hamiltonian}

\begin{figure}
\includegraphics[width=7 cm]{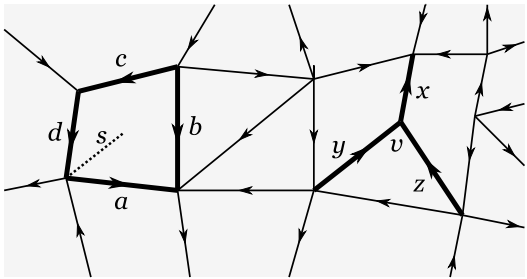}
\caption{
 The two dimensional lattice that we consider are arbitrary in shape and have oriented edges. Thick lines
 display the support of a face operator (left) and a vertex operator
 (right).
}\label{figura_red}
\end{figure}

The data necessary for building up the model, as introduced by
Kitaev\cite{kitaev}, are any given finite group $G$ and a lattice
embedded in an orientable surface. The edges of the lattice must be
oriented, as shown in Fig.~\ref{figura_red}. At every edge of the
lattice we place a qudit, that is, a $|G|$-dimensional quantum
system, with Hilbert space $\Hilb\prima_G$ with orthonormal basis
$\set {\ket g}{g \in G}$. This way, we identify $\Hilb\prima_G$ with
the group algebra $\C [G]$. The Hilbert space for the whole system
is then $\Hilb_G := {\Hilb_G\prima}^{\otimes n}$, with $n$ the
number of edges in the lattice. For notational convenience, we will
denote the inverse of elements of $G$ as $\bar g$ instead of the
usual $g\inv$. For completeness, we give a recollection of some
basic properties of the group algebra $\C[G]$ in Appendix
\ref{appendix_A}.

Usually, when we talk about sites in a lattice we mean its vertices.
However, here we will say that a site $s$ is a pair $s=(v,f)$ with
$f$ a face and $v$ one of its vertices\cite{kitaev}. The need to
consider sites will be clarified later, when we discuss the
excitations of the model in terms of strips associated to ribbon
operators. This is in contrast with the Abelian case where one only
needs to consider strings both in the direct and dual lattices. As
it happens, to obtain a non-Abelian generalization we need to
consider vertices and faces (plaquettes) in an unified manner
through the concepts of sites, and strings and dual strings in an
unified manner through the concept of ribbons.

The Hamiltonian of interest, as introduced in \cite{kitaev}, is
\begin{equation}\label{Hamiltoniano_kitaev}
H_G = - \sum_v A_v - \sum_f B_f,
\end{equation}
where the sums run over vertices $v$ and faces $f$. The terms $A_v$
and $B_f$ are projectors, called respectively vertex and face
operators, or electric and magnetic operators. They commute with
each other\eqref{conmutacion AB}. In what follows, we give their
explicit form.

First, we need a group of local operators at each vertex. We label
its elements as $A_v^g$, $g\in G$, with
$A_v^gA_v^{g\prime}=A_v^{gg\prime}$ so that they form a
representation of $G$ on $\Hilb_G$. The operators $A_v$ act only on
those edges that meet at $v$, and this action depends on the
orientation of the edge, inwards or outwards $v$. For example, for
the vertex $v$ of figure \eqref{figura_red} we have
\begin{equation}\label{definicion_vertex_op_g}
A_v^g \ket {x, y, z,\cdots} := \ket{gx, y\inver g, z\inver
g,\cdots},
\end{equation}
where the dots represent other qudits, which do not change. These
are the ``local gauge transformation''\cite{kitaev} operators. The
vertex operators $A_v$ that appear in the Hamiltonian are projectors
onto the trivial sector of the representation of $G$ at $v$, that is
\begin{equation}\label{definicion_vertex_op}
A_v := \frac 1 {|G|}\sum_{h\in G}A_v^h.
\end{equation}

Now let $s=(v,f)$ be a site and $p_s$ denote the closed path with
its endpoints in $v$ and running once and counterclockwise through
the border of $f$. That is, $p_s$ is related to an elementary
plaquette. We can then consider operators $B_s^g$, $g\in G$, that
project onto those states with value $g$ for the `product along
$p_s$'. For example, for the site $s$ of figure \eqref{figura_red}
we have
\begin{equation}\label{definicion_face_op_g}
B_s^g \ket {a, b, c, d, \dots} := \delta_{g,a\inver b cd}\ket {a, b,
c, d, \dots}.
\end{equation}
These are the ``magnetic charge''\cite{kitaev} operators. Note that
the orientation of the edges respect to the path is relevant. The
face operators $B_f$ that appear in the Hamiltonian are projectors
onto the trivial flux, that is
\begin{equation}\label{definicion_face_op}
B_f := B_s^1,
\end{equation}
where $s$ is any site with $s=(v,f)$ and $1$ is the unit of $G$. The
operator $B_f$ can be labeled just with the face, not with the
particular site, because if the flux is trivial for a site then it
is so for any other in the same face.

Since the Hamiltonian is a sum of projector operators, the ground
state subspace contains those states $\ket\xi$ which are left
invariant by the action of the vertex and face operators, namely,
\begin{equation}\label{propiedades_GS}
A_v \ket \xi = B_f \ket \xi = \ket \xi,
\end{equation}
for every $v$ and $f$. That is, the projector onto the ground state
is
\begin{equation}
P_{\text {GS}} = \prod_v A_v \prod_f B_f.
\end{equation}
In the sphere or the plane, there is no ground state
degeneracy\cite{kitaev}. In particular, the ground state can be
obtained easily
\begin{equation}\label{ground_state_G}
\ket{\psi_G} = P_{\text {GS}} \spc \vket 1 = \prod_v A_v \spc \vket
1,
\end{equation}
where $\vket 1$ is the state with all the qudits in the state $\ket
1$.

If an eigenstate violates some of the conditions
\eqref{propiedades_GS} it is an excited state. Note that there is an
energy gap from the ground state to excited states and that
excitations are localized. If $A_v\ket \xi = 0$, then we say that
there is an electric quasiparticle at vertex $v$. If $B_f\ket \xi =
0$, then we say that there is a magnetic quasiparticle at face $f$.
In general electric and magnetic charges are interrelated, as we
will see, and one says that quasiparticles are dyons that live at
sites.

The excitations of these models carry topological charge. Let us
explain what this means. First, consider a configuration with
several excitations, far apart from each other. Each of these
excitations has a type, a property that can be measured locally and
does not change\cite{kitaev}. It is this type what we refer as a
topological charge. The point is that there exist certain degrees of
freedom with a global, topological nature. In particular, there
exists a subsystem which depends on the value of the charges and
such that no local measurement is able to distinguish its
states\cite{kitaev}. This subsystem is thus protected and a good
place to store quantum information. When two quasiparticles get
close, some degrees of freedom of the protected subsystem become
local. This operation, called fusion, allows to perform
measurements. Finally, one can perform unitary operations on the
protected subsystem by suitably 'braiding' the excitations.

We will not be concerned with the particular rules that govern the
processes of fusion and braiding. Instead, we only want to be able
to label the topological charges. But for this, as we shall see, it
is enough to study certain ribbon operator algebras, which are
introduced next.

\subsection{Ribbon operators}
\label{seccion_ribbons}

This section is devoted to ribbon operators\cite{kitaev}, which will
be extensively employed throughout the paper. The main motivation is
that ribbon operators describe quasiparticle excitations above the
ground state in the Non-Abelian Kitaev model, much like string
operators describe the corresponding excitations in the Abelian
case. A full account of the properties and definitions for ribbon
operators used in this section is presented in Appendix
\ref{apendice_todo_ribbons}, specially in \ref{apendice_ribb_rho}
where  a basic characterization theorem for ribbon operators is
proven.

The basic idea behind ribbon
operators is the following. First, ribbons are certain `paths' that
connect sites (not vertices),
as shown in Fig.~\ref{figura_ribbons}. Suppose that for every
pair of sites $s$ and $s\prima$ and for every ribbon $\rho$
connecting them we have at our disposal certain family of operators
$\sset {O^i_\rho}_i$ with support in the ribbon $\rho$. In particular, suppose
that any state $\ket\psi$ with no excitations along $\rho$ except
possibly at $s$ and $s\prima$ can be written as
\begin{equation}\label{eliminar_excitaciones}
\ket\psi = \sum_i O^i_\rho \ket{\psi_i}
\end{equation}
in terms of certain states ${\ket{\psi_i}}$ which have no
excitations along $\rho$ except possibly at $s\prima$, but
\emph{not} at $s$. Then, any state can be obtained from states with
one excitation less by application of such ribbon operators. In the
sphere, where as we will see there are no states with one
excitation, this means that any configuration of excited sites can
be obtained from the GS by application of ribbon operators
connecting these sites. Thus, we are addressing a situation for
quasiparticle excitations which clearly resembles that of the
Abelian Kitaev model,
where strings in the dual and direct lattice have operators attached
to them that create excitations at their endpoints.

\begin{figure}
\includegraphics[width=8 cm]{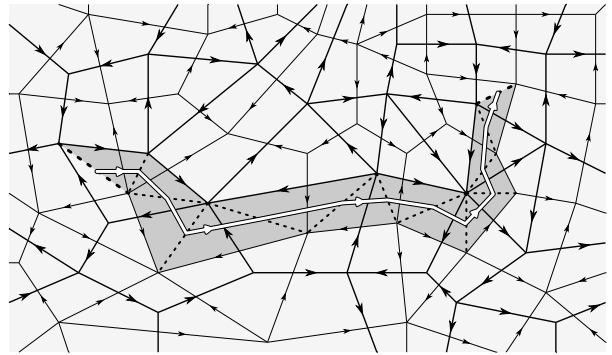}
\caption{
 Thick lines correspond to the lattice and thin lines to the dual lattice.
 Arrows show the orientation of edges and dual edges. Note that
 dual edges are oriented in agreement with edges (see explanation
in main text). The shaded area
 is a ribbon. All the sites that form the ribbon are displayed as
 dashed lines, thicker in the case of the two sites in the ends. The
 arrowed thick white line shows the orientation of the ribbon.
}\label{figura_ribbons}
\end{figure}

Before ribbons can be further considered, we need to give more
structure to our lattice. In particular, we will have to deal with a
`merged' lattice in which the lattice and its dual play a
simultaneous role. The reason to consider this merged lattice is
that the excitations, as commented above, are related to sites,
i.e.,  pairs $s=(v,f)$ of a vertex and a face. Since the dual of a
face is a vertex in the dual lattice, we could equally well say that
a site is a pair of a vertex $v$ and a neighboring dual vertex
$v\prima=f^\ast$. Thus, a site is best visualized as a line
connecting these two vertices, as the dashed lines shown in
Fig.~\ref{figura_ribbons}.

\begin{figure}
\includegraphics[width=6
 cm]{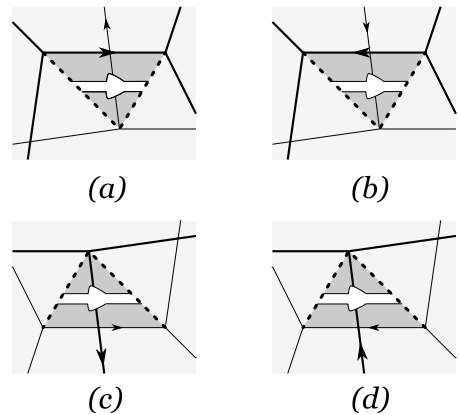}
\caption{
 Each figure represents a triangle $\tau$ (shaded area) that connects two
 sites (dashed lines): $\partial_0 \tau$ to the left and $\partial_1 \tau$ to the
 right. Thick lines correspond to the lattice and thin lines to the dual lattice.
 Arrows show the orientation of edges and dual edges. (a) A direct
 triangle with an edge which matches its direction. (b) A direct
 triangle with an edge which does not match its direction. (c) A
 dual triangle with a dual edge which matches its direction. (d) A
 dual triangle with a dual edge which does not match its direction.}
 \label{figura_triangulos}
\end{figure}

In order to have an oriented merged lattice, we orient the edges of
the dual lattice in such a way that a dual edge $e^\ast$ crosses the
edge $e$ `from right to left', as in Fig.~\ref{figura_ribbons}. This
can be done because we are considering orientable surfaces only.
Just as edges connect vertices in a normal lattice, we need
something that connects sites in the merged lattice. These
connectors turn out to be certain oriented triangles that come into
two types: direct and dual triangles. A direct triangle $\tau$ is
formed with two sites and an edge, as shown in
Fig.~\ref{figura_triangulos}(a,b). The idea is that $\tau$ points
from a site $\partial_0 \tau$ (dashed side to the left) to a site
$\partial_1\tau$ (dashed side to the right) through an edge $e_\tau$
in the direct lattice. Note that the directions of $\tau$ and
$e_\tau$ can either match or not, as the figure shows. A dual
triangle $\tau\prima$ is formed with two sites and a dual edge, see
Fig.~\ref{figura_triangulos}(c,d). Again, it points from a site
$\partial_0 \tau\prima$ to a site $\partial_1\tau\prima$ through an
edge $e_{\tau\prima}^\ast$, which now belongs to the dual lattice.
Again, the directions of $\tau\prima$ and $e_{\tau\prima}^\ast$ can
either match or not, as the figure shows.

Just as in a usual lattice a list of composable edges forms a path,
in the merged lattice a list of composable triangles forms a
triangle strip. So a strip is a sequence of triangles $\rho =
(\tau_1, \dots, \tau_n)$ with the end of a triangle being
the beginning of the next one, $\partial_1 \tau_i =
\partial_0 \tau_{i+1}$. The ends of a strip are $\partial_0 \rho=\partial_0 \tau_1$ and
$\partial_1 \rho=\partial_1 \tau_n$. A triangle strip is called a
ribbon when it does not self-overlap, except possibly on its ends. A
generic example of ribbon is shown in Fig.~\ref{figura_ribbons}. For
a detailed description of triangles, strips and ribbons on a
lattice, we refer to Appendix \ref{apendice_strips}.

Our next task is to attach to each triangle an algebra of operators
which is enough to move quasiparticles between its two ends, in the
sense of \eqref{eliminar_excitaciones}. With this aim in mind, we
first define triangle operators, which are single qudit operators
acting on the edge $e_\tau$ of a triangle $\tau$. these operators
depend on whether the triangle is direct or dual and on the relative
orientation of $e_\tau$. The four possibilities are depicted in
Fig.~\ref{figura_triangulos}. The corresponding operators are
\begin{alignat}{3}
\text{(a) }\,\,&T_\tau^g \ket k = \delta_{g,k}\ket k, \qquad
&\text{(b) }\,\,&T_\tau^g \ket k = \delta_{\inver g,k}\ket k,\\
\text{(c) } \,\, &L_\tau^g \ket k = \ket {gk}, \qquad &\text{(d) }
\,\, &L_\tau^g \ket k = \ket {k\inver g},
\end{alignat}
where $\ket k$ is the state of the qudit at the edge $e_\tau$.
Thus, the triangle operators $T_\tau^g$ of direct triangles are
projectors, like the $B_s^g$, and the triangle operators $L_\tau^g$
of dual triangles form a representation of $G$, like the $A_v^g$.

We start considering a direct triangle $\tau$. Since direct
triangles connect sites with the same face but different vertices,
triangle operators for direct triangles must be able to move
electric, or vertex, excitations.
Let $v,v\prima$ be the two vertices
of $\tau$.
 Then, as a special case of \eqref{transporte_electrico},
\begin{equation}\label{transporte_edge_directo}
|G|\sum_{g\in G} T_\tau^g A_v T_\tau^g = 1.
\end{equation}
Thus, any state $\ket\psi$ can be expressed as
\begin{equation}
\ket\psi=\sum_{g\in G} T_\tau^g \ket {\psi_g}
\end{equation}
with $\ket {\psi_g} = |G| A_v T_\tau^g \ket\psi$ an state with no
excitation at $v$ because $A_v$ projects out electric excitations.
Moreover, $T_\tau^g$ commutes with all face operators and all vertex
operators apart from those in the ends of $\tau$, so that $\ket
{\psi_g}$ has no excited spots which are not already in $\ket
{\psi}$,
except possibly at $v\prima$.
These are the properties we were looking for and thus we
define the algebra $\Alg_\tau$ as that with basis
$\sset{T_\tau^g}_{g\in G}$.

Next we consider a dual triangle $\tau$. Since dual triangles
connect sites with the same vertex but different face, triangle
operators for dual triangles must be able to move magnetic, or face,
excitations.
Let $f$, $f\prima$ be the two faces of $\tau$.
Then, as a
special case of \eqref{transporte_magnetico},
\begin{equation}\label{transporte_edge_dual}
\sum_{g\in G} L_\tau^{\inver g} B_f L_\tau^g = 1.
\end{equation}
Thus, any state $\ket\psi$ can be expressed as
\begin{equation}
\ket\psi=\sum_{g\in G} L_\tau^{\inver g} \ket {\psi_g}
\end{equation}
with $\ket {\psi_g} = B_f L_\tau^g \ket\psi$ an state with no
excitation at $f$, because $B_f$ projects out magnetic excitations.
Moreover, $L_\tau^g$ commutes with all vertex operators apart from
those in the only vertex of $\tau$ and all face operators except
those from the two faces connected by $\tau$, so that $\ket
{\psi_g}$ has no excited sites which are not already in $\ket
{\psi}$,
except possibly at $f\prima$.
These are the properties we were looking for and thus we
define the algebra $\Alg_\tau$ as that with basis
$\sset{L_\tau^g}_{g\in G}$.

Now that we have triangle operators at our disposal, we can move
quasiparticles at will, in the sense of
\eqref{eliminar_excitaciones}. In particular, if we want to move an
excitation from one end of a ribbon $\rho = (\tau_1, \dots, \tau_n)$
to the other end, we just proceed triangle by triangle. In other
words, we can introduce an algebra $\Alg_\rho:=\bigotimes_i
\Ribb_{\tau_i}$ which contains a family of operators
$\sset{O_\rho^i}$ with the properties related to
\eqref{eliminar_excitaciones}. $\Alg_\rho$ can be thought of as the
algebra of all quasiparticle processes along $\rho$. Note that it is
closed under the adjoint operator, $\Alg_\rho^\dagger=\Alg_\rho$.

However, if we are just interested in processes were no
quasiparticles are created or destroyed but in the ends of $\rho$,
as is the case for \eqref{eliminar_excitaciones}, then $\Alg_\rho$
is just too general. Instead, we consider the ribbon operator
algebra $\Ribb_\rho\subset\Alg_\rho$, which contains those operators
that do not create or destroy excitations along $\rho$. In other
words $F\in\Ribb_\rho$ if $[F,A_v]=[F,B_f]=0$ for any vertex $v$ and
face $f$ which do not lie in the ends of $\rho$. Note that
$\Ribb_\rho$ is closed under the adjoint operator because
$A_v=A_v^\dagger$, $B_f=B_f^\dagger$. These are the operators we
were searching for in \eqref{eliminar_excitaciones}: a basis of
$\Ribb_\rho$ gives the desired operators $O_\rho^i$,
see (\ref{transporte_electrico}, \ref{transporte_magnetico}).
$\Ribb_\rho$ can be thought of as the algebra of processes in which
a pair of quasiparticles is created in one end of the ribbon and
then one of them is moved to the other end. In these terms, it is
clear why excited states are expressible by means of ribbon
operators acting on ground states.

A particularly meaningful basis for $\Ribb_\rho$, explicitly given
in \eqref{base_ribb_particulas}, consists of certain operators
$F_\rho^{RC;\vect u\vect v}$, labeled by $C$, a conjugacy class of
the group $G$, $R$, an irreducible representation
of certain group $\Norm C$ defined below,
and the indices $\vect u = (i,j)$, $\vect v
= (i\prima, j\prima)$ with $i,i\prima =1,\dots,|C|$,
$j,j\prima=1,\dots,n_R$. Here $|C|$ is the cardinality of $C$ and
$n_R$ is the degree of the representation $R$. The group $\Norm C$
is defined as that with elements $g\in G$ with $gr_C=r_Cg$ for some
chosen representative $r_C\in C$. In order to construct the
operators $F_\rho^{RC;\vect u\vect v}$, one also has to choose a
particular unitary matrix representation $\Gamma_R$ for $R$ and
enumerate the elements of the conjugacy class as $C=\sset{c_i}$,
together with a suitable subset $\sset{q_i}_{i=1}^{|C|}\subset G$
such that $c_i=q_i r_C \inver q_i$. Later we will relate the labels
$R,C$ to the topological charges of the model and show how the
indices $\vect u,\vect v$ are related to local degrees of freedom at
both ends of the ribbon. We will use the following notation to
denote linear combinations of ribbon operators with the same
topological charge label $R,C$
\begin{equation}\label{op_creacion}
F_\rho^{RC}(\vect \alpha) := \sum_{\vect u,\vect v} \alpha^{\vect u,
\vect v} F_\rho^{RC;\vect u\vect v},
\end{equation}
where $\alpha^{\vect u\vect v}\in\C$.

In the case of abelian groups there are no local degrees of freedom
and the elements of the basis are $F_\rho^{RC}=F_\rho^{\chi,g}$ with
$g\in G$ and $\chi$ an element of the character group of $G$. These
operators are unitary and form a group:
\begin{equation}
     F_\rho^{\chi,g}F_\rho^{\chi\prima,g\prima}=F_\rho^{\chi\chi\prima,gg\prima},
     \qquad {F_\rho^{\chi,g}}^\dagger =F_\rho^{\comp \chi,\inver g}.
\end{equation}
Indeed, $T_\rho^\chi:=F_\rho^{\chi,1}$ are the string operators of
abelian models, and $L_\rho^g:=F_\rho^{e,g}$ the co-string
operators, with $e$ the identity character.

An essential property of ribbon operators, which reflects the
topological nature of the model, is that in the absence of
excitations the particular shape of the ribbon is unimportant: We
can deform the ribbon while keeping the action of the ribbon
operator invariant. More exactly, if the state $\ket\psi$ is such
that the ribbon $\rho$ can be deformed, with its ends fixed, to
obtain another ribbon $\rho\prima$ without crossing any excitation,
then
\begin{equation}
F_\rho^{RC}(\vect \alpha)\ket\psi =F_{\rho\prima}^{RC}(\vect
\alpha)\ket\psi.
\end{equation}
This is illustrated in Fig.~\ref{figura_deformacion}.

\begin{figure}
\includegraphics[width=8.3 cm]{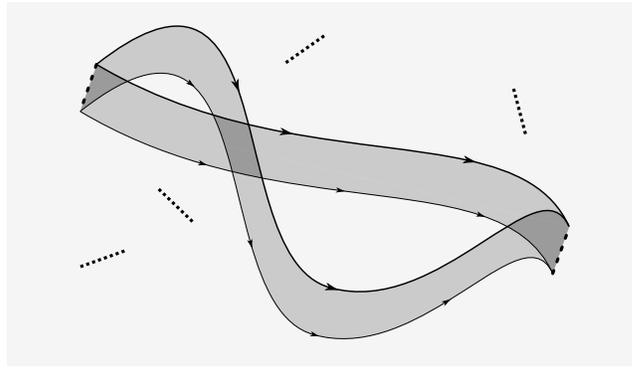}
\caption{
 An example of a deformation of a ribbon. The endpoints are fixed,
 and the area in between the two ribbons does not contain any
 excited site, which are represented with dotted lines.
}\label{figura_deformacion}
\end{figure}

\subsection{Closed ribbons}

For a closed ribbon $\sigma$ we mean one for which both ends
coincide, so that we can set $\partial \sigma := \partial_0
\sigma=\partial_1 \sigma$. In view of the definition of
$\Ribb_\rho$, in the case of closed ribbons it is natural to
consider a subalgebra $\Ribbclosed_\sigma\subset\Alg_\sigma$ such
that it forgets the single end $\partial \sigma$. With this goal in
mind, we let $\Ribbclosed_\sigma \subset \Alg_\sigma$ contain those
operators in $\Alg_\sigma$ that commute with all vertex and face
operators $A_v$, $B_f$. In terms of quasiparticle processes, such
closed ribbon operators are related to processes in which a pair of
quasiparticles is created and one end of them is moved along the
ribbon till they meet again to fuse into vacuum. Closed ribbon
operators play a fundamental role in characterizing the ground state
of the model in a similar fashion as how closed strings are the
building blocks for the ground state in the Abelian case (toric
code). A detailed analysis of closed ribbon operators is performed
in Appendix \ref{apendice_ribbclosed_sigma}.

We first consider the smallest examples of closed ribbons, i.e.,
dual and direct closed ribbons. We say that a ribbon is direct
(dual) if it consists only of direct (dual) triangles. A dual ribbon
like $\alpha$ in Fig.~\ref{figura_cerradas} encloses a single vertex
$v$, and $\Ribbclosed_{\alpha}$ has as basis the operators $A_v^h$,
$h\in G$. A direct ribbon like $\beta$ in Fig.~\ref{figura_cerradas}
encloses a single face $f$, and $\Ribbclosed_{\beta}$ has as basis
the operators $B_f^C$. These are labeled by the conjugacy classes
$C$ of $G$ and take the form $B_s^C=\sum_{g\in C} B_s^g$ for any
$s=(v,f)$. Thus, after defining ribbon operators by means of vertex
and face operators, we now see that vertex and face operators are
themselves ribbon operators.

\begin{figure}
\includegraphics[width=8.3 cm]{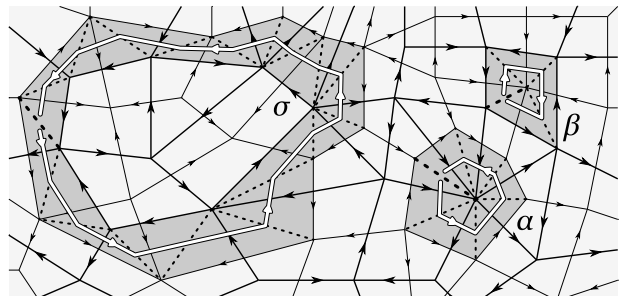}
\caption{Three examples of closed ribbons. $\sigma$ is a proper
closed ribbon, containing both dual and direct triangles. It is also
a boundary ribbon, as it encloses an area with the topology of a
disc. $\alpha$ is a dual closed ribbon and thus encloses a single
vertex. $\beta$ is a direct closed ribbon and thus encloses a single
face.}\label{figura_cerradas}
\end{figure}

As for the rest of closed ribbons $\sigma$, which we call proper
closed ribbons, it turns out that $\Ribbclosed_\sigma$ has as basis
certain orthogonal projectors $K_\sigma^{RC}$ that form a resolution
of the identity, as shown in proposition
\ref{prop_generadores_ribbclosed}. The labels $R,C$ of these
projectors are the same appearing in the basis for $\Ribb_\rho$. In
fact, in the next section we will characterize excitations in terms
of closed ribbon operators.

The algebra $\Ribbclosed_\sigma$ does not see the ends of $\sigma$.
Because of this, unlike $\Ribb_\sigma$, it can stand deformations in
which the end $\partial\sigma$ is not fixed or, for that matter,
rotations of the ribbon. More exactly, if the state $\ket\psi$ is
such that the closed ribbon $\sigma$ can be deformed to obtain
another ribbon $\sigma\prima$ without crossing any excitation then
\begin{equation}
K^{RC}_\sigma\ket\psi = K^{RC}_{\sigma\prima}\ket\psi,
\end{equation}
see appendix \ref{apendice_deformaciones_algebras}. This is
illustrated in Fig.~\ref{figura_deformacion_closed}. Another kind of
transformation is possible for closed ribbons. In particular, we can
consider deformations plus inversions of the orientation of the
ribbon, as shown in Fig.~\ref{figura_deformacion_closed}. When
$\sigma\prima$ is a transformation of $\sigma$ which includes an
inversion we have
\begin{equation}\label{inversion carga}
K^{\bar R^C \bar C}_\sigma\ket\psi = K^{RC}_{\sigma\prima}\ket\psi.
\end{equation}
where $\bar C$ is the inverse conjugacy class of $C$, $\bar R^C$ is
the conjugate representation of $R^C$ and $R^C$ is an irreducible
representation of $\Norm {\inver C}$ defined by $R^C(\cdot):=R(g
\cdot \inver g)$ if $\inver r_C = g r_{\inver C}\inver g$ for some
$g\in G$. In the next section we relate this to inversion of
topological charge.

\begin{figure}
\includegraphics[width=7.3 cm]{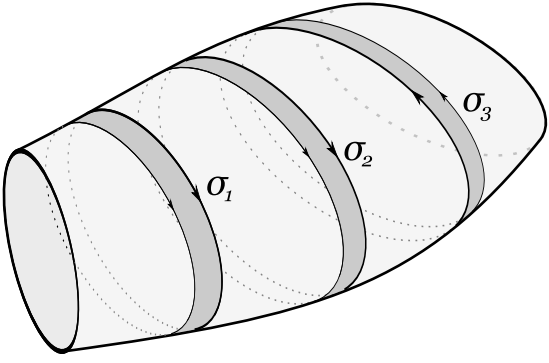}
\caption{
 Examples of closed ribbon transformations. A tubular piece of surface is displayed.
 The closed ribbon $\sigma_1$ is a deformation of $\sigma_2$ as long as
 there are no excitations between them. The ribbon $\sigma_3$ has an
 inverse orientation, and thus to obtain it from $\sigma_2$ we
 have to consider a deformation plus an inversion.
}\label{figura_deformacion_closed}
\end{figure}

\subsection{Topological charges}\label{seccion_kitaev_cargas}

Let $s_0$,$s_1$ be two non-adjacent sites in a lattice embedded in
the sphere. From the discussion on ribbon operators it follows that
the states
\begin{equation}\label{definicion_estado_RCiijj}
\ket {RC; \vect u\vect v} := F_\rho^{RC;\vect u\vect v} \ket
{\psi_G}
\end{equation}
form a basis for the subspace with excitations only at $s_0$ and
$s_1$. Here $\ket{\psi_G}$ is the ground state
\eqref{ground_state_G} and $\rho$ is any ribbon with $\partial_i
\rho = s_i$.

For each site $s=(v,f)$, we introduce the algebra $\mathcal D_s$
with basis $\sset{D_s^{hg}:=A_v^hB_s^g}_{h,g\in G}$. The reason to
introduce it is that its action on an excitation at $s$ gives all
possible local action on the excitation \cite{kitaev}. In other
words, $\mathcal D_s$ is useful to show why $\vect u, \vect v$ are
just local degrees of freedom. The action of the algebras $\mathcal
D_{s_i}$ on the states \eqref{definicion_estado_RCiijj} is
\begin{align} \label{accion_D_en_F}
D_{s_0}^{h,g} \ket {RC;\vect u\vect v} &= \delta_{g, c_i}
\sum_{s=1}^{n_R} \Gamma^{sj}_R(n(h q_i)) \spc\ket {RC;\vect u(s)\vect v}, \nonumber \\
D_{s_1}^{h,g} \ket {RC;\vect u\vect v} &= \delta_{g, \inver
c_{i\prima}} \sum_{s=1}^{n_R} \comp \Gamma^{sj\prima}_R (n(h
q_{i\prima})) \spc \ket {RC;\vect u\vect v(s)},
\end{align}
where $\vect u = (i,j)$, $\vect v = (i\prima,j\prima)$, $\vect u(s)
= (i(h q_i),s)$, $\vect v(s) = (i(h q_{i\prima}),s)$ and we set for
any $g\in G$ $g=: q_{i(g)} n(g)$ with $n(g)\in\Norm C$. Equations \eqref{accion_D_en_F} are a consequence of (\ref{conmutacion_D_F}, \ref{propiedades_GS}).

As shown in detail in appendix \ref{appendix_D}, it is possible to
find operators $d^{\vect u\prima}_{\vect u\primas}\in \D_{s_0}$ and
$d^{\vect v\prima}_{\vect v\primas}\in \D_{s_1}$ with
\begin{equation}\label{cambio_arbitrario}
d^{\vect u\prima}_{\vect u\primas} \spc  d^{\vect v\prima}_{\vect
v\primas} \spc \ket {RC; \vect u\vect v} = \delta_{\vect u, \vect
u\primas}\delta_{\vect v, \vect v\primas} \ket {RC; \vect
u\prima\vect v\prima}.
\end{equation}
Thus we see that a state with particular labels $\vect u$, $\vect v$
can be transformed with local operators into one with any other
labels $\vect u\prima$, $\vect v\prima$. Roughly speaking, for local
operators we mean operators which act on a neighborhood of the
excitations. More exactly, local operators should have a support
which does not connect excitations.

\begin{figure}
\includegraphics[width=7.3 cm]{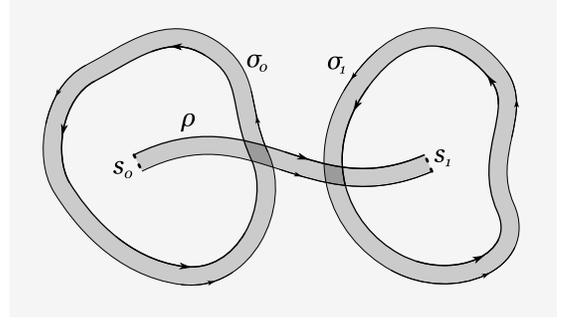}
\caption{An open ribbon $\rho$ that connects two sites $s_0$ and
$s_1$ and two closed ribbons $\sigma_0$ and $\sigma_1$ that surround
counterclockwise $s_0$ and $s_1$, respectively. The ribbon operators
$F^{h,g}_\rho$ of the open ribbon change the excitations at $s_0$,
$s_1$. The ribbon operators $K_{\sigma_0}^{RC}$, $K_{\sigma_1}^{RC}$
of the closed ribbons project the system onto states with a given
topological charge at $s_0$, $s_1$.}\label{figura_detectar_carga1}
\end{figure}

What about the degrees of freedom related to the labels $R$ and $C$?
They can certainly be measured locally, because there exists a set
of projectors $D^{RC}_{s_0}\in\D_{s_0}$ with
\begin{equation}\label{medir_carga_2particle}
D^{R C}_{s_0} \spc \ket {R\prima C\prima; \vect u\vect v} =
\delta_{R,R\prima}\delta_{C,C\prima} \ket {R C; \vect u\vect v}.
\end{equation}
However, $R$ and $C$ cannot be changed locally, in the sense that an
operator with a support not connecting both sites and which creates
no additional excitations will not change their values. To see this,
consider two closed ribbons $\sigma_0$ and $\sigma_1$ that enclose
respectively the sites $s_0$ and $s_1$ counterclockwise, as in
Fig.~\ref{figura_detectar_carga1}.
From the discussion in appendix \ref{apendice_charge_types} it follows that
\begin{equation}\label{medir_cargas}
K_{\sigma_0}^{RC} \spc \ket {R\prima C\prima; \vect u\vect v} =
K_{\sigma_1}^{\inver R^C \inver C} \spc \ket {R\prima C\prima; \vect
u\vect v} = \delta_{R,R\prima}\delta_{C,C\prima}\ket {R C; \vect
u\vect v}.
\end{equation}
Any operator with no common support with $\sigma$ will commute with
the projectors $K_\sigma^{RC}$, and thus cannot change the value of
$R$ and $C$. In particular, any operator which changes $R$ and $C$
must have a support that connects the sites $s_0$ and $s_1$.

Indeed, the preceding discussion shows that $R$ and $C$ are the
labels of the topological charges of the model. Thus the charge of
an excitation is the pair $(R,C)$, with $C$ a conjugacy class of $G$
and $R$ an irreducible representation of $\Norm C$. If a closed
ribbon $\sigma$ encloses certain amount of excitations, as in
Fig.~\ref{figura_detectar_carga2}, the projectors $K_\sigma^{RC}$
correspond to sectors with different total topological charge in the
region surrounded. If $\ket\xi$ is a state with no excitations in
the area enclosed by $\sigma$, we have
\begin{equation}\label{carga_unidad}
K^{e \spc 1}_\sigma\ket\xi=\ket\xi,
\end{equation}
with $e$ the identity representation, see appendix
\ref{apendice_condensacion}. Thus, $(e,1)$ is the trivial charge.
This offers a way to describe the ground state of
\eqref{Hamiltoniano_kitaev} as the space of states for which
\eqref{carga_unidad} holds for any boundary ribbon, that is, any
closed ribbon enclosing a disc or simply connected region.

\begin{figure}
\includegraphics[width=6.3 cm]{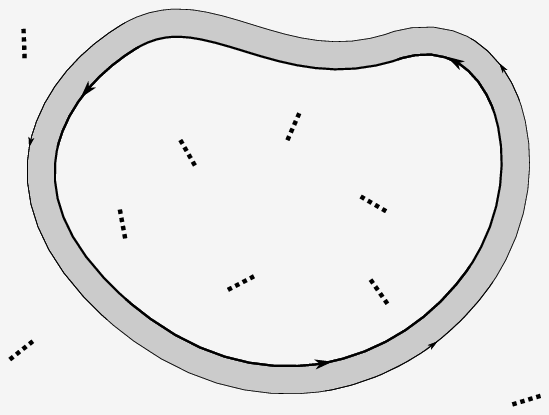}
\caption{
 A boundary ribbon $\sigma$ that encloses several excitations counterclockwise. The
 corresponding operators $K_\sigma^{RC}$ are projectors onto the
 sector with total topological charge $(R,C)$ inside the ribbon.
}\label{figura_detectar_carga2}
\end{figure}

In a region with no excitations, quasiparticles can only be locally
created in pairs, so that the two excitations have opposite charges
and the total charge in the region remains trivial. From
\eqref{inversion carga} or \eqref{medir_cargas} it follows that the
opposite of the charge $(R,C)$ is $(\bar R^C, \bar C)$,

\subsection{Single-quasiparticle states}

In a sphere there do not exist states with a single excitation. The
reason, as shown in Fig.~\ref{figura_esfera}, is that any closed
ribbon $\sigma$ divides the sphere in two regions, both of them
simply connected. The ribbon $\sigma$ surrounds one of this region
counterclockwise, call it $R_1$, and the other one clockwise, call
it $R_2$. Then the operator $K_\sigma^{RC}$ is a projector onto the
subspace with total charge $(R,C)$ in $R_1$, but also a projector
onto the subspace with total charge $(\bar R^C, \bar C)$ in $R_2$.
Thus, if there are no excitations in $R_1$, we have a total charge
$(e,1)$ in $R_1$ and also a total charge $(e,1)$ in $R_2$. But a
single excited site cannot have trivial charge, and thus $R_2$
contains either zero or more than one excitation

What about surfaces with non-trivial topology, such as a torus? In
the case of Abelian groups, the situation is the same as in the
sphere: there are no states with a single excitation. In the case of
vertex excitations, that is, electric charges, this follows from the
fact that
\begin{equation}
\prod_{v\in V} A_v^g = 1.
\end{equation}
For face excitations, that is, magnetic charges, an analogous result
holds. For any character $\chi$ of $G$, let $B_f^\chi := \sum_{g\in
G} \chi(g) B_s^g$ for $s=(v,f)$. Then
\begin{equation}
\prod_{f\in F} B_f^\chi = 1.
\end{equation}
For non-Abelian groups, the situation is very different. In fact,
examples of single-quasiparticle states can be constructed, see
appendix \eqref{apendice_excitaciones_solitarias}.

\begin{figure}
\includegraphics[width=4 cm]{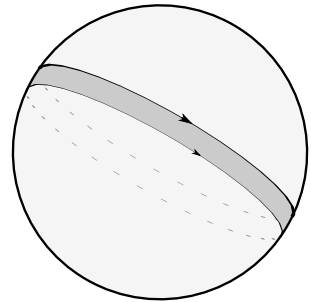}
\caption{A closed ribbon in a sphere. Its ribbon operators
$K_\sigma^{RC}$ project onto states with topological charge $(R,C)$
in the upper side of the sphere and $(\bar R^C, \bar C)$ in the lower
side.} \label{figura_esfera}
\end{figure}

\section{Condensation and confinement}
\label{sect_III}

\subsection{The models}

We want to modify the Hamiltonian $H_G$ by introducing single qudit
terms. In particular, we propose to consider projectors of the form
\begin{equation}\label{definicion_edge_op}
L_\tau^N := \frac 1 {|N|} \sum_{n\in N} L_\tau^n, \qquad T_\tau^M:=
\sum_{m\in M} T_\tau^m,
\end{equation}
where $\tau$ is a dual or direct triangle and $N$, $M$ are subgroups
of $G$. Thus $L_\tau^N$ projects out the trivial representation of
$N$ and $T_\tau^M$ selects those states within $M$. We want to have
single qudit operators that do not depend on the orientation of the
edge $e=e_\tau$. This is automatic for $T_e^M:=T_\tau^M$, but in the
case of dual triangles this is true if and only if $N$ is normal, so
that we can set $L_e^N:=L_\tau^N$. That is, if $n\in N$ and $g\in
G$, then $gn\inver g\in N$. Moreover, we want these two kinds of
single-qudit terms to commute
\begin{equation}
[L_e^N,T_e^M]=0,
\end{equation}
which is true if and only if $N\subset M$.

Now consider a Hamiltonian of the form
\begin{equation}\label{Hamiltoniano_intermedio}
H=H_G-\mu \sum_e \left ( L_e^N + T_e^M \right )
\end{equation}
where $\mu$ is a positive coupling constant and the sum runs over
edges $e$. The problem with this Hamiltonian is that the new terms
do not commute with $H_G$. However, as we show now, we can still
consider the limit of large $\mu$. In this limit, the low energy
sector is projected out by
\begin{equation}\label{proyector_sin_paredes}
P := \bigotimes_e T^M_e L^N_e.
\end{equation}
Let us define the following vertex and face projectors
\begin{equation}\label{definicion_vertex_face_op_NM}
A_v^M := \frac 1 {|M|} \sum_{m\in M} A_v^m, \qquad B_f^N:=B_s^N:=
\sum_{n\in N} B_s^n,
\end{equation}
where $s=(v,f)$ is a site. Note that $B_s^N$ only depends on $f$
 because $N$ is normal. We now make the
following observation
\begin{align}
|M| P \spc A_v^M \spc P &= |G| P \spc A_v \spc
P, \nonumber\\%
P \spc B_f^N \spc P &= |N|\spc P \spc B_f \spc P.
\end{align}
Thus, studying the low energy sector of
\eqref{Hamiltoniano_intermedio} for large $\mu$ amounts to study the
sector with no edge excitations of the Hamiltonian
\begin{equation}\label{Hamiltoniano_NM}
H_G^{N,M}:=-\sum_v A^M_v - \sum_f B^N_f - \sum_e \left (T_e^M +
L_e^N \right).
\end{equation}
The point of these Hamiltonians is that all its vertex, face and
edge terms commute and thus the ground state of the system can be
exactly given. It turns out that it is related to that of
\eqref{Hamiltoniano_kitaev} but for the group $G\prima:= M/N$, as we
will see in the next section. Note that $H_G^{1,G}$ is just the
original Hamiltonian \eqref{Hamiltoniano_kitaev}, up to a constant.
Although we have motivated the introduction of
\eqref{Hamiltoniano_NM} through \eqref{Hamiltoniano_intermedio}, our
aim is to study the models $H_G^{NM}$ in their own right, for
arbitrary subgroups $N\subset M\subset G$ with $N$ normal.

\subsection{Ground state}\label{seccion_kitaev_NM_GS}

The ground state of Hamiltonian \eqref{Hamiltoniano_NM} is described by the conditions
\begin{equation}\label{condiciones_GS_locales}
A_v^M \ket{\psi} = B_f^N \ket{\psi} = L_e^N \ket{\psi} = T_e^M \ket{\psi} =\ket{\psi}
\end{equation}
where $v$ is any vertex, $f$ any face  and $e$ any edge. Violations of these conditions amount to vertex, face or edge excitations.
Let $V$ be the subspace of states with no edge excitations, which is projected out by the
projector $P$ of \eqref{proyector_sin_paredes}. $V$ is a tensor
product of single qudit subspaces $V := \bigotimes_e V_{M/N}$, with
$V_{M/N}\subset\Hilb\prima_G$ the subspace with orthonormal basis:
\begin{equation}\label{base_MN}
\ket {\tilde m} := {|N|}^{-\frac 1 2} \sum_{n\in N} \ket {mn}, \qquad
{\tilde m\in M/N}.
\end{equation}
Thus $V\simeq\Hilb_{M/N}$, that is, within the subspace $V$ we are
effectively dealing with qudits of dimension $|M/N|$ which are
naturally labeled through the group quotient. We denote the
corresponding isomorphism by
\begin{equation}
\funcion p{\Hilb_{M/N}}{V}.
\end{equation}

Let us write
\begin{equation}\label{isomorfismo_Hamiltonianos}
H\prima_{M/N}:= p \ H_{M/N} \ p\inv,
\end{equation}
that is, $H_{M/N}\prima$ is the Hamiltonian
\eqref{Hamiltoniano_kitaev}, for the group $M/N$, applied to the
subspace $V_{M/N}$. We have
\begin{equation}\label{relacion_Hamiltonianos}
H\prima_{M/N}P = \left(H_G^{N,M} + 2|E|\right)P.
\end{equation}
Thus, within the sector with no edge excitations we are effectively
 dealing with the Hamiltonian $H_{M/N}$
\eqref{Hamiltoniano_kitaev}. Moreover, the ground state of
$H_G^{N,M}$ in $\Hilb_G$ is that of $H\prima_{M/N}$ in $V$. The
projector onto the ground state is
\begin{equation}
P_{\text {GS}}^{N,M} := P \prod_v A_v^M \prod_f B_f^N = P \prod_v
A_v\prima \prod_f B_f\prima
\end{equation}
where $A_v\prima := p A_v p\inv$, $B_f\prima := p B_f p\inv$ with
$A_v$ and $B_f$ acting in $\Hilb_{M/N}$. In the sphere, the
normalized ground state is
\begin{equation}\label{ground_state_N_M}
\ket{\psi_G^{N,M}} \propto P_{\text {GS}}^{N,M} \spc \vket 1 \propto
\prod_v A\prima_v \spc \vket {\tilde  1}.
\end{equation}

Thus the new edge terms in the Hamiltonian, which can be thought of
as a sort of generalized `Zeeman terms'\cite{abelian_confinement},
have the role of selecting a particular sector of the Hilbert space
in which a new non-Abelian discrete gauge symmetry appears, namely
$G\prima = M/N$. Thus, edge terms amount to an explicit symmetry
breaking mechanism, since in general the gauge symmetry is reduced,
even to a trivial one if $M=N$. Alternatively, we can say that they
provide a symmetry-reduction mechanism. Thus, the sector with no
edge excitations is completely understood. In the remaining sections
we study the meaning of edge excitations.

\subsection{An example}\label{seccion_ejemplo}

Before we go on with the general case and its details, let us first
give a flavor of what is going on by considering a family of
examples. We take $N=1$ and $M$ normal in $G$, so that the new gauge
group is $G\prima = M$.
Note that in this case we can forget about the $L_e^N$ terms because $L_e^1=1$.
Our aim is to study the result of applying quasiparticle creation operators \eqref{op_creacion} on a ground
state of the Hamiltonian \eqref{Hamiltoniano_NM}:
\begin{equation}\label{estado_ribbon_NM}
 \ket{\psi}:=F_\rho^{RC}(\vect \alpha) \ket{\psi_G^{NM}}.
\end{equation}

We first consider purely magnetic quasiparticle creation operators, fixing $R$ as the
identity representation. For simplicity we set $\alpha^{\vect u\vect v}=c\in \C$. Then from (\ref{conmutacion_D_F}) it follows that $[A_v^M,F^{RC}(\vect\alpha)]=0$ for every vertex $v$ and from \eqref{conmutacion_Ln_Tm_F_inside} it follows that $[T_e^M,F^{RC}(\vect\alpha)]=0$ for any edge $e$ not in a dual triangle of $\rho$.
Then due to \eqref{condiciones_GS_locales} the state $\ket\psi$ can have,
at most, face excitations on the ends of $\rho$ and edge
excitations on dual triangles of $\rho$.
In particular, from (\ref{conmutacion_D_F},  \ref{conmutacion_Ln_Tm_F_proyeccion},\ref{condiciones_GS_locales}) it follows that
if $f$ is an
end face of $\rho$ and $e$ is \textit{any} dual edge of $\rho$ we
have
\begin{equation}
 B_f^1 \ket \psi = u \ket \psi, \qquad T_e^M \ket \psi =u\prima \ket \psi,
\end{equation}
with $u,u\prima=0,1$. As long as $C\neq 1$, we have $u=0$ and thus
$\ket\psi$ contains a pair of face excitations. On the other
hand, $u\prima=1$ iff $C\subset M$, which means that $\ket\psi$
contains a chain of edge excitations along $\rho$ if we try to
create magnetic charges which do not belong to the new gauge group
$G\prima$.
%%%
Therefore, we find out that some face excitations are confined, in particular those created with $C\not\subset M$.
%%%
By this, we mean that the energy of $\ket\psi$ increases linearly with the length of $\rho$ in terms of dual triangles.

Next, we consider  purely electric quasiparticle creation operators, that is, we set $C=1$.
%%%
Reasoning in the same way as in the previous case, one finds out that the state $\ket\psi$ can have, at most, vertex excitations on the ends of $\rho$, but no face or edge excitations.
%%%
In particular, if $v$ is an end
vertex of $\rho$ we have
\begin{equation}
 A_v^M \ket \psi = u \ket \psi,
\end{equation}
with $u=1$ if the restriction of $R$ to $M$ is an identity
representation and $u=0$ otherwise.
%%% beg
That is, in some cases
$\ket\psi$ is a ground state although $R$ is not trivial. Since there is no local degeneracy in the ground state we know that $\ket\psi=c\ket{\psi_G^{NM}}$ for some $c\in\C$. Moreover, $c$ can be nonzero, because, as we will see below, for $\alpha^{\vect u\vect v}=\delta_{\vect u,\vect v}$ and $R$ trivial in $M$
\begin{equation}\label{condensacionEj1}
     \xpctd{F_\rho^{R1}(\vect \alpha)}_{\psi_G^{NM}} = 1.
\end{equation}
%%% end
Moreover, if $\sigma$ is any
boundary ribbon we have for $R$ trivial in $M$
\begin{equation}\label{condensacionEj2}
     \xpctd{K_\sigma^{R1}}_{\psi_G^{NM}} = \frac{n_R\spc |M|}{|G|}.
\end{equation}
Thus, those electric charges with trivial restriction of $R$ to $M$
are condensed: they are part of the ground state.

\subsection{Condensation} \label{subseccion_condensacion}

Let $\sigma$ be a boundary ribbon, that is, a ribbon that encloses
some region $r$. Motivated by the previous example, we want to study
the expectation value in the ground state of $H_G^{NM}$ of the
operators $K_\sigma^{RC}$. Recall that these operators project onto
the space with total topological charge $(R,C)$ in systems with
Hamiltonian $H_G$. Then if for a particular charge type we have
\begin{equation}
    \xpctd {K_\sigma^{RC}} := \bra{\psi_G^{N,M}}
    K_\sigma^{RC}\ket{\psi_G^{N,M}}> 0
\end{equation}
we say that the charges $(R,C)$ of the original Hamiltonian $H_G$
get condensed in the system with Hamiltonian $H^{NM}_G$: if one
measures the charge of the region $r$ in the ground state of
$H_G^{NM}$ there exists some probability of finding the charge
$(R,C)$.

%%% beg
As we show in appendix \ref{apendice_condensacion}
\begin{equation}\label{condensacion}
\xpctd {K_\sigma^{RC}} = \frac{n_R\spc |M|}{|G||N|}\spc
(\chi_R,\chi_{e_M\!\uparrow})_{\Norm C} \spc |C\cap N|,
\end{equation}
where the product $(\cdot, \cdot)_{M}$ is defined in
\eqref{producto_funciones_clase} and ${e_M\!\uparrow}$ is the induced representation in $G$ of the
identity representation of $M$, see appendix \ref{apendice_induced_representations}.
%%% end
Another way to write the product is
\begin{equation}\label{condensacion_auxiliar}
(\chi_R,\chi_{e_M\!\uparrow})_{\Norm C} = \frac {1}{|\Norm
C||M|}\sum_{g\in G} |M_C^g| (\chi_R,1)_{M_C^g},
\end{equation}
where $M_C^g := \Norm C\cap \inver g M g$. Note in particular that
for $M$ normal the sum has a single term, simplifying the form of
\eqref{condensacion}.  The result \eqref{condensacion} not only
shows that some of the charges are condensed, but also that the
expectation value is independent of the shape or size of the ribbon,
a feature that underlines the topological nature of the
condensation. Such behavior for a perimeter expectation is called a
zero law\cite{hastings_wen05}.

Let us consider several examples. First, under
%%%
Kitaev's original Hamiltonian $H_G^{1,G}$ we have
\begin{equation}
    \xpctd {K_\sigma^{RC}} =  \delta_{C,1}, \delta_{R,e_G}
\end{equation}
where $e_G$ is the identity representation of $G$.
%%%
Thus none of the nontrivial charges is condensed, as expected. In the case
$N=M=G$ we have
\begin{equation}
    \xpctd {K_\sigma^{RC}} = \frac{|C|}{|G|}\
    \delta_{R,e_{{\Norm C}}},
\end{equation}
which means that the purely magnetic charges are condensed. On the
contrary, in the case $N=M=1$ we have
\begin{equation}
    \xpctd {K_\sigma^{RC}} = \frac{n_R^2}{|G|}\
    \delta_{C,1}
\end{equation}
which means that the purely electric charges are condensed. Another
illustrative case is that of an Abelian group $G$. In that case we
can label the projectors as $K_\sigma^{\chi,g}$ with $g\in G$ and
$\chi$ is an element of the character group of $G$. Then
\begin{equation}\label{condensacion_abeliano}
    \xpctd {K_\sigma^{\chi,g}} = \frac{|M|}{|G|\spc|N|}\
    \delta_{gN,N} \ \delta_{\chi_M, e_M}
\end{equation}
where $\chi_M$ is the restriction of $\chi$ to $M$.

It is possible to show which charges condense using another kind of
expectation values, namely those for operators \eqref{op_creacion},
which create a particle-antiparticle pair in the original model
\eqref{Hamiltoniano_kitaev}.
%%%
From the discussion in appendix \ref{apendice_condensacion} it follows that
%%%
if $|C\cap N|=\emptyset$ or
$(\chi_R,1)_{M_C^ g}=0$ for all $g$, we have
\begin{equation}\label{condensacion2auxiliar}
     \xpctd{F_\rho^{RC}(\vect \alpha)}=0.
\end{equation}
In the case of $M$ normal, for a charge $(R,C)$ that is not
condensed according to \eqref{condensacion} the operator
$F_\rho^{RC}(\vect \alpha)$ always has expectation value zero. Even
if $M$ is not normal, if we trace out the local degrees of freedom
and set $F_\rho^{RC} := F_\rho^{RC}(\vect \alpha)$ with
$\alpha^{\vect u\vect v}=
\delta_{\vect u,\vect v}$ we get
\begin{equation}\label{condensacion2}
     \xpctd{F_\rho^{RC}}=\frac 1 {|C|}\spc(\chi_R,\chi_{e_M\!\uparrow})_{\Norm C}\spc|C\cap N|.
\end{equation}
Therefore, for that particular choice we get an expectation value
which vanishes if an only if \eqref{condensacion} does, showing that
both approaches agree. Again topology makes its appearance in the
fact that the length or shape of the ribbon $\rho$ are not relevant.
For Abelian groups \eqref{condensacion2} reads
\begin{equation}
     \xpctd{F_\rho^{\chi,g}}=\delta_{gN,N}\delta_{\chi_M,e_M}.
\end{equation}

\subsection{Confinement}

The example studied in \eqref{seccion_ejemplo} suggests that the
edge terms $L_e^N$ and $T_e^M$ could be interpreted as string
tension terms, which in turn would confine some of the charges of
the original model $H_G$. However, one has to be a bit cautious with
such a viewpoint in general. Certainly, in those cases in which $N$
is central in $M$ and $M$ is normal (from now on, case I) such a
viewpoint makes sense. Only in those cases do certain properties
hold, see appendix \ref{apendice_edge_operators}. In case I we can
write the relations,
%%%
see \eqref{conmutacion_Ln_Tm_F_proyeccion},
%%%
\begin{equation}\label{string_tension}
P_{e,e\prima}^{NM}\spc F_\rho^{RC;\vect u\vect v} \spc
P_{e,e\prima}^{NM}= d_{RC}^{NM} \spc F_\rho^{RC;\vect u\vect v} \spc
P_{e,e\prima}^{NM}
\end{equation}
where $d_{RC}^{NM}$ equals one (zero) if $C\subset M$ and the
restriction of $R$ to $N\subset N_C$ is trivial (in other case),
$P_{e,e\prima}^{NM}=L_e^NT_{e\prima}^M$ and
$e=e_\tau,e\prima=e_{\tau\prima}$ with $\tau$, $\tau\prima$ direct
and dual triangles in an open ribbon $\rho$, respectively. The
relation \eqref{string_tension} show that edge operators project out
certain states among those which were created by applying a string
operator to a ground state. Note that the projection only takes into
account the quasiparticle labels $R$, $C$ of the string operators.
Moreover, it is not important which the particular edges are.
Outside of case I such nice properties, reasonable for string
tension terms, do not hold. As a consequence of
\eqref{string_tension}, we now that a state of the form
\eqref{estado_ribbon_NM} will have a chain of edge excitations along
$\rho$ unless $C\subset M$ and $R$ is trivial in $N$: all other
charges get confined when moving from $H_G$ to $H_G^{NM}$, which
means that they exist at the end of chains of edge excitations. We
shall refer to these chains of excitations as domain walls. They can
be labeled just as we labeled topological charges in $H_G$,
something that we will do in the next section.

Unfortunately, as soon as any of the mentioned conditions for case I
fails many nice properties of the models are lost. Indeed, only for
those systems that fall in that class will we be able to classify
domain walls and confined charges in terms of open an closed ribbon
operator algebras, in the fashion of what we already did for
topological charges in $H_G$. On the other hand, in certain more
general cases it is still possible to classify domain wall types. In
particular, we will show that this can be done in all models with
$N$ abelian (from now on, the case II). Interestingly enough, the
domain wall fluxes in case II are qualitatively richer than in case
I. Such fluxes have in general a non-abelian character for case II
systems, whereas for case I they have always an abelian nature.

\subsection{Case I systems}

This section is devoted to those models $H_G^{N,M}$ with $N$ central
in $M$ and $M$ normal in $G$. We will show that edge excitations
appear in the form of domain walls that terminate in certain site
excitations which are therefore confined. With this goal in mind, we
start introducing the excitations which will turn out to be
confined. Consider the projectors
\begin{equation}\label{proyectores_confinadas}
A_v^N := \frac 1 {|N|} \sum_{n\in N} A_v^n, \qquad B_f^M := B_s^M=
\sum_{m\in M} B_{s}^{m}.
\end{equation}
They commute among each other and with the terms of Hamiltonian
\eqref{Hamiltoniano_NM}, so that we could choose the energy
eigenstates to be eigenstates of the projectors
\eqref{proyectores_confinadas}. We say that the state $\ket\psi$ has
a confined excitation at a site $s=(v,f)$ whenever
$A_v^NB_f^M\ket\psi = 0$. That these are really excitations follows
from
\begin{equation}\label{operadores_A_B_confinadas}
A_v^NB_f^M\ket\psi = 0 \quad \Longrightarrow\quad A_v^MB_f^N\ket\psi
= 0,
\end{equation}
where we have used $A_v^NA_v^M=A_v^M$ and $B_f^MB_f^N=B_f^N$. That
they are really confined will be revealed later, but we can already
give a clue: a state with a confined excitation at the site $s$ must
have an edge excitation at least at one of the edges $e$ meeting at
$v$ or in the border of $f$. Thus, confined excitations cannot
appear isolated: there must be edge excitations around them.
Conversely, it is also true that a chain of edge excitations cannot
terminate without a confined excitation in its end. However, this is
not enough to demonstrate confinement, but we can do it better by
introducing suitable ribbon operator algebras.

\subsubsection{Open ribbon operators.}

Domain walls have a type, and this type behaves as a flux in the
absence of confined excitations. We thus need a family of projectors
that distinguish between the domain wall fluxes that cross a
particular line, analogous to the projectors $K_\sigma^{RC}$ that
distinguished the topological charge in a the area enclosed by
$\sigma$. So we consider as our starting point the ribbon operator
algebra $\Ribb_\rho$ for an open ribbon $\rho$. Since the flux
should be the same if we move the ends of our flux-measuring ribbon
in an area with no edge excitations, we need a ribbon algebra which
to some extent forgets the ends of the ribbon. In particular, we
should choose ribbon operators that do not create or destroy
excitations. Thus, we define $\Ribbopen_\rho \subset \Ribb_\rho$ as
the subalgebra containing those operators $F\in\Ribb_\rho$ which
commute with all vertex operators $A_v^M$ and face operators
$B_f^N$.  Such operators also commute with all edge terms $L_e^N$
and $T_e^M$,
%%%
see colollary \ref{cor_ribbopen_LT}.
%%%

%%%
As we show in proposition \ref{prop_generadores_ribbopen}, $\Ribbopen_\rho$ is linearly generated by certain
orthogonal projectors $J_\rho^{\chi t}$ that form a resolution of the identity.
%%%
The labels $(\chi,t)$ of these projectors are $\chi$, an element of
the character group of $N$, and $t$, an element of the quotient
group $G/M$. If $\ket\psi$ has no edge excitations along $\rho$ then
\begin{equation}\label{domain_wall_trivial}
J_\rho^{e1} \ket\psi =\ket\psi,
\end{equation}
where $e$ and $1$ are both identity elements,
%%%
see (\ref{GS_open_ribbon}, \ref{bases_ribbopen_abeliano}).
%%%

The algebra $\Ribbopen_\rho$ can stand deformations in which the
ends of $\rho$ are not fixed. More exactly, if the state $\ket\psi$
is such that the open ribbon $\rho$ can be deformed,
%%%
without fixing its ends,
%%%
to obtain another ribbon $\rho\prima$ in such a way that no confined
excitations are crossed and the ends do not touch edge excitations,
then
\begin{equation}\label{domain_wall_deformacion}
J^{\chi t}_\rho\ket\psi = J^{\chi t}_{\rho\prima}\ket\psi,
\end{equation}
see appendix \ref{apendice_deformaciones_algebras}. This is
illustrated in Fig.~\ref{figura_deformacion_open}. As in the case of
closed ribbons in $H_G$, here we can consider inversions in the
orientation of the ribbon, as shown in
Fig.~\ref{figura_deformacion_open}. When $\rho\prima$ is a
transformation of $\rho$ which includes an inversion we have
\begin{equation}\label{domain_wall_inversion}
J^{\chi t}_\rho\ket\psi = J^{\inver \chi^t \inver
t}_{\rho\prima}\ket\psi,
\end{equation}
where for any $n\in N$ we set $\chi^t(n):=\chi(tn\inver t)$,
%%%
see appendix \ref{apendice_deformaciones_algebras}.
%%%

When a ribbon $\rho_2$ crosses two domain walls which are
respectively crossed by two other ribbons $\rho_3$, $\rho_4$ with
the same orientation as $\rho$, as in
Fig.~\ref{figura_domain_wall_vs_charge}, we have
\begin{equation}\label{domain_wall_coproducto}
J_{\rho_2}^{\chi t} =\sum_{\xi}\sum_{k\in G/M} J_{\rho_3}^{\inver
\xi^{k} \inver k}J_{\rho_4}^{\xi\chi\,  k t},
\end{equation}
where $\xi$ runs over the group of characters, see appendix
\ref{apendice_ribb_open}.

\begin{figure}
\includegraphics[width=7.3 cm]{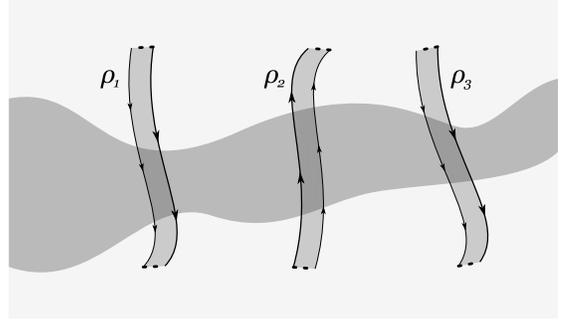}
\caption{ Examples of open ribbon transformations for
$\Ribbopen_\rho$. Edge excitations are present only in the shaded
area. No confined quasiparticle excitations are present. The ribbon
$\rho_1$ is a deformation of $\rho_3$, and $\rho_2$ has an
 inverse orientation with respect to them.}\label{figura_deformacion_open}
\end{figure}

\subsubsection{Closed ribbon operators.}

Before we can further analyze the consequences of the properties of
$\Ribbopen_\rho$, as we shall do in the next section, we have to
introduce a family of projectors that distinguish the charges in our
models. The situation is different to the one we found in the models
$H_G$, because now some charges are confined, and thus the rules for
deforming closed ribbons will change qualitatively to reflect this
fact. Our starting point is the ribbon operator algebra
$\Ribb_\sigma$ for a closed ribbon $\sigma$, from which we want to
select certain suitable operators, just as we did for other
projector algebras. In the case of open ribbons just considered, we
made this selection requiring that the operators commuted with all
vertex and face terms. Here that will not be enough, because since
$\sigma$ is closed we would be considering operators that create a
closed domain wall, with no ends and no confined excitations.
Therefore, we define $\Ribbclosed\prima_\sigma \subset \Ribb_\sigma$
as the subalgebra containing those operators $K\in\Ribb_\sigma$
which commute with all vertex operators $A_v^M$, face operators
$B_f^N$ and edge operators $L_e^N,T_e^M$.

%%%
As we show in proposition \ref{prop_generadores_ribbclosed_prima}, $\Ribbclosed\prima_\sigma$ is linearly generated by certain orthogonal projectors $K_\sigma^{RC}$ that form a resolution of the identity.
%%%
The labels $(R,C)$ of these projectors are $C$, a set of
the form $\set{mg\inver m}{m\in M}$ for some $g$ in $G$, and $R$, an
irreducible representation of the group $\Norm C\prima:=\set{m\in
M}{mr_C\inver m\inver r_C\in N}$ for some fixed $r_C\in C$. If
$\sigma$ is a boundary ribbon surrounding an area with no vertex or
face excitations in the state $\ket\psi$, then
\begin{equation}\label{carga_prima_trivial}
K_\sigma^{e1} \ket\psi =\ket\psi,
\end{equation}
where $e$ and $1$ are both identity elements, see (\ref{GS_boundary_ribbon}, \ref{base_ribbclosed_prima}, \ref{base_ribbclosed_2}).

The algebra $\Ribbclosed\prima_\sigma$ can stand deformations in
which the end $\partial\sigma$ is not fixed, as long as it crosses
no domain walls. If the state $\ket\psi$ is such that the open
ribbon $\sigma$ can be deformed to obtain another ribbon
$\sigma\prima$ in such a way that no vertex or face excitations are
crossed and the end of $\sigma$ touches no edge excitations, then
\begin{equation}\label{carga_prima_deformacion}
K^{R C}_\sigma\ket\psi = K^{R C}_{\sigma\prima}\ket\psi,
\end{equation}
see appendix \ref{apendice_deformaciones_algebras}. This is
illustrated in Fig.~\ref{figura_deformacion_closed_prima}. As in the
case of closed ribbons in $H_G$, here we can consider inversions in
the orientation of the ribbon, as shown in
Fig.~\ref{figura_deformacion_closed_prima}. When $\sigma\prima$ is a
transformation of $\sigma$ which includes an inversion we have
\begin{equation}\label{carga_prima_inversion}
K^{R C}_\sigma\ket\psi = K^{\inver R^C \inver
C}_{\sigma\prima}\ket\psi,
\end{equation}
where $R^C$ is an irreducible representation of $\Norm {\inver
C}\prima$ defined by $R^C(\cdot):=R(m \cdot \inver m)$ if $\inver
r_C = m r_{\inver C}\inver m$ for some $m\in M$,
%%%
see appendix \ref{apendice_deformaciones_algebras}.
%%%

\begin{figure}
\includegraphics[width=7.3 cm]{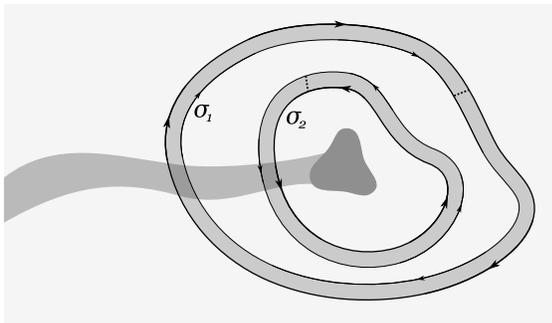}
\caption{ Examples of closed ribbon transformations for
$\Ribbclosed\prima_\rho$. The light shaded area represents edge or
domain wall excitations, and the dark shaded area confined
quasiparticle excitations. The ribbon $\sigma_1$ is a deformation of
$\sigma_2$ which includes an
inversion.}\label{figura_deformacion_closed_prima}
\end{figure}

\subsubsection{Domain walls and charges.}

The deformation properties of the projectors in $\Ribbopen_\rho$
that we have introduced indicate that edge excitations appear in the
form of domain walls to which a flux can be attached. Branching
points are possible in these walls. Each value of the flux
corresponds to a projector $J_\rho^{\chi t}$, and we know that it is
preserved along a domain wall due to the deformation property
\eqref{domain_wall_deformacion}. The trivial flux is given by
\eqref{domain_wall_trivial} and the inverse flux by
\eqref{domain_wall_inversion}. Since deformations of $J_\rho^{\chi
t}$ only require that no confined excitations are crossed, a domain
wall can only end in the presence of confined excitations. Domain
wall fluxes have an abelian nature: as indicated by
\eqref{domain_wall_coproducto}, the addition of two given fluxes
always produces the same combined flux. All these ideas are
reflected in Fig.~\ref{figura_domain_wall_vs_charge}.

Regarding charge labeling, the properties of the projectors in
$\Ribbclosed_\sigma\prima$ very much resemble those already found in
the study of $\Ribbclosed_\sigma$ in systems with Hamiltonian $H_G$.
The trivial charge is given by \eqref{carga_prima_trivial} and the
inverse charge by \eqref{carga_prima_inversion}. Indeed, the new
element that appears is that now the deformation properties
\eqref{carga_prima_deformacion} take into account that the charge
could be attached to a domain wall.

But if we want to neatly describe confinement, we have to establish
the relationship between domain walls and charges. To this end,
consider Fig.~\ref{figura_domain_wall_vs_charge}. We can deform each
ribbon $\rho_i$, without changing the flux it measures, till
$\rho_i$ is equal to $\sigma_i$, the boundary ribbon enclosing the
charge at the end of the domain wall. At that point we can compare
both projector algebras, with the following result: there exists a
function $f$ onto the group character of $N$ such that
\begin{equation}\label{domain_wall_vs_charge}
J_\sigma^{\chi t} = \sum_{C \subset tM} \spc \sum_R
\delta_{\chi,f(R)} K_\sigma^{R C},
\end{equation}
where the sum on $R$ runs over irreducible representations in $\Norm
C\prima$,
%%%
see \eqref{ribbopen_en_ribbclosedprima2}.
%%%
Equation \eqref{domain_wall_vs_charge} tells us at the end
of which domain walls can each charge exist. In other words, the
projectors in $\Ribbopen_\sigma$ classify charge types from
$\Ribbclosed_\sigma\prima$ in different compartments or sectors:
each of these sectors gives a confined charge type. The different
charge labels within a sector give us the topological part of the
charge. In particular, for the trivial confined charge we recover
the topological charge types for a system with Hamiltonian
$H_{M/N}$, in accordance with the study of the ground state of
section \ref{seccion_kitaev_NM_GS}. Finally, all charges which do
not belong to the trivial confined charge sector are indeed
confined, because if we take any circle surrounding them we must
always have a domain wall crossing it. When excitations are
localized in a single site, it turns out that confined excitations
are exactly described, as expected, by the projectors
\eqref{proyectores_confinadas}.

As we already did in the particular case of $H_G$ models, the ground
state can be described in terms of ribbon operators of arbitrary
size. In this case, we have to impose that no region should contain
a nontrivial charge \emph{and} no line should be crossed by a
nontrivial domain wall flux. That is, a state $\xi$ is a ground
state if and only if
\begin{equation}\label{ground_state_N_M_descripcion}
K_\sigma^{e 1} \ket \xi = \ket \xi,\qquad J_\rho^{e 1} \ket \xi =
\ket \xi,
\end{equation}
for all boundary ribbons $\sigma$ and proper ribbons $\rho$. This
conditions generalize to arbitrary $N$ and $M$, see appendix
\ref{apendice_condensacion}.

\subsection{Case II systems}

This section is devoted to those models $H_G^{N,M}$ with $N$
abelian. As we have already commented, in this case we will only be
able to describe, using ribbon projector algebras, domain wall
fluxes, but not charge types, except those in the sector with no
edge excitations which are already classified through the mapping to
$H_{M/N}$. Thus we proceed to describe the algebra $\Ribbopen_\rho$,
which is defined exactly as in case I.
%%%
As discussed in appendix \ref{apendice_ribb_open},
%%%
in this case the projectors
are $J_\rho^{RT}$ with $T$ an element of the double coset
$M\backslash G/M$ and $R$ an induced representation in $\Norm T$ of
an irreducible representation in $N$, with $\Norm T$ the group
$\set{m\in M}{mr_TM=r_TM}$ for some fixed $r_T\in T$. We recall that
each element of the double coset takes the form
$T=\set{mr_Tm\prima}{m,m\prima\in M}$. From the double coset
structure, we can obtain a subalgebra of the group algebra $\C(G)$.
Indeed, for $T,T\prima\in M\backslash G/M$ we have
\begin{equation}\label{algebra_doble_coset}
TT\prima = \sum_{T\primas\in M\backslash G/M}
C^{T\primas}_{TT\prima}\, T\primas,
\end{equation}
where $C^{T\primas}_{TT\prima}$ are positive integers. The rules for
deformations of $J_\rho^{RT}$,
%%%
discussed in appendix \ref{apendice_deformaciones_algebras},
%%%
state that the ribbon $\rho$ can move across regions in which
$B_f^N\ket\psi=\ket\psi$ and $A_v^N\ket\psi=\ket\psi$ are satisfied,
and the ends can be displaced as long as they do not touch edge
excitations, with the result that the flux measured by $J_\rho^{RT}$
does not change. The inverse flux of $(R,T)$ is $(\inver R^T,\inver
T)$, where $\inver M$ contains the inverses of the elements of $M$
and $R^T$ is defined by $R^T(\cdot):=R(r_T m \cdot \inver m \inver
r_T)$ if $m\in M$ is such that $M = r_T m r_{\inver T}M$. Finally,
there is no analog of \eqref{domain_wall_coproducto}, in the sense
that the knowledge of two fluxes is not enough to determine the
total combined flux. This can be seen for example in
\eqref{algebra_doble_coset}. When $M$ is normal (case I), the sum is
reduced to a single term, so that combining a flux $(e,T)$ with a
flux $(e,T\prima)$ will give certain determinate total flux
$(e,T\primas)$. This is no longer true when $M$ is not normal (case
II), giving a non-abelian nature to the domain wall fluxes, which
disappears in case I.

\begin{figure}
\includegraphics[width=7.3 cm]{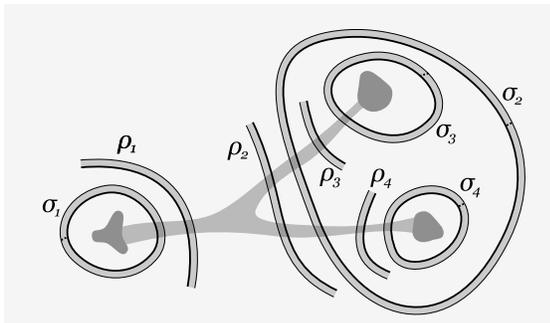}
\caption{An illustration of the relationship between domain wall
fluxes and quasiparticle charges. The light shaded area represents
edge or domain wall excitations, and the dark shaded area
quasiparticle excitations. Each of the open ribbons $\rho_i$ can be
deformed to the boundary ribbon $\sigma_i$, so that the domain wall
flux measured by $\rho_i$ corresponds to the confined charge
measured by $\sigma_i$. Since $\rho_1$ and $\rho_2$ are equivalent
up to an inversion, the confined charges measured by $\sigma_1$ and
$\sigma_2$ are inverses. The total flux in $\rho_2$ is the
combination of that in $\rho_3$ and $\rho_4$, and thus the confined
charge measured by $\sigma_2$ is the combination of that measured by
$\sigma_3$ and $\sigma_4$.}\label{figura_domain_wall_vs_charge}
\end{figure}

\section{Conclusions}
\label{sect_conclusions}

In this paper we have introduced a family of quantum lattice
Hamiltonians with a discrete non-Abelian gauge symmetry such that
the standard Kitaev model for topological quantum computation is a
particular case of this class. The ground state of the models can be
exactly given and, in many cases, quasiparticles or at least domain
wall excitations can be classified. They can be characterized by
operator algebras corresponding to closed and open ribbon operators.
This is done in full generality for arbitrary topologies. The models
can be understood in terms of topological charge condensation and
confinement with respect to the standard Kitaev models.

We have given a detailed account of the quasiparticle excitations in
the standard non-Abelian Kitaev model. In particular, we have seen
that for orientable closed surfaces other than the sphere, like the
torus, excitations may show up in the form of single quasiparticles.

One of the features exhibited by the family of non-Abelian models
considered in this paper is the existence of a string tension for
the motion of quasiparticle excitations. It would be interesting to
study the role of such tensions when the action of external source
of decoherence such as local external fields or thermal effects are
studied \cite{dennis_etal02}, \cite{AFH_07}. Another aspect that
deserves further study is the properties of these models for
topological quantum computation, since here we have only focused on
their properties as far as topological order is concerned.

\noindent {\em Acknowledgements} We acknowledge financial support
from a PFI grant of the EJ-GV (H.B.), DGS grants  under contracts
BFM 2003-05316-C02-01, FIS2006-04885 (H.B., M.A.M.D,), and the ESF
Science Programme INSTANS 2005-2010 (M.A.M.D.).

%%%%%%%%%%%%%%%%%%%%%%%%%%%%%%%%%%%%%%%%%%%%%%%%%%%%%%%%%%%%%%%%%%%%%
%%%%%%%%%%%%%%%%%%%%%%%%%%%%%%%%%%%%%%%%%%%%%%%%%%%%%%%%%%%%%%%%%%%%%
%%%%%%%%%%%%%%%%%%%%%%%%%%%%%%%%%%%%%%%%%%%%%%%%%%%%%%%%%%%%%%%%%%%%%
%%%%%%          APPENDIX            %%%%%%%%%%%%%%%%%%%%%%%%%%%%%%%%%
%%%%%%%%%%%%%%%%%%%%%%%%%%%%%%%%%%%%%%%%%%%%%%%%%%%%%%%%%%%%%%%%%%%%%
%%%%%%%%%%%%%%%%%%%%%%%%%%%%%%%%%%%%%%%%%%%%%%%%%%%%%%%%%%%%%%%%%%%%%
%%%%%%%%%%%%%%%%%%%%%%%%%%%%%%%%%%%%%%%%%%%%%%%%%%%%%%%%%%%%%%%%%%%%%

\appendix

\section{Group algebras}
\label{appendix_A}

We present some basic properties of the algebras of finite groups. We will find them useful for establishing several results on ribbon operator algebras in the next section. In particular, we need them to label the charges and domain walls of our models.

\subsection{Representations and classes}

Given a finite group $G$, let $\Conj G$ be the set of conjugacy
classes of $G$ and $\Irr G$ be the set of irreducible
representations of $G$, up to isomorphisms\cite{serre}. For
$R\in\Irr G$, $g\in G$, we denote by $R_g$ the image of $g$ under
$R$ and by $\Gamma_R(g)$ the unitary matrix of $R_g$ in a particular
basis. The character of a representation $R$ is $\chi_R(g) := \sum_i
\Gamma_R^{ii}(g)$. Characters are examples of class functions
$\funcion {\phi, \psi} {\Conj G}\C$, for which we introduce the
product \cite{serre}
\begin{equation}\label{producto_funciones_clase}
    (\phi, \psi)_G := \frac 1 {|G|} \sum_{C\in \Conj G} |C| \spc \phi(C)\spc\bar
    \psi(C),
\end{equation}
where the bar denotes complex conjugation.

The following are the well-known orthogonality relations for irreducible
representations, characters and conjugacy classes \cite{serre}
\begin{align}
\sum_{g\in G} \Gamma_R^{ij}(g) \bar \Gamma_{R\prima}^{i\prima
j\prima} (g) &= \frac {|G|}{n_R}
\delta_{R,R\prima}\delta_{i,i\prima}\delta_{j,j\prima},\label{ortogonalidad_representaciones}\\
\sum_{C\in \Conj G} |C|\spc \chi_R (C)\spc  \bar \chi_{R\prima} (C) &=
|G| \spc \delta_{R,R\prima}, \label{ortogonalidad_caracteres}\\
\sum_{R \in \Irr {G}}  \chi_R(C) \bar \chi_R(C\prima) &= \frac {|G|} {|C|}\spc\delta_{C,C\prima} ,\label{ortogonalidad_clases}
\end{align}
where $R,R\prima \in \Irr G$, $C,C\prima \in \Conj G$ and
$n_R=\chi_R(1)$ is the degree of $R$.
(\ref{ortogonalidad_caracteres},\ref{ortogonalidad_clases}) imply
that $|\Conj G|=|\Irr G|$. The identities
\eqref{ortogonalidad_caracteres}, which are just a particular case
of \eqref{ortogonalidad_representaciones}, can be written more
concisely as $(\chi_R, \chi_{R\prima})_G = \delta_{R,R\prima}$.

\subsection{Induced representations} \label{apendice_induced_representations}

Given a finite group $G$ and a normal subgroup $H\subset G$, let
$\Conj {H,G}$ be the set of conjugacy classes of $G$ contained in
$H$ and $\Irr {H,G}$ be the set of induced representations in $G$ of
irreducible representations in $H$, up to isomorphisms\cite{serre}.
Recall that for each representation $R$ of $H$ there exists a
representation $\Ind R$, called the induced representation of $R$ in
$G$, such that for $g\in G$
\begin{equation}\label{caracter_inducido}
\chi_{\Ind R} (g) :=
\begin{cases}
\sum_{r\in G/H} \chi_R (rg\inver r), &g\in H, \\%
0, &g\nin H.
\end{cases}
\end{equation}
The Frobenius reciprocity formula asserts that for $\funcion \phi {\Conj G}\C$ and $R\in \Irr{H}$
\begin{equation}\label{Frobenius}
(\phi|_H, \chi_R)_H = (\phi, \chi_{\Ind R})_G,
\end{equation}
where $\phi|_H$ is the restriction of $\phi$ to $H$.

Let us introduce an equivalence relation in $\Conj H$. For
$D,D\prima\in \Conj H$, we set $D\sim D\prima $ if $D\prima=gD\inver
g$ for some $g\in G/H$.  Each $C\in\Conj{H,G}$ is related to a
unique equivalence class $\tilde C$ in the following way
\begin{equation}\label{descomposicion_clases_HG}
C = \bigcup_{D\in \tilde C} D.
\end{equation}

Given $S\in \Irr H$ and $r\in G/H$, define a representation
$S^r\in\Irr H$ setting $S^r_h := S_{rh\inver r}$. We introduce an
equivalence relation in $\Irr H$. For $S,S\prima\in \Irr H$, we set
$S\sim S\prima $ if $S=S^g$ for some $g\in G/H$.  Each
$R\in\Irr{H,G}$ is related to a unique equivalence class $\tilde R$
in the following way
\begin{equation}\label{descomposicion_indreps_HG}
\chi_{R} (g)=
\begin{cases}
\frac {|G|}{|H||\tilde R|}\sum_{S\in \tilde R} \chi_S(g) , &g\in H, \\%
0, &g\nin H,
\end{cases}
\end{equation}
and $S\in \tilde R$ iff $\Ind S = R$.

As a generalization of (\ref{ortogonalidad_caracteres},\ref{ortogonalidad_clases}) we have the following orthogonality relations
\begin{align}
\sum_{C\in \Conj {H,G}} |C|\spc \chi_R (C)\spc  \bar \chi_{R\prima} (C) =
\frac {|G|^2}{|H|\spc|\tilde R|} \spc \delta_{R,R\prima},\label{ortogonalidad_caracteres_HG}\\
\sum_{R \in \Irr {H,G}} |\tilde R|\spc \chi_R(C) \spc \bar \chi_R(C\prima) = \frac {|G|^2} {|H|\spc|C|}\spc\delta_{C,C\prima},\label{ortogonalidad_clases_HG}
\end{align}
where $R,R\prima \in \Irr {H,G}$ and $C,C\prima \in \Conj {H,G}$.
(\ref{ortogonalidad_caracteres_HG},\ref{ortogonalidad_clases_HG})
imply that $|\Conj {H,G}|=|\Irr {H,G}|$. Note that
\eqref{ortogonalidad_caracteres_HG} can be rewritten in terms of the
product $(\cdot,\cdot)_G$ and derived from \eqref{Frobenius} as
follows:
\begin{equation}
(\chi_R, \chi_{R\prima})_G = (\chi_{S_0}, \chi_{R\prima})_H= \frac {|G|}{|H|\spc|\tilde R|} \spc \delta_{R,R\prima},
\end{equation}
where $S_0\in \tilde R$. As for \eqref{ortogonalidad_clases_HG}, it follows from (\ref{descomposicion_indreps_HG}, \ref{ortogonalidad_clases})
\begin{multline}
\sum_{R \in \Irr {H,G}} |\tilde R|\spc \chi_R(C) \spc \bar \chi_R(C\prima) =\\
= \sum_{R \in \Irr {H,G}} \frac {|G|^2} {|H|^2|\tilde R|}\sum_{S,S\prima \in \tilde R} \chi_S(D_0) \spc \bar \chi_{S\prima}(D_0\prima) =\\
=  \sum_{g\in G/H}  \sum_{R \in \Irr {H,G}} \frac {|G|} {|H|}\sum_{S \in \tilde R} \chi_S(D_0) \spc \bar \chi_{S}(g D_0\prima \inver g) =\\
= \frac {|G|} {|D_0|} \sum_{g\in G/H} \delta_{D_0, g D_0\prima
\inver g} = \frac {|G|^2} {|H|\spc|C|}\spc\delta_{C,C\prima},
\end{multline}
where $D_0\in \tilde C$, $D_0\prima\in \tilde C\prima$.

\subsection{The group algebra $\C[G]$}

Given a finite group $G$, the group algebra $\C[G]$ consists of
formal sums $\sum_{g\in G} c_g g$, $c_g\in \C$. We are interested in
certain representation $\funcion {\mathcal R}{G\times
G}{\GL{\C[G]}}$. Let us denote by $\mathcal R_{g_1,g_2}$ the image
of $(g_1,g_2)\in G\times G$. Then $\mathcal R$ is defined by:
\begin{equation}
\mathcal R_{g_1,g_2}(g) := g_1 g \inver g_2, \qquad g\in G,
\end{equation}
where $\inver g$ denotes the inverse of $g$. It turns out that the
following isomorphism holds, as we shall check below,
\begin{equation}\label{descomposicion_GC}
\C[G]\simeq \sum_{R\in \Irr G} V_{R}\otimes V_{\bar R}
\end{equation}
where $V_R$ is the representation space of $R$.

With the aim of checking \eqref{descomposicion_GC} explicitly, let
us consider the following elements of $\C[G]$:
\begin{equation}\label{definicion_E_R_ij}
e_{R}^{ij} := \frac {n_R} {|G|} \sum_{g\in G} \bar \Gamma_{R}^{ij}
(g) \ g
\end{equation}
where $R$ is a representation of $G$ and $i,j=1,\dots,n_R$. For irreducible $R$, there are $|G|$ such elements, because $\sum_{R\in\Irr G} n_R^2=|G|$ due to \eqref{descomposicion_GC}. In fact, they give a new basis for $\C[G]$:
\begin{equation}
g = \sum_{R\in \Irr G} \sum_{i,j=1}^{n_R}\Gamma_R^{ij} (g) \
e_{R}^{ij},
\end{equation}
which can be checked using \eqref{ortogonalidad_clases}. In this basis
\begin{equation}
\mathcal R_{g_1,g_2}\ e_R^{ij} = \sum_{k,l=1}^{n_R}
\Gamma^{ki}_R(g_1) \comp \Gamma^{lj}_{R}(g_2)\ e_R^{kl},
\end{equation}
which gives explicitly the isomorphism
\eqref{descomposicion_GC}, as desired. Let us
define
\begin{equation}\label{adjunto}
\overline{\left( \sum_g c_g \spc g \right)} := \sum_g \bar c_g \spc
\inver g.
\end{equation}
Then, we have
\begin{equation}\label{propiedades_e_R_ij}
e_{R}^{ij}e_{R\prima}^{i\prima j\prima} = \delta_{R,R\prima}
\delta_{j,i\prima} \ e_R^{ij\prima}, \qquad \inver e_{R}^{ij} =
e_{R}^{ji}.
\end{equation}
which follow from \eqref{ortogonalidad_representaciones}.

\subsection{The algebra $\mathcal Z_G$}\label{appendix_A_center}

The center of a group algebra $\C[G]$, denoted $\mathcal Z_G$, is
the subalgebra of elements that commute with all the elements of
$\C[G]$. A basis for the center of $\C[G]$ is the following:
\begin{equation}\label{base_clases_Z_G}
e_C := \sum_{g\in C} g \qquad\qquad C\in\Conj G.
\end{equation}
We have
\begin{equation}\label{propiedades_e_C}
e_C e_{C\prima} =: \sum_{C\primas\in\Conj G}
N_{C,C\prima}^{C\primas} e_{C\primas}, \qquad \inver e_{C} = e_{\bar
C},
\end{equation}
where $\inver C$ denotes the inverse class of $C$ and
$N_{C,C\prima}^{C\primas}\geq 0$ are integers.

With the aim of finding an alternative basis for $\mathcal Z_G$, we
define
\begin{equation}\label{base_proyectores_Z_G}
e_R := \sum_i e_R^{ii}, \qquad R\in\Irr G,
\end{equation}
which are a nice set of projectors:
\begin{align}\label{propiedades e_R}
e_{R}e_{R\prima} &= \delta_{R,R\prima} e_R, \nonumber\\%
\qquad \inver e_{R} &= e_{R}, \nonumber\\%
\sum_{R\in \Irr G} e_R &= 1,
\end{align}
as follows from \eqref{ortogonalidad_clases}. They provide us with a
new basis for $\mathcal Z_G$ since
\begin{align}\label{cambio_base_C_a_R}
e_R &= \frac {n_R} {|G|} \sum_{C\in \Conj G} \bar \chi_R(C) \
e_C,\\
\label{cambio_base_R_a_C}
e_C &= \sum_{R\in \Irr G} \frac {|C|} {n_R} \chi_R(C) \ e_R,
\end{align}
as can be checked using \eqref{ortogonalidad_clases}.
%Since $e_R$ only depends on $\chi_R$, we can set for $\chi=\chi_R$
%\begin{equation}
%e_\chi := e_R.
%\end{equation}

\subsection{The algebra $\mathcal Z_{H,G}$}

Let $H$ be a normal subgroup of $G$. There is a natural inclusion
$\C[H]\subset \C[G]$. We are interested in the intersection of
their centers $\mathcal Z_{H,G} := \mathcal Z_H\cap \mathcal Z_G$.
A basis for the algebra $\mathcal Z_{H,G}$ is
\begin{equation}
e^G_C := \sum_{g\in C} g \qquad\qquad C\in\Conj {H,G}.
\end{equation}
Its elements can be rewritten in terms of elements of $\mathcal Z_H$ as follows
\begin{equation}\label{descomposicion_clase_G_H}
e^G_C = \sum_{D\in \tilde C} e^H_{D}=\frac {|H|\spc |\tilde
C|}{|G|}\sum_{g\in G/H} g e^H_{D} \inver g,
\end{equation}
where we are using the notation of \eqref{descomposicion_clases_HG}.

We have the aim of finding a basis of projectors for $\mathcal Z_{H,G}$, analogous to \eqref{propiedades e_R}. Let
\begin{equation}\label{definicion_proyectores_Z_H_G}
e\prima_R := \frac{|H|\spc |\tilde R|}{|G|} e^G_{R} = \sum_{S\in \tilde R} e^H_S,\qquad R\in \Irr {H,G},
\end{equation}
so that,
\begin{align}\label{propiedades_e_prima_R}
e\prima_R e\prima_{R\prima} &= \delta_{R,R\prima} e\prima_R, \nonumber\\%
\qquad {e\prima_R}^\dagger &= e\prima_R, \nonumber\\%
\sum_{R\in \Irr{H,G}} e\prima_R &= 1,
\end{align}
showing also that the $e_R$ are linearly independent. Indeed,
$\set{e\prima_R}{R\in\Irr{H,G}}$ is the desired projector basis for
$\mathcal Z_{H,G}$, because we have
\begin{align}\label{cambio_base_C_a_R_HG}
e\prima_R &=\frac{n_{R}\spc |\tilde R|\spc |H|}{|G|^2}\sum_{C\in \Conj {H,G}} \bar
\chi_R (C) \spc e^G_C,\\
\label{cambio_base_R_a_C_HG}
e^G_C &= \sum_{R\in\Irr {H,G}} \frac{|C|}{n_R} \chi_R(C) \spc e\prima_R,
\end{align}
where $R\in\Irr {H,G}$ and $C\in \Conj{H,G}$. In order to check
\eqref{cambio_base_R_a_C_HG}, insert \eqref{cambio_base_C_a_R_HG}
and apply \eqref{ortogonalidad_clases_HG}. Finally, not that if $H$
belongs to the center of $G$ we have $e\prima_R = e_R^H$.

\section{Ribbon operators}\label{apendice_todo_ribbons}
\label{appendix_B}

In this appendix we discuss ribbon operator algebras. We will first
introduce the geometric aspects of ribbons, then we will attach
operators to ribbons, and we will finish describing and
characterizing certain projector ribbon subalgebras. These
projectors are directly related to the topological charges and
domain wall fluxes in the models under study.

\subsection{Sites, triangles and strips}\label{apendice_strips}

Our starting point is a lattice embedded in an orientable
two-dimensional manifold. Let $V$, $E$ and $F$ be the sets of its
vertices, edges and faces respectively. The edges of the direct and
dual lattices must be oriented accordingly, as explained in the main
text in sect.\ref{sect_II} and Fig.\ref{figura_red}. A direct edge
$e$ points from the vertex $\partial_0 e$ to the vertex $\partial_1
e$, and a dual edge $e^\ast$ points from the dual vertex $\partial_0
e^\ast$ to the dual vertex $\partial_1 e^\ast$. The shape of the
lattice is arbitrary, but with certain conditions. Namely \emph{(i)}
if $e$ is an edge, then $\partial_0 e\neq
\partial_1 e$ and \emph{(ii)} a face with $s$ edges must have $s$ different vertices. The same conditions
must hold for the dual lattice.

For each edge $e\in E$, we introduce an inverse edge $\comp e $
which is an edge with the direction reversed,i.e., with
$\partial_0 \comp e = \partial_1 e$ and $\partial_1 \comp e = \partial_0 e$.
We let $\comp { \comp e } = e$ and denote by
$\alledges:=E\cup \comp E $ the disjoint union of the original and
inverse edges. For dual edges, we set $\complong {(e^\ast)} = (\comp
e)^\ast$ so that $\alledges^\ast = E^\ast\cup \comp E^\ast$.

A (direct) path $p$ is a list $(v_0,e_1,v_1,\dots, e_n,v_n)$ such
that $v_i\in V$, $e_i\in \alledges$, $\partial_0 e_i= v_{i-1}$ and
$\partial_1 e_i= v_{i}$. A dual path $p^\ast$ is a list
$(f_0^\ast,e_1^\ast, f_1^\ast, \dots,f_r^\ast)$ such that $f_i\in
F$, $e_i\in \alledges$, $\partial_0 e_i^\ast= f_{i-1}^\ast$ and
$\partial_1 e_i^\ast= f_{i}^\ast$.

\emph{Sites.} A site is a pair $s=(v,f)$ with $f$ a face and $v$ one
of its vertex. We visualize sites as dashed lines connecting the vertex $v$
and the dual vertex $f^\ast$, as shown in Fig.\ref{figura_strips},
and use the notation $s=:(v_s, f_s)$.

\begin{figure}
\includegraphics[width=7 cm]{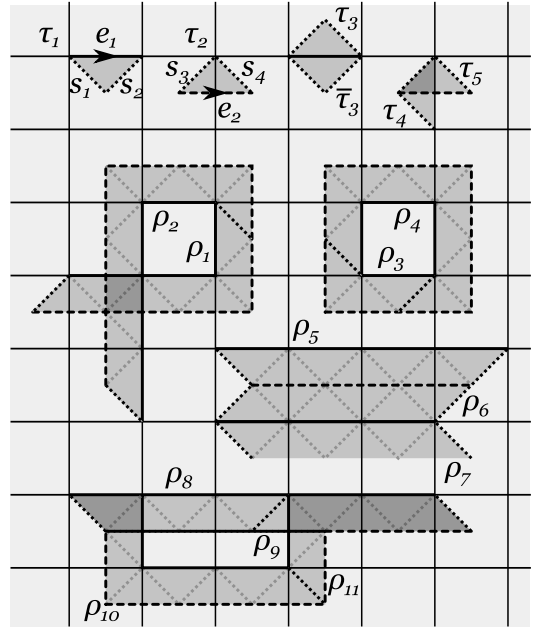}
\caption{
Several examples of sites, triangles and strips. Sites are displayed as black dotted lines when they are an end of some ribbon. Triangles and strips are displayed as grey bands. The direct path of each strip is a thick black line, and the dual path a dashed line. The $s_i$ are sites, the $\tau_i$ are triangles and the $\rho_i$ strips (indeed ribbons). $\tau_1=(s_1,s_2,e_1)$ is a direct triangle and $\tau_2=(s_3,s_4,e_2^\ast)$ is a dual triangle. $\bar\tau_3$ is the complementary triangle of $\tau_3$. $\tau_4$ and $\tau_5$ overlap. $\rho=\rho_1\rho_2$ is a strip but not a ribbon. $\sigma=\rho_3\rho_4$ and $\sigma\prima =\rho_4\rho_3$ are closed strips (indeed ribbons), with $\sigma\prima$ a rotation of $\sigma$ and $\sigma\triangleright\sigma\prima = \rho_3$. $\rho_5$ is a dual complementary ribbon or $\rho_6$, $(\rho_5, \rho_6)_\bigtriangleup$, and $\rho_7$ is direct complementary ribbon of $\rho_6$, $(\rho_6, \rho_7)_\bigtriangledown$. The ribbons $\rho_8,\rho_9,\rho_{10},\rho_{11}$ illustrate various types of joints, so that $(\rho_8,\rho_{10})_\prec$, $(\rho_9,\rho_{11})_\succ$ and $(\rho_8\rho_9,\rho_{10}\rho_{11})_{\prec\succ}$.
}\label{figura_strips}
\end{figure}

\emph{Triangles.} A direct triangle $\tau=(s_0,s_1,e)$ consists of
two sites $s_i$ and a direct edge $e\in \alledges$ such that (i)
$f_{s_0}=f_{s_1}$ (ii) $\partial_i e=v_{s_i}$ and (iii) $s_0$, $s_1$
and $e$ form a triangle with sides listed in \emph{counterclockwise}
order. We use the notation $\tau=: (\partial_0 \tau, \partial_1
\tau, e_\tau)$, and say that $\tau$ points from $\partial_0 \tau$ to
$\partial_1\tau$ through $e_\tau$.

A dual triangle $\tau=(s_0,s_1,e^\ast)$ consists of two sites $s_i$
and a dual edge $e^\ast\in \alledges^\ast$ such that (i)
$v_{s_0}=v_{s_1}$ (ii) $\partial_i e^\ast = f_{s_i}^\ast$ and (iii)
$s_0$, $s_1$ and $e^\ast$ form a triangle with sides listed in
\emph{clockwise} order. We use the notation $\tau=: (\partial_0
\tau, \partial_1 \tau, e_\tau^\ast)$, and say that $\tau$ points
from $\partial_0 \tau$ to $\partial_1\tau$ through $e_\tau^\ast$.

Each direct (dual) triangle $\tau$ has a complementary triangle
$\comp \tau$, the unique direct (dual) triangle with $e_{ \comp \tau
}=\comp e_\tau$.

Two triangles overlap if they `share part of their area'.
Specifically, a dual triangle $\tau$ and a direct triangle
$\tau\prima$ overlap if $\partial_i\tau=\partial_i\tau\prima$ either
for $i=0$ or $i=1$, and two triangles of the same type overlap if
they are the same triangle. Triangles and their properties are illustrated in Fig.\ref{figura_strips}.

\emph{Strips.} A (triangle) strip of length $n\geq 0$ is an
alternating sequence of sites and triangles $\rho = (s_0, \tau_1,
s_1, \tau_2, \dots,s_{n-1}, \tau_n, s_n)$ with $\partial_0 \tau_i =
s_{i-1}$ and $\partial_1 \tau_i = s_i$. We define $\partial_0
\rho:=s_0$ and $\partial_1 \rho:=s_n$. Strips can be given just as a
list of sites $\rho=(s_0,s_1, \dots, s_n)$ or, if they have non-zero
length, as a list of triangles $\rho=(\tau_1, \dots, \tau_n)$. They
could also be given as a pair $\rho=(s, \vect x)$, with $s=s_0$ the
initial site and $\vect x\in 2^n$ a binary vector, because for any
given site there exists exactly a dual and a direct triangle
pointing from it. Examples of strips are given in Fig.\ref{figura_strips}.

Let $\rho$ be a strip with $s_i=\partial_i\rho$. We say that $\rho$
is \emph{(i)} trivial if it has length zero, \emph{(ii)} direct
(dual) if it consists only of direct (dual) triangles, \emph{(iii)}
proper if it is neither direct nor dual, \emph{(iv)} open if
$v_{s_0}\neq v_{s_1}$ and $f_{s_0}\neq f_{s_1}$ and \emph{(v)}
closed if $s_0=s_1$. When $\sigma$ is a closed strip we write
$\partial\sigma$ instead of $\partial_i\sigma$.

Two strips are composable if $\partial_1 \rho_1 = \partial_0
\rho_2$. In that case the composed strip is $\rho=\rho_1\rho_2 :=
(s_0, \dots, s_m)$ with $\rho_1=(s_0, \dots, s_n)$ and $\rho_2=(s_n,
\dots, s_m)$. This composition operation is clearly associative.

The cyclic nature of the list of triangles of a closed strip allows
to rotate it, see Fig.\ref{figura_strips}. We say that $\sigma\prima$ is a rotation of $\sigma$,
denoted $(\sigma, \sigma\prima)_\circ$, if $\sigma\neq\sigma\prima$,
$\sigma = \rho_1\rho_2$ and $\sigma\prima=\rho_2\rho_1$. In that
case we set $\sigma\triangleright\sigma\prima := \rho_1$. Note that
if $(\sigma, \sigma\prima)_\circ$ then both $\sigma$ and
$\sigma\prima$ are closed.

We can attach both a direct and a dual path to a strip $\rho$, see Fig.\ref{figura_strips}. If
$(\tau_1,\dots,\tau_q)$ is the ordered list of direct triangles in
$\rho$, we set $p_{\rho}=(v_{\partial_0 \tau_1},
e_{\tau_1},v_{\partial_1 \tau_1},\dots,e_{\tau_q},v_{\partial_1
\tau_q})$. Similarly, if $(\tau_1\prima,\dots,\tau_r\prima)$ is the
ordered list of dual triangles in $\rho$, we set
$p_{\rho}^\ast=(f^\ast_{\partial_0 \tau_1\prima},
e^\ast_{\tau_1\prima},f^\ast_{\partial_1
\tau_1\prima},\dots,e^\ast_{\tau_r\prima},f^\ast_{\partial_1
\tau_r\prima})$.

Consider two strips $\rho_1$ and $\rho_2$. We say that $\rho_1$ and
$\rho_2$ \emph{(i)} do not overlap, denoted
$(\rho_1,\rho_2)_{\oslash}$, if no triangle of $\rho_1$ overlaps
with a triangle of $\rho_2$, \emph{(ii)} form a left joint, denoted
$(\rho_1, \rho_2)_\prec$, if $\rho_i=\rho\tau_i\rho_i\prima$ with
$\tau_1$ a dual triangle, $\tau_2$ a direct triangle and
$(\rho_1\prima,\rho_2\prima)_\oslash$, \emph{(iii)} form a right
joint, denoted $(\rho_1, \rho_2)_\succ$, if
$\rho_i=\rho_i\prima\tau_i\rho$ with $\tau_1$ a dual triangle,
$\tau_2$ a direct triangle and
$(\rho_1\prima,\rho_2\prima)_\oslash$, \emph{(iv)} form a left-right
joint, denoted $(\rho_1, \rho_2)_{\prec\succ}$, if
$\rho_i=\rho\tau_i\rho_i\prima\tau_i\prima\rho\prima$ with
$\tau_1$,$\tau_1\prima$ dual triangles, $\tau_2$,$\tau_2\prima$
direct triangles and $(\rho_1\prima,\rho_2\prima)_\oslash$, and
\emph{(v)} form a crossed joint, denoted $(\rho_1, \rho_2)_+$, if
$\rho_i=\rho_i\prima\rho_i\primas$ with
$(\rho_1\prima,\rho_2\prima)_\prec$,
$(\rho_2\primas,\rho_1\primas)_\succ$,
$(\rho_1\prima,\rho_2\primas)_\oslash$ and
$(\rho_1\primas,\rho_2\prima)_\oslash$. Closed crossed joints are
only possible in surfaces of nontrivial topology, see
Fig.~\ref{figura_crossed_joint}. Indeed, their name was chosen with
those cases in which $\rho_1$, $\rho_2$ are closed in mind. The
other joint types are illustrated in Fig.\ref{figura_strips}.

We denote by $V_\rho$ ($F_\rho$) the set of vertices (faces) in a
strip $\rho$, by $\Edual_\rho$ ($\Edirect_\rho$) the set of edges
$e\in\alledges$ with $e=e_\tau$ for some dual (direct) triangle
$\tau$ in $\rho$, and by $\Edualcomp_\rho$ ($\Edirectcomp_\rho$) the
set of their inverses. If $\rho$ and $\rho\prima$ are non-direct
(non-dual) closed strips with $\Edual_\rho=\Edualcomp_{\rho\prima}$
($\Edirect_\rho=\Edirectcomp_{\rho\prima}$), we say that
$\rho\prima$ is a dual (direct) complementary ribbon of $\rho$,
denoted $(\rho, \rho\prima)_\bigtriangleup$ ($(\rho,
\rho\prima)_\bigtriangledown$), see Fig.\ref{figura_strips}.

\subsection{Ribbons}\label{apendice_ribbons}

Ribbons are strips such that its direct and dual path do not
self-cross.
\begin{defn}\label{definicion_ribbon} Let $\rho$ be a triangle strip, $p_\rho = (v_0,\dots,e_q,v_q)$
and $p_\rho^\ast = (f_0^\ast,\dots,e_r^\ast,f_r^\ast)$ . We say that
$\rho$ is a ribbon if
\begin{itemize}
  \item for any triangle $\tau$ in $\rho$, $\comp \tau$ is
  not in $\rho$,
  \item for $0\leq i\neq j\leq q$ with $i\neq 0$ or $j\neq q$, $v_i\neq v_j$ and
  \item for $0\leq i\neq j\leq r$ with $i\neq 0$ or $j\neq r$, $f_i\neq f_j$.
\end{itemize}
\end{defn}

A rotation of a ribbon is a ribbon. Two ribbons $\rho_1$ and
$\rho_2$ are composable if they are composable as strips and
$\rho=\rho_1\rho_2$ is a ribbon. If $\rho$ is a ribbon and
$\rho=\rho_1\rho_2$ as strips, then $\rho_1$ and $\rho_2$ are
ribbons also, and $(\rho_1,\rho_2)_{\oslash}$. Complementary ribbons
do not overlap.

Consider any two sites $s$ and $s\prima$. If $v_s=v_{s\prima}$, we
denote by $\alpha_{s,s\prima}$ the unique nontrivial dual ribbon
$\rho$ with $\partial_0 \rho = s$ and $\partial_1 \rho = s\prima$.
If $f_s=f_{s\prima}$, we denote by $\beta_{s,s\prima}$ the unique
nontrivial direct ribbon $\rho$ with $\partial_0 \rho = s$ and
$\partial_1 \rho = s\prima$. All nontrivial dual (direct) ribbons
take the form $\alpha_{s,s\prima}$ ($\beta_{s,s\prima}$) for some
$s,s\prima$. We also write $\alpha_s:=\alpha_{s,s}$ and
$\beta_s:=\beta_{s,s}$. Then if $(\alpha_s, \alpha_{s\prima})_\circ$
we have $\alpha_s\triangleright
\alpha_{s\prima}=\alpha_{s,s\prima}$, and if $(\beta_s,
\beta_{s\prima})_\circ$ we have $\beta_s\triangleright
\beta_{s\prima}=\beta_{s,s\prima}$.

\subsection{Triangle operators}

From this point on we are working with a fixed finite group $G$. To
each edge $e\in E$ we attach a Hilbert space $\Hilb_G\prima$
with orthonormal basis $\sset{\ket g}_{g\in G}$. The total Hilbert
space of our system is then $\Hilb_G:={\Hilb_G\prima}^{\otimes
|E|}$. If $O$ is a single-qudit operator, $O_e$ with $e\in\alledges$
denotes that operator acting on the qudit attached to the edge $e$
($\comp e$) if $e\in E$ ($e\in\comp E$).

Before we can define operators for arbitrary ribbons we must
consider their elementary components, triangles. To this end,  let
\begin{equation}\label{definicion_L_T_I}
L^h := \sum_{g \in G} \ket {hg} \bra g , \quad T^g := \ket g \bra g
, \quad I := \sum_{g \in G} \ket {\inver g } \bra g.
\end{equation}
If $\tau$ is a dual triangle, we set
\begin{equation}\label{definicion_op_triangulo_dual}
L_\tau^{h} := I^xL^h_{e_\tau}I^x
\end{equation}
with $x=0$ ($x=1$) if $e_\tau\in E$ ($e_\tau\in \comp E$). If
$\tau\prima$ is a direct triangle, we set
\begin{equation}\label{definicion_op_triangulo_directo}
T_{\tau\prima}^{g}:= I^xT^g_{e_{\tau\prima}}I^x
\end{equation}
with $x=0$ ($x=1$) if $e_\tau\in E$ ($e_\tau\in \comp E$)). With
these definitions we have
\begin{align}\label{propiedades_algebras_L_T}
L_\tau^h L_\tau^{ h\prima}=L_\tau^{hh \prima}, \qquad
&T_{\tau\prima}^{g}T_{\tau\prima}^{g\prima} = \delta_{g,g\prima}T_{\tau\prima}^{g}, \nonumber\\
{L_\tau^h}^\dagger=L_\tau^{\inver h}, \qquad
&{T_{\tau\prima}^{g}}^\dagger = T_{\tau\prima}^{g}, \nonumber\\
L_\tau^1 = 1, \qquad &\sum_{g\in G} T_{\tau\prima}^g = 1.
\end{align}
As for the commutation rules, we have
\begin{equation}\label{conmutacion_triangulos}
L_\tau^h T_{\tau\prima}^g =
\begin{cases}
T_{\tau\prima}^{hg} L_\tau^h, &\text{if } \partial_0
\tau=\partial_0 \tau\prima;\\
T_{\tau\prima}^{g\inver h} L_\tau^h &\text{if } \partial_1
\tau=\partial_1 \tau\prima;\\
T_{\tau\prima}^g L_\tau^h, &\text{otherwise.}
\end{cases}
\end{equation}
If $\tau_1\neq\tau_2$ are dual triangles and $\tau_1\prima$,
$\tau_2\prima$ are direct triangles then
\begin{equation}\label{conmutacion_triangulos_mismo_tipo}
[L_{\tau_1}^h,L_{\tau_2}^{h\prima}]=[T_{\tau_1\prima}^g,T_{\tau_2\prima}^{g\prima}]=0.
\end{equation}
Thus, non-overlapping triangles have commuting triangle operators.

\subsection{Ribbon operators}\label{apendice_Fhg}

For each ribbon $\rho$ we introduce a set of operators
$\sset{F_\rho^{h,g}}$ with $h,g\in G$. We call them ribbon
operators. First, if $\epsilon$ is a trivial ribbon
we define
\begin{equation}\label{definicion_ribb_trivial}
F_\epsilon^{h,g} := \delta_{1,g}
\end{equation}
If $\tau$ is a dual triangle and $\tau\prima$ a direct triangle we
set \cite{kitaev}
\begin{equation}\label{definicion_ribb_triangulos}
F_\tau^{h,g}:= \delta_{1,g} L^h_{\tau}, \qquad
F_{\tau\prima}^{h,g}:= T^g_{\tau\prima}.
\end{equation}
If $\rho$ is an arbitrary ribbon of length $l
> 1$, we let $\rho=\rho_1\rho_2$ and recursively define
a gluing or composition procedure by means of the folowing
relations
\begin{equation}\label{definicion_ribb_recursiva}
 F_\rho^{h,g}:=\sum_{k\in G}F_{\rho_1}^{h,k}F_{\rho_2}^{\inver k hk, \inver k g}.
\end{equation}
We must of course ensure that this definition of $F_\rho^{h,g}$ is
independent of the particular choice of $\rho_1$ and $\rho_2$. But
this amounts to check that if $\rho=\rho_1\rho_2\rho_3$ then
\begin{equation}
\sum_{k\in G}F_{\rho_1}^{h,k}F_{\rho_2\rho_3}^{\inver k hk, \inver k
g} = \sum_{k\in G}F_{\rho_1\rho_2}^{h,k}F_{\rho_3}^{\inver k hk,
\inver k g}
\end{equation}
which follows by expanding $F_{\rho_2\rho_3}^{\inver k hk, \inver k
g}$ and $F_{\rho_1\rho_2}^{h,k}$ with
\eqref{definicion_ribb_recursiva}. We also have to check that if
$\rho=\epsilon \rho = \rho\epsilon\prima$ then
\begin{multline}
F_\rho^{h,g} = \sum_{k\in G} F_{\epsilon}^{h,k}F_{\rho}^{\inver k h
k, \inver k g} = \sum_{k\in G}F_{\rho}^{h,k}
F_{\epsilon\prima}^{\inver k hk, \inver k g},
\end{multline}
which indeed holds true.

We will find useful the notation
\begin{equation}
T_\rho^g:=F_\rho^{1,g}, \qquad L_\rho^h:=\sum_{g\in G} F_\rho^{h,g}.
\end{equation}
conceived so that
\begin{equation}\label{F_segun_TL}
F_\rho^{h,g}= L_\rho^h T_\rho^g = T_\rho^g L_\rho^h.
\end{equation}
Also, for any $S\subset G$, $g\in G$ we set
\begin{equation}
F^{S,g}:= \frac 1 {|S|}\sum_{s\in S} F^{s,g},\quad F^{g,S}:=
\sum_{s\in S} F^{g,s}.
\end{equation}

We now list several properties of ribbon operators. They follow from
the properties of triangle operators and
\eqref{definicion_ribb_recursiva}. For any ribbon $\rho$
\begin{equation}\label{producto_ribbons}
F_\rho^{h,g}F_\rho^{h\prima, g\prima} = \delta_{g,g\prima}
F_\rho^{hh\prima, g}, \qquad {F_\rho^{h,g}}^\dagger = F_\rho^{\inver
h ,g},
\end{equation}
\begin{equation}\label{unidades}
L_\rho^{1}=\sum_{g\in G}T_\rho^{g}=1.
\end{equation}
Thus, for each $\rho$, ribbon operators linearly generate an algebra
closed under Hermitian conjugation. If $\rho$ is dual and $\rho\prima$ direct then
\begin{equation}\label{simplicidad_directo_dual}
F_{\rho}^{h,g}=\delta_{g,1}L_{\rho}^h, \qquad\qquad
F_{\rho\prima}^{h,g}=T_{\rho\prima}^g.
\end{equation}
If $(\rho_1, \rho_2)_\prec$ then
\begin{equation}\label{conmutacion_left_joint}
F_{\rho_1}^{h,g} F_{\rho_2}^{k,l} = F_{\rho_2}^{hk\inver h ,hl}
F_{\rho_1}^{h,g}.
\end{equation}
If $(\rho_1, \rho_2)_\succ$ then
\begin{equation}\label{conmutacion_right_joint}
F_{\rho_1}^{h,g} F_{\rho_2}^{k,l} = F_{\rho_2}^{k,l\inver g \inver h
g} F_{\rho_1}^{h,g}.
\end{equation}
If $(\rho_1, \rho_2)_{\prec\succ}$ then
\begin{equation}\label{conmutacion left-right joint}
F_{\rho_1}^{h,g} F_{\rho_2}^{k,l} = F_{\rho_2}^{hk\inver h ,hl\inver
g \inver h g} F_{\rho_1}^{h,g}.
\end{equation}
If $(\rho_1, \rho_2)_{+}$ then
\begin{equation}\label{conmutacion crossed_joint}
F_{\rho_1}^{h,g} F_{\rho_2}^{k,l} = F_{\rho_2}^{hk\inver h,hl}
F_{\rho_1}^{h,g\inver lkl}.
\end{equation}
If $(\sigma_1, \sigma_2)_\circ$ then
\begin{equation}\label{rotacion}
F_{\sigma_1\triangleright\sigma_2}^{k,l} F_{\sigma_1}^{h,g} =
F_{\sigma_1\triangleright\sigma_2}^{kh\inver g \inver h g,l}
F_{\sigma_2}^{\inver l hl, \inver l gl}.
\end{equation}
If $(\rho_1,\rho_2)_\oslash$ then
\begin{equation}\label{conmutacion_overlap}
[F_{\rho_1}^{h,g}, F_{\rho_2}^{k,l}]=0.
\end{equation}

\begin{prop}
Let $\rho$ be a ribbon, $h,g\in G$.

\noindent{(i)} If $\rho$ is dual
\begin{equation}\label{valor_traza_dual}
 \traza(L_\rho^{h})= \delta_{h,1}\traza(1).
\end{equation}

\noindent{(ii)} If $\rho$ is direct
\begin{equation}\label{valor_traza_direct}
 |G|\traza(T_\rho^{g})= \traza(1).
\end{equation}

\noindent{(iii)} If $\rho$ is proper
\begin{equation}\label{valor_traza}
 |G|\traza(F_\rho^{h,g})= \delta_{h,1}\traza(1).
\end{equation}
\end{prop}

\begin{proof}
If $\rho$ is not dual, choose any direct triangle $\tau$ in $\rho$,
so that $\rho = \rho_1\tau\rho_2$. Let $\rho\prima =
\rho_1\tau\prima$ with $\tau\prima$ dual. Then $(\rho\prima,
\rho)_\prec$ and thus using \eqref{conmutacion_left_joint}
\begin{multline}\label{traza_intermedia1}
\traza(F_\rho^{h,g}) = \traza(L_{\rho\prima}^k F_\rho^{h,g}
L_{\rho\prima}^{\inver k}) = \\ \traza(F_\rho^{kh\inver k ,kg}
L_{\rho\prima}^k L_{\rho\prima}^{\inver k }) =
\traza(F_\rho^{kh\inver k ,kg}).
\end{multline}
If $\rho$ is not direct, choose any dual triangle $\tau$ in $\rho$,
so that $\rho = \rho_1\tau\rho_2$. Let $\rho\prima =
\rho_1\tau\prima$ with $\tau\prima$ direct. Then $(\rho,
\rho\prima)_\prec$ and thus using \eqref{conmutacion_left_joint}
\begin{multline}\label{traza_intermedia2}
\traza(F_\rho^{h,g}) = \sum_{k\in G}\traza(T_{\rho\prima}^k
F_\rho^{h,g} T_{\rho\prima}^k) = \\
\sum_{k\in G} \traza(F_\rho^{h,g} T_{\rho\prima}^{hk}
T_{\rho\prima}^k) = \delta_{h,1}\traza(F_\rho^{h,g}).
\end{multline}
Equations (\ref{traza_intermedia1}, \ref{traza_intermedia2})
together with (\ref{unidades}, \ref{simplicidad_directo_dual}) give
the desired results.
\end{proof}

As a consequence of the previous proposition and
\eqref{producto_ribbons} we have the following orthogonality
results.

\begin{cor}\label{cor_ortogonalidad_ribbon_operators}
Let $\rho$ be a ribbon, $h,g\in G$.

\noindent{(i)} If $\rho$ is dual
\begin{equation}
\traza({L^{h}_\rho}^\dagger
L^{h\prima}_\rho)=\delta_{h,h\prima}\traza(1).
\end{equation}

\noindent{(ii)} If $\rho$ is direct
\begin{equation}
|G|\traza({T^{g}_\rho}^\dagger
T^{g\prima}_\rho)=\delta_{g,g\prima}\traza(1).
\end{equation}

\noindent{(iii)} If $\rho$ is proper
\begin{equation}
|G|\traza({F^{h,g}_\rho}^\dagger
F^{h\prima,g\prima}_\rho)=\delta_{h,h\prima}\delta_{g,g\prima}\traza(1).
\end{equation}
\end{cor}

\begin{defn}
{\bf Rotationally invariant ribbon operators} Given a closed ribbon
$\sigma$, we say that an operator $F=\sum_{h,g\in G} c_{h,g}
F_\sigma^{h,g}$, $c_{h,g}\in \C$, is rotationally invariant if for
any $\sigma\prima$ with $(\sigma, \sigma\prima)_\circ$ we have
$F=\sum_{h,g\in G} c_{h,g} F_{\sigma\prima}^{h,g}$.
\end{defn}

In this regard, results (\ref{producto_ribbons}, \ref{unidades},
\ref{rotacion}) imply that if $(\sigma_1, \sigma_2)_\circ$ then
\begin{equation}\label{efecto_rotacion}
F_{\sigma_1}^{h,g} = \sum_{l\in G}
F_{\sigma_1\triangleright\sigma_2}^{h\inver g \inver h g,l}
F_{\sigma_2}^{\inver l hl, \inver l gl}.
\end{equation}

\subsection{Vertex and face operators}\label{apendice_vertex_face}

We now define vertex and face operators in terms of ribbon
operators. Let $\alpha_s$ ($\beta_s$) be the unique dual (direct)
closed ribbon with $\partial\alpha_s=s$ ($\partial\beta_s=s$). For
any site $s=(v,f)$ let
\begin{equation}
A_s^h := F_{\alpha_s}^{h, 1}\qquad B_s^g := F_{\beta_s}^{1, \inver
g}.
\end{equation}
Let $s\prima, s\primas$ be sites with $(\alpha_s,
\alpha_{s\prima})_\circ$, $(\beta_s, \beta_{s\primas})_\circ$. From
(\ref{simplicidad_directo_dual}, \ref{efecto_rotacion}) we get
\begin{equation}
A_{s}^h = A_{s\prima}^h, \qquad B_{s}^k = T_{\beta_{s,s\primas}}^{g}
B_{s\primas}^{\inver g kg}.
\end{equation}
Thus, vertex operators are rotationally invariant and we can write
$A_v^h:=A_s^h$ for $v=v_s$. These definitions of $A_v^h$ and $B_s^g$ agree with those given in (\ref{definicion_vertex_op_g},\ref{definicion_face_op_g}).

Let us list several useful properties. If $s\neq s\prima$ then
\begin{equation}\label{propiedades_A}
A_s^{h} A_s^{h\prima}= A_s^{hh\prima}, \quad A_s^1 = 1, \quad
{A_s^h}^\dagger = A_s^{\inver h },
\end{equation}
\begin{equation}\label{propiedades_B}
 B_s^{g}
B_s^{g\prima}=\delta_{g,g\prima} B_s^{g}, \quad \sum_{g\in G}
B_s^g=1, \quad {B_s^g}^\dagger = B_s^g.
\end{equation}
\begin{equation}\label{conmutacion AB}
A_s^{h} B_s^g = B_s^{hg\inver h } A_s^h,
\end{equation}
\begin{equation}\label{conmutacion_AB_resto}
[A_s^h,B_{s\prima}^g]=[A_s^g,A_{s\prima}^{g\prima}]=[B_s^h,B_{s\prima}^{h\prima}]=0.
\end{equation}
All these properties follow from the properties of ribbon operators.
Note in particular that the well-known\cite{Bais_80} flux metamorphosis \eqref{conmutacion AB} is a consequence of $(\alpha_s, \beta_s)_{\prec\succ}$.

For subgroups $H,H\prima\subset G$, $H\prima$ normal, we define the
operators
\begin{equation}\label{definicion_AM_BN}
A_v^H:=A_s^H:=L_{\alpha_s}^H,\qquad
B_f^{H\prima}:=B_s^{H\prima}:=T_{\beta_s}^{H\prima},
\end{equation}
where $s=(v,f)$ is a site. We set $A_v:=A_v^G$ and $B_f:=B_f^1$.
From (\ref{conmutacion AB},\ref{conmutacion_AB_resto}) we have for
arbitrary vertices $v,v\prima$ and faces $f,f\prima$
\begin{equation}\label{conmutacion_AM_BN}
[A_v^H,A_{v\prima}^H]=[A_v^H,B_{f}^{H\prima}]=[B_f^{H\prima},B_{f\prima}^{H\prima}]=0.
\end{equation}

Of particular interest are the commutation rules between ribbon
operators and vertex and face operators at their ends. We first
consider non-closed ribbons. Let $s_i = \partial_i \rho$. If
$v_0\neq v_1$ then $(\alpha_{s_0}, \rho)_\prec$, $(\alpha_{s_1},
\rho)_\succ$ and from (\ref{conmutacion_left_joint}, \ref{conmutacion_right_joint}) we get
\begin{align}\label{conmutacion_extremos_A}
A_{s_0}^k F^{h,g}_\rho &= F^{kh\inver k ,kg}_\rho A_{s_0}^k, \nonumber\\
A_{s_1}^k F^{h,g}_\rho &= F^{h,g\inver k }_\rho A_{s_1}^k.
\end{align}
If $f_0\neq f_1$ then  $(\rho, \beta_{s_0})_\prec$, $(\rho,
\beta_{s_1})_\succ$ and from (\ref{conmutacion_left_joint}, \ref{conmutacion_right_joint}) we get
\begin{align}\label{conmutacion_extremos_B}
B_{s_0}^k F^{h,g}_\rho &= F^{h,g}_\rho B_{s_0}^{kh}, \nonumber\\
B_{s_1}^k F^{h,g}_\rho &= F^{h,g}_\rho B_{s_1}^{\inver g \inver h g
k}.
\end{align}
If $\rho$ is dual but not closed then $\alpha_{s_0} =
\alpha_{s_0,s_1} \alpha_{s_1,s_0}$ and from (\ref{definicion_ribb_recursiva}, \ref{producto_ribbons}, \ref{conmutacion_overlap}) we get
\begin{equation}\label{conmutacion_A_dual}
A_{s_i}^k F^{h,1}_\rho = F^{kh\inver k ,1}_\rho A_{s_i}^k,
\end{equation}
and if $\rho$ is direct but not closed then $\beta_{s_0} =
\beta_{s_0,s_1} \beta_{s_1,s_0}$ and from (\ref{definicion_ribb_recursiva}, \ref{producto_ribbons}, \ref{conmutacion_overlap}) we get
\begin{equation}\label{conmutacion_B_directo}
[B_{s_i}^k,F^{1,g}_\rho] = 0.
\end{equation}

Now we consider closed ribbons. So let $\sigma$ be a closed ribbon
with $s = \partial \sigma$. If $\sigma$ is a proper closed ribbon
then $(\alpha_s, \sigma)_{\prec\succ}$ and $(\sigma,
\beta_s)_{\prec\succ}$ so that from (\ref{conmutacion left-right joint}) we get
\begin{align}\label{conmutacion_extremos_closed}
A_s^k F^{h,g}_\sigma &= F^{kh\inver k ,kg\inver k }_\sigma A_s^k, \nonumber\\
B_s^k F^{h,g}_\sigma &= F^{h,g}_\sigma B_s^{\inver g \inver h g k
h}.
\end{align}
If $\sigma$ is closed but not proper, then either $\sigma=\alpha_s$
or $\sigma=\beta_s$, but for that case we already have
\eqref{conmutacion AB}.

Equations (\ref{transporte_edge_directo},
\ref{transporte_edge_dual}) can be generalized. Let $\rho$ be a
ribbon with two ends $i=0,1$ and set $s_i=\partial_i\rho$,
$v_i=v_{s_i}$, $f_i=f_{s_i}$. If $v_0\neq v_1$ from \eqref{conmutacion_extremos_A} we get
\begin{equation}
|G|\sum_{g\in G} T_\rho^g A_{v_i} T_\rho^g = 1
\label{transporte_electrico}
\end{equation}
and if $f_0\neq f_1$ from \eqref{conmutacion_extremos_B} we get
\begin{equation}
\sum_{h\in G} L_\rho^{\inver h } B_{f_i} L_\rho^h = 1.
\label{transporte_magnetico}
\end{equation}
As explained in the main text, section \ref{seccion_ribbons}, these identities show how ribbon
operators can be used to obtain arbitrary states with a number of
excited spots from states with one excited spot less.

\subsection{Edge operators}\label{apendice_edge_operators}

For subgroups $H\subset H\prima\subset G$, $H$ normal in $G$, we
define the operators
\begin{equation}\label{definicion_LN_TM}
L_e^H:=L_\tau^H,\qquad
T_{e\prima}^{H\prima}:=T_{\tau\prima}^{H\prima},
\end{equation}
where $e=e_\tau$, $e\prima=e_{\tau\prima}$ with $\tau$ a dual
triangle and $\tau\prima$ a direct one. Then from
(\ref{conmutacion_triangulos},\ref{conmutacion_triangulos_mismo_tipo})
we have for arbitrary edges $e,e\prima$
\begin{equation}\label{conmutacion_LN_TM}
[L_e^H,L_{e\prima}^H]=[L_e^H,T_{e\prima}^{H\prima}]=[T_e^{H\prima},T_{e\prima}^{H\prima}]=0.
\end{equation}

In some particular cases, triangle operators and ribbons have nice
commuting properties, but this is not always the case. If
$H,H\prima,\tau,\tau\prima$ are as above, with $\tau$ and
$\tau\prima$ either in $\rho$ or with no overlap with it then from (\ref{definicion_ribb_recursiva}, \ref{producto_ribbons}, \ref{conmutacion_overlap}) we have
\begin{equation}\label{conmutacion_Ln_Tm_F_inside}
[L_\tau^H,F_\rho^{h,g}]=[T_{\tau\prima}^{H\prima},F_\rho^{h,g}]=0.
\end{equation}
Other triangles are more complicated. Let $\rho=\rho_1\rho_2$,
$(\tau_1,\rho_1)_\succ$, $(\tau_2,\rho_2)_\prec$,
$(\rho_1,\tau_3)_\succ$, $(\rho_2,\tau_4)_\prec$, $h\in H$ and
$k,l\in G$ with $ks\inver k=s$ for any $s\in H$. Then from (\ref{definicion_ribb_recursiva}, \ref{conmutacion_left_joint}, \ref{conmutacion_right_joint}) we have
\begin{align}\label{conmutacion_Ln_Tm_F_overlap}
L_{\tau_1}^h F_\rho^{k,l} &= \sum_{g\in G} T_{\rho_1}^g F_\rho^{k, g
h\inver g l} L_{\tau_1}^h,\nonumber\\
L_{\tau_2}^h F_\rho^{k,l} &= \sum_{g\in G} T_{\rho_1}^g F_\rho^{k, g
\inver h \inver g l} L_{\tau_2}^h, \nonumber\\
T_{\tau_3}^h F_\rho^{k,l} &= \sum_{g\in G} T_{\rho_1}^g F_\rho^{k,
l} T_{\tau_3}^{h\inver g k g}, \nonumber\\
T_{\tau_4}^h F_\rho^{k,l} &= \sum_{g\in G} T_{\rho_1}^g F_\rho^{k,
l} T_{\tau_4}^{\inver g \inver k g h}.
\end{align}
As a result, if $N\subset M\subset G$ with $N,M$ normal and $N$
central in $M$ we have for any $h,g\in G$, $\tau=\tau_1$ (or
$\tau_2$) and $\tau\prima=\tau_3$ (or $\tau_4$) with $\tau_i$ as
above
\begin{equation}\label{conmutacion_Ln_Tm_F_proyeccion}
L_\tau^N T_{\tau\prima}^M F_\rho^{h,g} L_\tau^N T_{\tau\prima}^M =
\delta_{hM,M} \frac 1 {|N|} F^{h,Ng} L_\tau^N T_{\tau\prima}^M,
\end{equation}
which then gives \eqref{string_tension}.

\subsection{The algebra $\Alg_\rho$}\label{apendice_ribb_rho}

We want to define the ribbon operator algebra $\Ribb_\rho$, but as
an intermediate step we introduce the algebra $\Alg_\rho$. If
$\epsilon$ is a trivial ribbon then $\Alg_\epsilon := \C$. If $\tau$
is a direct (dual) triangle then $\Alg_\tau := \lin
\set{T^g_\tau}{g\in G}$ ($\Alg_\tau := \lin \set{L^h_\tau}{h\in
G}$). Finally, if $\rho=(\tau_i)$ is an arbitrary ribbon $\Alg_\rho
:= \bigotimes_i \Alg_{\tau_i}$. That $\Alg_{\tau}$ is really an algebra follows from \eqref{propiedades_algebras_L_T}.

We now proceed to show several results which are essential in order
to characterize ribbon operator algebras in the next sections.

\begin{lem}\label{lema_cancelacion_ABLT}
Let $\rho$ be a ribbon, $O\in\Alg_\rho$ an operator, $H$ a subgroup
of $G$, and $s$ a site

\noindent(i) If $\rho$ is not a rotation or a complement of $\alpha_s$ then
\begin{equation}\label{cancelacion_Av}
OA_s^H=0\quad\Longrightarrow\quad O=0.
\end{equation}

\noindent(ii) If $\rho$ is not a rotation or a complement of $\beta_s$ then
\begin{equation}\label{cancelacion_Bf}
OB_s^H=0\quad\Longrightarrow\quad O=0.
\end{equation}

\noindent(iii) If $\tau$ is a dual triangle such that neither it nor its complement belong to $\rho$
\begin{equation}\label{cancelacion_Le}
OL_\tau^H=0\quad\Longrightarrow\quad O=0.
\end{equation}

\noindent(iv) If $\tau$ is a direct triangle such that neither it nor its complement belong to $\rho$
\begin{equation}\label{cancelacion_Te}
OT_\tau^H=0\quad\Longrightarrow\quad O=0.
\end{equation}

\end{lem}
\begin{proof}
First, note that $A_s^H A_s^G = A_s^G$ implies $OA_s^H=0 \Rightarrow
OA_s^G=0$, and similarly for $L_\tau^H$, so that it suffices to
consider $H=G$ in these cases. Also, $B_s^H B_s^1 = B_s^1$ and
similarly for $T_\tau^H$, so that it suffices to consider $H=1$ for
them.

\noindent(i) There exists a direct triangle $\tau$ such that it
overlaps with $\alpha_s$ but not with $\rho$, so that
$[T_\tau^h,O]=0$, and a site $s\prima$ such that
$(\alpha_{s\prima},\alpha_s)_\circ$ and $(\tau,
\alpha_{s\prima})_\prec$ or $(\tau,\alpha_{s\prima})_\succ$. In the
first case $A_s^G = A_{s\prima}^G$ and from \eqref{conmutacion_extremos_A} we have $0=\sum_{g} T_\tau^g O A_{s\prima}^G T_\tau^g =
|G|\inv\sum_{g,h} O A_{s\prima}^h T_\tau^{\inver h g} T_\tau^g = O A_{s\prima}^1
T_\tau^G = O$ and the other case is similar.

\noindent(ii) There exists a dual triangle $\tau$ such that it
overlaps with $\beta_s$ but not with $\rho$, so that
$[L_\tau^h,O]=0$, and a site $s\prima$ such that
$(\beta_{s\prima},\beta_s)_\circ$ and $(\beta_{s\prima},
\tau)_\prec$ or  $(\beta_{s\prima}, \tau)_\succ$. In the first case $B_s^1 = B_{s\prima}^1$
from \eqref{conmutacion_extremos_B} we have $0=\sum_h L_\tau^h O B_{s\prima} L_\tau^{\inver h} = \sum_h O B_{s\prima}^{\inver h} L_\tau^h L_\tau^{\inver h} = O B_{s\prima}^G L_\tau^{1} = O$
and the other case is similar.

\noindent(iii, iv) The proofs are analogous to (i,ii).
\end{proof}

Given a vertex $v$ and a dual triangle $\tau$ we set for any $O \in \Alg_\rho$
\begin{equation}
O_k := A_v^k O A_v^{\bar k},\qquad O_k\prima := L_\tau^k O L_\tau^{\bar k}.
\end{equation}
One can check that $O_k,O_k\prima \in \Alg_\rho$.
Then if $\rho$ satisfies the conditions of lemma
\ref{lema_cancelacion_ABLT}
\begin{align}\label{conmutacion_con_Av}
[O,A_v^H]=0\quad\iff\quad |H| O=\sum_{k\in H} O_k,\\
\label{conmutacion_con_Le} [O,L_\tau^H]=0\quad\iff\quad
|H| O=\sum_{k\in H} O\prima_k,
\end{align}
because for $k\in H$ we have $O_k A_v^H = O_k A^k A_v^H = A^k_v O
A_v^H =A^k A_v^H O = A_v^H O = O A_v$ giving $O_k=O$, and similarly
for $L_e^H$. As a consequence, we also get under the same conditions
and $k\in H$
\begin{equation}\label{conmutacion_AvH_Avk}
[O,A_v^H]=0\quad\Longrightarrow\quad [O,A_v^k]=0.
\end{equation}

Given a site $s$ in $\rho$ and a direct triangle $\tau$ , one can check that
unless $\rho$ is both a non-closed and non-direct ribbon with $f_{\partial_i\rho}=f_s$, for any $O\in\Alg_\rho$ there exist $O_{k,k\prima},O_{k,k\prima}\prima \in \Alg_\rho$ such that $O = \sum_{k,k\prima\in G} O_{k,k\prima}=\sum_{k,k\prima\in G} O\prima_{k,k\prima}$ with
\begin{equation}
\quad B_s^g O_{k,k\prima} = O_{k,k\prima}B_s^{\inver k g k\prima}, \quad
T_\tau^g O_{k,k\prima} = O_{k,k\prima}T_\tau^{\inver k g
k\prima}.
\end{equation}
Then if $\rho$ satisfies the conditions of lemma \ref{lema_cancelacion_ABLT} and $H$ is normal in $G$
\begin{align}\label{conmutacion_con_Bf}
[O,B_f^H]=0\quad\iff\quad O= \sum_{k,k\prima\in G \spc|\spc \inver kk\prima\in H} O_{k,k\prima},\\
\label{conmutacion_con_Te} [O,T_\tau^H]=0\quad\iff\quad
O= \sum_{k,k\prima\in G \spc|\spc \inver kk\prima\in H}
O\prima_{k,k\prima}
\end{align}
because $O B_f^H = O B_f^H B_f^H = B_f^H O B_f^H = \sum_{k,k\prima}
B_s^H O^{k,k\prima} B_s^H = \sum_{k,k\prima} O^{k,k\prima}
B_s^{\inver k H k\prima} B_s^H = \sum_{\inver kk\prima\in H}
O^{k,k\prima} B_f^H$ and similarly for $T_\tau^H$. As a consequence,
we also get under the same conditions and $C\in \Conj {G/H}$
\begin{equation}\label{conmutacion_BfH_BfC}
[O,B_f^H]=0\quad\Longrightarrow\quad [O,B_f^{C}]=0,
\end{equation}
where $B_f^C=\sum_{c\in C}\sum_{h\in H} B_f^{ch}$.

\subsection{The algebra $\Ribb_\rho$}\label{apendice_ribb_rho}

In this section we characterize the ribbon operator algebra that has
been introduced so far.

\begin{defn}
Let $\rho$ be a ribbon. The ribbon operator algebra
$\Ribb_\rho\subset\Alg_\rho$ consists of those operators
$F\in\Alg_\rho$ such that $[F,A_v]=[F,B_f]=0$ for any vertex $v\neq
v_{\partial_i \rho}$ and any face $f\neq f_{\partial_i \rho}$,
i=0,1.
\end{defn}

\begin{prop}\label{teorema_generadores_ribb}
Let $\rho$ be a ribbon. The $|G|^2$ ribbon operators $F_\rho^{h,g}$,
$h,g\in G$, linearly generate $\Ribb_\rho$. Moreover, $\rho$ is
proper if and only if they form a basis of $\Ribb_\rho$.
\end{prop}

\begin{proof}
For ribbons $\rho$ of length zero or one, $\Ribb_\rho=\Alg_\rho$
because of \eqref{conmutacion_overlap} and the first part of the
statement follows since ribbon operators generate $\Alg_\rho$. For
ribbons of length $l>1$, we proceed inductively on $l$. So let
$\rho$ be such a ribbon and set $\rho=:\rho\prima\tau$, with $\tau$
a triangle. Observe that $e_{\tau}$ is not part of $\rho\prima$ and
that $\rho\prima$ and $\tau$ share vertices or faces only at their
ends, so that $\Ribb_\rho \subset \Ribb_{\rho}\prima :=
\Ribb_{\rho\prima} \otimes
\Ribb_{\tau}=\genset{F_{\rho\prima}^{h_1,g_1}F_{\tau}^{h_2,g_2}}{h_i,
g_i\in G}$, where $(\cdot)$ is the subspace linearly generated by
the set $\cdot$. In view of \eqref{definicion_ribb_recursiva}, what
we want to show is that
\begin{equation}\label{objetivo_teorema}
\Ribb_{\rho}\primas=\genset{\sum_{k\in G}F_{\rho\prima}^{h,k}
F_{\tau}^{\inver k hk, \inver k g}}{h, g\in G}
\end{equation}
is indeed equal to $\Ribb_\rho$. We set $s=\partial_0\tau$, $v=v_s$,
$f=f_s$ and distinguish two cases.

\noindent\emph{(a) $\tau$ is direct.}
In this case, $\Ribb_\rho\subset \Ribb_{\rho}\prima$ is the
subalgebra of operators commuting with $A_v$. Then from\eqref{conmutacion_con_Av} and \eqref{conmutacion_extremos_A} we get $\Ribb_{\rho}=\genset{\sum_{k\in G} F_{\rho\prima}^{h,k} F_{\tau}^{\inver k h\prima k,
\inver k g}} {h, h\prima, g\in G}$.
Applying $F_{\tau}^{h,g}=F_{\tau}^{h\prima,g}$ here and in \eqref{objetivo_teorema} gives
$\Ribb_\rho = \Ribb_\rho\primas$.

\noindent\emph{(b) $\tau$ is dual.} In this case,
$\Ribb_\rho\subset\Ribb_{\rho}\prima$ with $\Ribb_{\rho}$ the
subalgebra of operators commuting with $B_f$. Then
from\eqref{conmutacion_con_Bf} and \eqref{conmutacion_extremos_B} we
get $\Ribb_{\rho}=\genset{ F_{\rho\prima}^{h,g} F_{\tau}^{\inver g h
g, g\prima}} {h, g, g\prima\in G}$. Applying
$F_{\tau}^{h,g}=F_{\tau}^{h,1}\delta_{g,1}$ here and in
\eqref{objetivo_teorema} gives $\Ribb_\rho = \Ribb_\rho\primas$.

This completes the inductive step. The second part of the statement follows from \eqref{simplicidad_directo_dual} and corollary \ref{cor_ortogonalidad_ribbon_operators}.
\end{proof}

We now construct an alternative basis for $\Ribb_\rho$.
For each conjugacy class $C\in\Conj G$ we choose an element $r_C$
and denote by $\Norm C\subset G$ the subgroup of elements commuting
with $r_C$ and by $Q_C$ a set of representatives of $G/\Norm C$.
Then for each $C\in\Conj G$ we set $C=\sset{c_i}_{i=1}^{|C|}$, and
$Q_C=\sset{q_i}_{i=1}^{|C|}$ so that $c_i=q_i r_C \inver q_i$. Any
$g\in G$ can be written in a unique way as $g= q_i n$, with $n\in
\Norm C$. We introduce index functions as follows: $i(g):=i$ and
$n(g):=n$. For each irreducible representation $R\in\Irr{\Norm C}$,
we choose a particular basis and denote by $\Gamma_R(k)$, $k\in G$,
the corresponding unitary matrices of the representation. The
desired new basis is the following:
\begin{equation} \label{base_ribb_particulas}
F_\rho^{RC;\vect u\vect v} := \frac {n_R} {|\Norm {C}|} \sum_{n\in
\Norm C} \bar \Gamma_R^{jj\prima}(n)\spc F_\rho^{\inver c_i,q_i n
\inver q_{i\prima}},
\end{equation}
where $\vect u=(i,j)$, $\vect v = (i\prima, j\prima)$ with
$i,i\prima=1,\dots, |C|$ and $j,j\prima=1,\dots, n_R$. The inverse
change is
\begin{equation} \label{base_ribb_particulas_inverso}
F_\rho^{h,g} = \sum_{R\in\Irr{\Norm C}} \sum_{j,j\prima=1}^{n_R}
\Gamma_R^{jj\prima}(n_{h,g})\spc
 F_\rho^{RC;\vect u \vect v}
\end{equation}
where $\inver h\in C\in\Conj G$, $n_{h,g} = \inver q_{i(\inver h)} g
q_{i(\inver g \inver h g)}$, $\vect u = (i(\inver h), j)$ and $\vect
v = (i(\inver g \inver h g), j\prima)$ using the index functions for
$C$. That \eqref{base_ribb_particulas} is really a basis follows
from
\begin{multline}\label{ortogonalidad_ribb_particulas}
\traza({F_\rho^{RC;\vect u\vect v}}^\dagger F_\rho^{R\prima C\prima
;\vect u\prima \vect v\prima})=\\=\frac{|n_R|}{|\Norm C|\spc
|G|}\delta_{R,R\prima}\delta_{C,C\prima}\delta_{\vect u,\vect
u\prima}\delta_{\vect v,\vect v\prima}\traza(1).
\end{multline}
Instead of (\ref{conmutacion_extremos_A}, \ref{conmutacion_extremos_B}) we can now write for $D_s^{h,g}:=A_s^hB_s^g$, $s_i=\partial_i\rho$,
\begin{align}\label{conmutacion_D_F}
D_{s_0}^{h,g} F^{RC;\vect u\vect v} &=
\sum_{s=1}^{n_R} \Gamma^{sj}_R(n(h q_i)) \spc F^{RC;\vect u(s)\vect v}D_{s_0}^{h,g\inver c_i}, \nonumber \\
D_{s_1}^{h,g} F^{RC;\vect u\vect v} &= \sum_{s=1}^{n_R} \comp \Gamma^{sj\prima}_R (n(h
q_{i\prima})) \spc F^{RC;\vect u\vect v(s)}D_{s_1}^{h,  c_{i\prima} g},
\end{align}
where $\vect u = (i,j)$, $\vect v = (i\prima,j\prima)$, $\vect u(s)
= (i(h q_i),s)$, $\vect v(s) = (i(h q_{i\prima}),s)$.

\subsection{The algebra $\Ribbclosed_\sigma$}\label{apendice_ribbclosed_sigma}

Here we discuss the algebra of operators that gives the projectors
onto states of different topological charge in systems with
Hamiltonian $H_G$ \eqref{Hamiltoniano_kitaev}.

\begin{defn}
Let $\sigma$ be a closed ribbon. The closed ribbon operator algebra
$\Ribbclosed_\sigma\subset\Alg_\sigma$ consists of those operators
$K\in\Alg_\sigma$ such that $[K,A_v]=[K,B_f]=0$ for every vertex $v$
and face $f$.
\end{defn}

Note that $\Ribbclosed_\sigma\subset\Ribb_\sigma$. It is not
difficult to check that $\Ribbclosed_{\alpha_s}$ is linearly
generated by the operators $A_s^h$, $h\in G$, and
$\Ribbclosed_{\beta_s}$ is linearly generated by the operators
$B_s^C$, with $C\in\Conj G$ and $B_s^C=\sum_{g\in C} B_s^g$. Note
that these are exactly the rotationally invariant subalgebras of
$\Ribb_{\alpha_s}$ and $\Ribb_{\beta_s}$.

For any closed ribbon $\sigma$ we define the operators
\begin{equation}\label{base_ribbclosed}
K_\sigma^{DC} := \sum_{q\in Q_C} \sum_{d\in D} F_\sigma^{qd\inver q
,q r_C \inver  q },
\end{equation}
where $C\in\Conj G$ and $D\in\Conj {\Norm C}$. The point of these
operators is that they are rotationally invariant:
\begin{equation}\label{invariancia_rotacional_ribbclosed}
(\sigma, \sigma\prima)_\circ\quad\Longrightarrow\quad
K_{\sigma\prima}^{D C} = K_\sigma^{D C},
\end{equation}
as can be checked applying \eqref{efecto_rotacion}. In fact, it can
be shown that if $\sigma$ is proper they form a basis of the
subalgebra of rotationally invariant ribbon operators of
$\Ribb_\sigma$. From (\ref{producto_ribbons}, \ref{unidades}) we get
\begin{align}\label{reglas_closed_ribbons}
K_\sigma^{D C}K_\sigma^{D\prima C\prima} &= \delta_{C,C\prima}
\sum_{D\primas} N^{D\primas}_{D D\prima} K_\sigma^{D\primas C},
\nonumber\\%
{K_\sigma^{D C}}^\dagger &= K_\sigma^{\inver D C},
\nonumber\\%
\sum_{C\in \Conj G} K_\sigma^{1C} &= 1.
\end{align}
where the sum runs over $D\primas\in\Conj{\Norm C}$,
$DD\prima=\sum_{D\primas}N^{D\primas}_{D,D\prima}D\primas$ and
$\inver D$ denotes the inverse class of $D$. The result
\eqref{valor_traza} implies for any proper $\sigma$
\begin{equation}\label{valor_traza_closed}
 |G| \traza(K_\sigma^{D C})= |C| \delta_{D,1}\traza(1),
\end{equation}
which together with \eqref{reglas_closed_ribbons} and
$N^1_{DD\prima} = \delta_{\bar D,D\prima}|D|$ gives
\begin{equation}\label{ortogonalidad_ribb_closed}
 \traza({K_\sigma^{D C}}^\dagger K_\sigma^{D\prima C\prima})= \frac {|D||C|}{|G|}
 \delta_{D,D\prima}\delta_{C,C\prima}\traza(1).
\end{equation}

\begin{prop}\label{prop_generadores_ribbclosed}
Let $\sigma$ be a proper closed ribbon. The operators
$K_\sigma^{DC}$, $C\in\Conj G$, $D\in\Conj {\Norm C}$, form a basis
of $\Ribbclosed_\sigma$.
\end{prop}

\begin{proof}
This is just a particular case of proposition
\ref{prop_generadores_ribbclosed_prima}.
\end{proof}

For any proper closed ribbons $\sigma$, consider the subalgebra
$\Ribbclosed^C_\sigma\subset\Ribbclosed_\sigma$ with basis
$\set{K_\sigma^{D C}}{D\in\Conj{\Norm C}}$. The point is that in
view of (\ref{reglas_closed_ribbons}, \ref{propiedades_e_C}) we have
$\Ribbclosed^C_\sigma\simeq\mathcal Z_C$, where $\mathcal Z_C$ is
the center of the group algebra $\C[\Norm C]$. In particular the
isomorphism identifies $K^{DC}$ with $e_D:=\sum_{h\in D}h$. Note
that the isomorphism preserves adjoints as defined in
\eqref{adjunto}. This suggests the introduction of a different basis
for $\Ribbclosed_\sigma$. We define in analogy with
\eqref{cambio_base_C_a_R}
\begin{equation}\label{base_ribbclosed_2}
K_\sigma^{R C} := \frac {n_R} {|\Norm C|} \sum_{D \in \Conj{\Norm
C}} \bar \chi_R(D) \spc K_\sigma^{D C},
\end{equation}
where $R \in \Irr {\Norm C}$. Due to \eqref{cambio_base_R_a_C}, the reverse change of
basis is:
\begin{equation}\label{base_ribbclosed_12}
K_\sigma^{D C} = \sum_{R \in \Irr {\Norm C}} \frac {|D|} {n_R}
 \spc \chi_R(D) \spc K_\sigma^{R C}.
\end{equation}
And due to \eqref{propiedades e_R}, the elements of the new basis are orthogonal
projectors summing up to the identity:
\begin{align}\label{propiedades_proyectores_ribbclosed}
{K_\sigma^{R C}}^\dagger &= K_\sigma^{R C}, \nonumber \\
K_\sigma^{R C} K_\sigma^{R\prima C\prima} &=
\delta_{R,R\prima}\delta_{C,C\prima} K_\sigma^{R C}, \nonumber\\
\sum_{R,C} K_\sigma^{R C} &= 1.
\end{align}

\subsection{The algebra $\Ribbopen_\rho$}\label{apendice_ribb_open}

We discuss now the algebra of operators that gives the projectors
onto states with different domain wall types in systems with
Hamiltonian $H_G^{NM}$ \eqref{Hamiltoniano_NM}, with $N\subset
M\subset G$ subgroups, $N$ normal in $G$.

\begin{defn}\label{defn_open_ribbon_algebra}
Let $\rho$ be an open ribbon. The ribbon operator algebra
$\Ribbopen_\rho\subset\Ribb_\rho$ consists of those operators
$J\in\Ribb_\rho$ such that $[J,A^M_v]=[J,B^N_f]=0$ for every vertex
$v$ and face $f$.
\end{defn}

We denote by $r_T$ an arbitrarily chosen representative of a class
$T$ of the double coset $M\backslash G/M$, by $\Norm T\subset M$ the
subgroup of elements $m$ such that $m r_T M = r_T M$ and by $Q_T$ a
set of representatives of $M/\Norm T$. For any open ribbon $\rho$ we
define the operators
\begin{equation}\label{base_ribbopen}
J_\rho^{CT} := \sum_{q\in Q_T} \sum_{c\in C} F_\rho^{qc\inver q, q
r_T M},
\end{equation}
where $C\in\Conj {N,\Norm T}$ and $T\in M\backslash G/M$. From
\eqref{producto_ribbons} we get
\begin{align}\label{reglas_open_ribbons}
J_\rho^{C T}J_\rho^{C\prima T\prima} &= \delta_{T,T\prima}
\sum_{C\primas} N^{C\primas}_{C C\prima} J_\rho^{C\primas T},
\nonumber\\%
{J_\rho^{C T}}^\dagger &= J_\rho^{\inver C T}, \nonumber\\%
\sum_{T\in M\backslash G/M} J_\rho^{1T} &= 1,
\end{align}
where the sum runs over $C\primas\in\Conj{N,M}$,
$CC\prima=\sum_{C\primas}N^{C\primas}_{C,C\prima}C\primas$ and
$\inver C$ denotes the inverse class of $C$. From
\eqref{valor_traza} we get
\begin{equation}\label{valor_traza_open}
 |G| \traza(J_\rho^{CT})= |T| \delta_{C,1}\traza(1),
\end{equation}
which together with \eqref{reglas_open_ribbons} and
$N^1_{CC\prima} = \delta_{\bar C,C\prima}|C|$ gives
\begin{equation}\label{ortogonalidad_ribb_open}
 |G| \traza({J_\rho^{C T}}^\dagger J_\rho^{C\prima T\prima})= |C||T|
 \delta_{C,C\prima}\delta_{T,T\prima}\traza(1)
\end{equation}

\begin{prop}\label{prop_generadores_ribbopen}
Let $\rho$ be an open ribbon. The operators $J_\rho^{CT}$,
$C\in\Conj {N,\Norm T}$, $T\in M\backslash G/M$, form a basis of
$\Ribbopen_\rho$.
\end{prop}

\begin{proof}
Set $(v_i,f_i) = \partial_i \rho$ and let $\Ribb_\rho\prima\subset
\Ribb_\rho$ be the subalgebra of operators commuting with
$A_{v_i}^M$. From  \eqref{conmutacion_extremos_A} and
\eqref{conmutacion_con_Av} we get $\Ribb_\rho\prima=
\genset{\sum_{m\in M} F_{\rho}^{m h \inver m, m g M}} {h, g\in G}$.
$\Ribbopen_\rho\subset \Ribb_\rho\prima$ is the subalgebra of
operators commuting with $B_{f_i}^N$. From
(\ref{conmutacion_extremos_B}, \ref{conmutacion_con_Bf}) we get
$\Ribbopen_\rho=\genset{\sum_{m\in M} F_{\rho}^{m h \inver m, m g
M}} {h \in N, g\in G}$. Finally, if $h\in C\in \Conj {N,\Norm T}$
and $g\in T \in M\backslash G/M$ we have $\sum_{m,\in M} F_{\rho}^{m
h \inver m, m g M}=\frac{|\Norm T|}{|C|} J_\rho^{C T}$. The result
follows in view of \eqref{ortogonalidad_ribb_open}.
\end{proof}

From the previous proposition and (\ref{definicion_ribb_recursiva},
\ref{conmutacion_Ln_Tm_F_overlap}, \ref{conmutacion_Ln_Tm_F_proyeccion})
we get the following result, which is no longer true if the
condition of $N$ being abelian is removed.

 \begin{cor}\label{cor_ribbopen_LT}
Let $\rho$ be an open ribbon and $e$ an edge. If $N$ is abelian
$[J,L_e^N]=[J,T_e^M]=0$ for any $J\in\Ribbopen_\rho$.
 \end{cor}

For any open ribbon $\rho$ and $T \in M\backslash G/M$, consider the
subalgebra $\Ribbopen^T_\rho\subset\Ribbopen_\rho$ with basis
$\set{J_\rho^{C T}}{C\in\Conj{N,\Norm T}}$. The point is that in
view of (\ref{propiedades_e_C}, \ref{reglas_open_ribbons}) we have
$\Ribbopen_C\simeq\mathcal Z_{N,\Norm T}$. In particular the
isomorphism identifies $J^{CT}$ with $e^M_C$. Note that the
isomorphism preserves adjoints as defined in \eqref{adjunto}. This
suggests the introduction of a different basis for $\Ribbopen_\rho$.
We define in analogy with \eqref{cambio_base_C_a_R_HG},
\begin{equation}\label{base_ribbopen_2}
J_\rho^{R T} := \frac{n_R\spc |\tilde R|\spc |N|}{|\Norm T|^2}
\sum_{C \in \Conj{N,\Norm T}} \bar \chi_R(C) \spc J_\rho^{C T},
\end{equation}
where $R \in \Irr {N,\Norm T}$. Due to \eqref{cambio_base_R_a_C_HG},
the reverse change of basis is:
\begin{equation}
J_\rho^{C T} = \sum_{R\in\Irr {N,\Norm T}} \frac{|C|}{n_R}
\chi_R(C) \spc J_\rho^{R T}
\end{equation}
And due to \eqref{propiedades_e_prima_R}, the elements of the new basis are orthogonal
projectors summing up to the identity:
\begin{align}\label{reglas_open_ribbons_2}
{J_\rho^{R T}}^\dagger &= J_\rho^{R T}, \nonumber \\
J_\rho^{R T} J_\rho^{R\prima T\prima} &=
\delta_{R,R\prima}\delta_{T,T\prima} J_\rho^{R T}, \nonumber\\
\sum_{R,T} J_\rho^{R T} &= 1.
\end{align}

Two comments should be made here. First, in the particular case of
$M$ normal in $G$, $M\backslash G/M=G/M$ and for $T\in G/M$ we have
$\Norm T = M$, so that the two labels for the basis of
$\Ribbopen_\rho$ are not related anymore. Secondly, although
definition \ref{defn_open_ribbon_algebra} only applies to open
ribbons, the algebra $\Ribbopen_\rho$ can be extended to any $\rho$
taking proposition \ref{prop_generadores_ribbopen} as a definition.
As long as $\rho$ is proper, the properties (\ref{reglas_open_ribbons}-\ref{reglas_open_ribbons_2}) will still hold.

The special case of $M$ normal and $N$ central in $M$ deserves
special attention. Instead of
(\ref{base_ribbopen},\ref{base_ribbopen_2}) we can write
\begin{equation}\label{bases_ribbopen_abeliano}
J_\rho^{n, t} :=F^{n,tM},\quad J_\rho^{\chi, t} := \frac{1}{|N|}
\sum_{n \in N} \bar \chi (n) \spc J_\rho^{n, t},
\end{equation}
where $n\in N$, $\tilde t\in G/M$ and $\chi\in(N)_{\mathrm{ch}}$,
with $(N)_{\mathrm{ch}}$ the character group of $N$. Then, if
$\rho=\rho_1\rho_2$ is an open ribbon from
\eqref{definicion_ribb_recursiva} we get
\eqref{domain_wall_coproducto}.

\subsection{The algebra $\Ribbclosed_\sigma\prima$}\label{apendice_ribbclosed_sigma_prima}

Here we discuss the algebra of operators that gives the projectors
onto states of different charge, confined and topological, in
systems with Hamiltonian $H_G^{NM}$ \eqref{Hamiltoniano_NM}, where
$N\subset M\subset G$ are normal subgroups in $G$ with $N$ central
in $M$.

\begin{defn}
Let $\sigma$ be a closed ribbon. The closed ribbon operator algebra
$\Ribbclosed\prima_\sigma\subset\Ribb_\sigma$ consists of those
operators $K\in\Ribb_\sigma$ such that
$[K,A_v^M]=[K,B_f^N]=[K,T_e^M]=[K,L_e^N]=0$ for any vertex $v$, face
$f$ and edge $e$.
\end{defn}

Note that if $N=1$ and $M=G$ then
$\Ribbclosed\prima_\sigma=\Ribbclosed_\sigma$. We extend our
previous notation and set $\Conj {A,B}:=\set{\set{ba\inver b}{b\in B}}{a\in A}$ for two subgroups $A,B$ of some other group. For each class
$C\in \Conj {G/N,M/N}$, we choose a representative $r_C\in G$. Let
$\Norm C\prima := \set{m\in M}{mr_C\inver m \inver r_C\in N}$ and
choose a set $Q_C\subset M$ of representatives of $M/\Norm C\prima$.
For any closed ribbon $\sigma$ we define the operators
\begin{equation}\label{base_ribbclosed_prima}
K_\sigma^{DC} := \sum_{q\in Q_C} \sum_{d\in D} \sum_{n\in N}
F_\sigma^{qd\inver q ,q r_C \inver  q n},
\end{equation}
where $C\in\Conj {G/N,M/N}$ and $D\in\Conj {\Norm C\prima}$. With
this notation, the results \eqref{reglas_closed_ribbons} remain
true, and (\ref{valor_traza_closed},
\ref{ortogonalidad_ribb_closed_prima}) only need a slight
modification:
\begin{align}
 |G| \traza(K_\sigma^{D C})&= |C||N|
 \delta_{D,1}\traza(1), \label{valor_traza_closed_prima}\\%
 \traza({K_\sigma^{D C}}^\dagger K_\sigma^{D\prima C\prima})&= \frac {|D||C||N|}{|G|}
 \delta_{D,D\prima}\delta_{C,C\prima}\traza(1). \label{ortogonalidad_ribb_closed_prima}
\end{align}

\begin{prop}\label{prop_generadores_ribbclosed_prima}
Let $\sigma$ be a proper closed ribbon. The operators
$K_\sigma^{DC}$, $C\in\Conj {G/N,M/N}$, $D\in\Conj {\Norm C\prima}$,
form a basis of $\Ribbclosed_\sigma\prima$.
\end{prop}

\begin{proof}
Set $(v,f) = \partial\sigma$ and let $\Ribb_\sigma\prima\subset
\Ribb_\sigma$ be the subalgebra of operators commuting with $A_v^M$
and $B_f^N$. From (\ref{conmutacion_extremos_closed}, \ref{conmutacion_con_Av},
\ref{conmutacion_con_Bf}) we get
$\Ribb_\sigma\prima=\genset{\sum_{m\in M} F_{\sigma}^{\inver m h
m, \inver m g m}} {h, g\in G, hg\inver h \inver g \in
N}$. $\Ribbclosed_\sigma\prima\subset \Ribb_\sigma\prima$ is the
subalgebra of operators commuting with $L_e^N$ and $T_e^M$ for every
edge $e$. From (\ref{conmutacion_Ln_Tm_F_inside}, \ref{conmutacion_Ln_Tm_F_overlap},
\ref{conmutacion_con_Le}, \ref{conmutacion_con_Te}) it follows that
$\Ribbclosed_\sigma=\genset{\sum_{m\in M} \sum_{n\in N} F_\sigma^{m
h\inver m ,m g \inver m n}}{h\in M, g\in G, hg\inver h \inver g \in
N}$. But given such $h$ and $g$ there exists a class
$C\in\Conj{G/N,M/N}$ and $q\in Q_C$ with $\inver q g q \inver r_C
\in N$, and a class $D\in \Conj{\Norm C\prima}$ with $\inver q hq
\in D$, so that $\sum_{m\in M} \sum_{n\in N} F_\sigma^{m h\inver m
,m g \inver m n}=\frac{|\Norm C\prima|}{|D|} K_\sigma^{D
C}=\frac{|M|}{|C||D|} K_\sigma^{D C}$. The result follows in view of
\eqref{ortogonalidad_ribb_closed_prima}.
\end{proof}

The change of basis \eqref{base_ribbclosed_2} that leads to the
relations \eqref{propiedades_proyectores_ribbclosed} is possible for
$\Ribbclosed_\sigma\prima$ just as it was for $\Ribbclosed_\sigma$,
with the only difference that now the representations $R$ belong to
$\Irr{\Norm C\prima}$.

For any proper closed ribbon we have $\Ribbopen_\sigma\subset
\Ribbclosed_\sigma\prima$. In fact
\begin{equation}\label{ribbopen_en_ribbclosedprima1}
J_\sigma^{nt} = \sum_{C \subset tM} K_\sigma^{nC},
\end{equation}
where $n\in N$, $t\in G/M$, $C\in \Conj {G/N,M/N}$ and
$K_\sigma^{nC}=K_\sigma^{DC}$ with $D=\sset{n} \in \Conj {N, \Norm
{C}\prima}\subset \Conj {\Norm {C}\prima}$. From
(\ref{base_ribbclosed_12}, \ref{base_ribbopen_2},
\ref{ribbopen_en_ribbclosedprima1}) we get the following relation
between the corresponding projector bases
\begin{equation}\label{ribbopen_en_ribbclosedprima2}
J_\sigma^{\chi t} = \sum_{C \subset tM} \spc \sum_{R\in \Irr {\Norm
C\prima}} \frac {(\chi_{R}, \chi)_N}{n_{R}} K_\sigma^{R C},
\end{equation}
where $\chi\in (N)_{\mathrm{ch}}$, $t\in G/M$ and $C\in \Conj{G/N,
M/N}$.

\section{Ribbon transformations}
\label{appendix_C}

\begin{figure}
\includegraphics[width=7 cm]{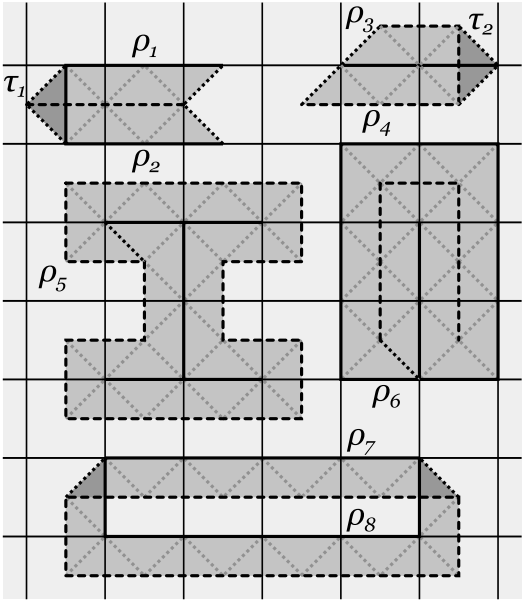}
\caption{
Several constructions with nice strips and ribbons. All elements are displayed as in Fig.~\ref{figura_strips}. The $\rho_i$ are ribbons, except $\rho_5$ and $\rho_6$ which are just nice strips. We have $(\rho_1,\tau_1,\rho_2)_\triangleleft$ and $(\rho_3,\tau_2,\rho_4)_\triangleright$. $\rho_5$ is a dual block and $\rho_6$ is a direct block. $\rho_7$ and $\rho_8$ form a simple deformation, $(\rho_7,\rho_8)_{=}$.
}\label{figura_deformations}
\end{figure}

In this appendix we discuss several transformations that can be
applied to ribbons. These transformations are interesting because
they leave invariant the action of certain ribbon operator algebras
on suitable subspaces of $\Hilb_G$. In order to proof the desired
properties, we need some preliminary results. We will find useful
the notation $\cdot =_\psi \cdot$ for $\cdot
\ket\psi=\cdot\ket\psi$. Also, for strips $\rho,\rho\prima$ and a
triangle $\tau$, we write $(\rho,\tau,\rho\prima)_\triangleleft$ if
$\tau$ is direct, $(\rho,\tau)_\prec$, $(\rho\prima,\tau)_\succ$ and
we write $(\rho,\tau,\rho\prima)_\triangleright$ if $\tau$ is dual,
$(\tau, \rho)_\prec$, $(\tau,\rho\prima)_\succ$, see Fig.~\ref{figura_deformations}.

\begin{lem}\label{lema_inversion}
Let $\rho$, $\rho\prima$ be ribbons, $\ket\psi\in \Hilb_G$ and $H,
H\prima\subset G$ normal subgroups with $hh\prima =h\prima h$ for
any $h\in H$, $h\prima\in H\prima$.

\noindent (i) If $(\rho,\rho\prima)_\bigtriangleup$, there exist a
direct triangle $\tau$ such that
$(\rho,\tau,\rho\prima)_\triangleleft$ and for any $f\in
F_\rho-\sset{f_{\partial_0\rho},f_{\partial_1\rho}}$ we have
$B_f^{H}=_\psi 1$ then for $h\prima\in H\prima$
\begin{equation}\label{lema_inversion_dual}
L_{\rho}^{h\prima}=_\psi \sum_{k,l\in G} F_{\rho\prima}^{l\inver k
\inver h\prima k \inver l,l} T_\tau^k.
\end{equation}
\noindent (ii) If $(\rho,\rho\prima)_\bigtriangledown$ then for
$g\in G$
\begin{equation}\label{lema_inversion_direct}
T_{\rho}^{g}=_\psi T_{\rho\prima}^{\inver g}
\end{equation}
\end{lem}

\begin{proof}
\noindent \emph{(i)} Let $p_\rho^\ast =
(f_0^\ast,\dots,e_r^\ast,f_r^\ast)$ and set $s_i=(\partial_0
e_{i+1}, f_i)$ for $i=1,\dots, r-1$. Consider the states
$\ket{\psi^{g, \vect h}}:= L_{e_1}^{H\prima} T_{\tau}^{g}
L_{e_2}^{H\prima} B_{s_{1}}^{h_1} \cdots L_{e_r}^{H\prima}
B_{s_{r-1}}^{h_{r-1}} \ket\psi$, $g\in G$, $\vect h\in H^{r-1}$.
Then $L_{e_i}^{H\prima} =_{\psi^{g, \vect h}} 1$ and thus
$L_\rho^{H\prima} =_{\psi^{g, \vect h}} L_{\rho\prima}^{H\prima}
=_{\psi^{g, \vect h}} 1$. But $\ket\psi=\sum_{g\in G}T_\tau^g
\sum_{\vect h\in H^{r-1}}\prod_{i=1}^{r-1} B_{s_{i-1}}^{h_i}
\ket{\psi^{\vect g}}$, and the result follows using
\eqref{conmutacion_Ln_Tm_F_overlap} and the fact that
$[L_\rho^{k},B_{s_i}^{h}]=[L_{\rho\prima}^{k},B_{s_i}^{h}]=0$ for
$k\in H\prima$, $h\in H$.

\noindent \emph{(ii)} The proof is dual to (i). Just note that one
must use $\sum_{k\in G} L_{\tau}^{\inver k} T_{e_1}^1 L_{\tau}^k=1$
instead of $\sum_{k\in G} T_{\tau}^{k} L_{e_1}^{H\prima}
T_{\tau}^k=1$ and $\sum_{k\in G} A_{v_{i-1}}^{\inver k} T_{e_i}^1
A_{v_{i-1}}^k=1$ instead of $\sum_{k\in H} B_{f_{i-1}}^{k}
L_{e_i}^{H\prima} B_{s_{i-1}}^k=B_{f_{i-1}}^H$. Alternatively, the
result is trivial in terms of Wilson loops.
\end{proof}

When working with ribbon deformations, it is useful to consider
triangle strips that are more general than ribbons but still allow
to introduce operators. We say that a strip $\rho$ is nice if no two of its triangles overlap or, equivalently, if
$\rho=\rho_1\cdots\rho_n$ with $\rho_i$ ribbons such
that$(\rho_i,\rho_j)_\oslash$ for $i\neq j$. Then ribbon operators
can be generalized for nice strips using
\eqref{definicion_ribb_recursiva}. Although such nice strip
operators still commute with all vertex operators $A_v$ and face
operators $B_f$ except those in their ends, they can no longer be
characterized by this property. A direct (dual) block $\rho$ is a
nice closed strip such that $(\rho,\rho)_{\bigtriangleup}$
($(\rho,\rho)_{\bigtriangledown}$), see Fig.~\ref{figura_deformations}.

\begin{lem}\label{lema_bloques}
Let $\rho$ be a nice closed strip, $\ket\psi\in \Hilb_G$ and
$H\subset G$ a normal subgroup.

\noindent (i) If $\rho$ is a dual block and for any $v\in V_\rho$ we
have $A_v^{H}=_\psi1$ then
\begin{equation}\label{lema_bloques_dual}
L_{\rho}^{H}=_\psi 1
\end{equation}

\noindent (ii) If $\rho$ is a direct block and for any $f\in F_\rho$
we have $B_f^{H}=_\psi 1$ then
\begin{equation}\label{lema_bloques_direct}
T_{\rho}^{H}=_\psi 1
\end{equation}
\end{lem}

\begin{proof}
\noindent \emph{(i)} We proceed recursively on $|V_\rho|$. For
$|V_\rho|=0,1$ the result is trivial. So let $|V_\rho|>1$. Note that
if $(\rho,\rho\prima)_{\circ}$ then $L_{\rho}^{H}= L_{\rho\prima}^H$
due to \eqref{efecto_rotacion}. Also, the path
$p_{\rho}=(v_0,\dots,{e_q\prima},v_r)$ forms a tree. Thus, w.l.o.g.
we can choose $\rho$ such that $v_1= v_{r-1}$ and there exists a
dual triangle $\tau$ such that $\rho=\rho_1\rho_2$ with
$\alpha:=\alpha_{\partial_0 \rho}=\rho_1\tau$ and
$(\rho_2,\tau,\rho_2)_{\triangleleft}$. Set
$\rho_2=\tau\prima\rho_3\comp\tau\prima$, with $\tau\prima$ a
direct triangle. Then
$L_\rho^{h\prima}=L_{\rho_1}^{h\prima}L_{\rho_2}^{h\prima}=L_{\rho_1}^{h\prima}\sum_{k\in
G} T_{\tau\prima}^k L_{\rho_3}^{\inver k h\prima k }=_\psi
L_{\rho_1}^{h\prima}\sum_{k\in G} T_{\tau\prima}^k L_{\inver
\tau}^{\inver k h\prima k}=L_{\rho_1}^{h\prima}
L_{\tau}^{h\prima}=L_{\alpha}^{h\prima}=_\psi 1$, where we have used
the fact that $\rho_3\inver\tau$ is a block and
\eqref{lema_inversion_dual} for $\tau,\inver\tau$.

\noindent \emph{(ii)} Again the proof is dual to (i) or,
alternatively, trivial in terms of Wilson loops.
\end{proof}

\begin{cor}\label{cor_bloques}
Let $\rho=\rho_1\rho_2$ be a nice strip. Under the same conditions of the previous lemma we have,
respectively,

\noindent (i) for $h\in H$
\begin{equation}\label{coroloario_bloques_dual}
L_{\rho_1}^{h}=_\psi \sum_{g\in G} T_{\rho_1}^g L_{\rho_2}^{\inver g
\inver h g},
\end{equation}

\noindent (ii) for $g\in G$
\begin{equation}\label{coroloario_bloques_direct}
T_{\rho_1}^{gH}=_\psi T_{\rho_2}^{\inver gH}.
\end{equation}
\end{cor}

\begin{proof}
Apply (\ref{definicion_ribb_recursiva}, \ref{producto_ribbons}) to (i) (\ref{lema_bloques_dual}) and (ii) (\ref{lema_bloques_direct}).
\end{proof}

For a region $R$ we will understand a collection of faces $f$. We
also consider dual regions $R^\ast$, collections of dual faces
$v^\ast$.

\subsection{Transformation rules}

\subsubsection{Deformations}

Before we define general ribbon deformations, such as the one in Fig.~\ref{figura_deformacion}, we have to introduce certain simpler ones which are easier to manage in proofs, as the one depicted in Fig.~\ref{figura_deformations}. Then simple deformations can be combined together to give the general ones. We say that the ribbons $\rho$, $\rho\prima$ form a simple
deformation, denoted $(\rho, \rho\prima)_{=}$, if \emph{(i)} they
are open, \emph{(ii)} they share no triangles, \emph{(iii)} $(\rho,
\rho\prima)_{\prec\succ}$ and \emph{(iv)} for any $e\in \Edual_\rho$
we have $\partial_1 e\in V_\rho\prima$. The dual of \emph{(iv)} is
automatically true: for any $e \in \Edirect_{\rho\prima}$ we have
$f\in F_\rho$ for $f^\ast=\partial_1 e^\ast$. We will use the
notation $F_{\rho,\rho\prima}=F_\rho-\sset{f_{s_0},f_{s_1}}$ and
$V_{\rho,\rho\prima}=V_{\rho\prima}-\sset{v_{s_0},v_{s_1}}$, where
$s_i=\partial_i\rho=\partial_i\rho\prima$.

Let $R=(R_1,R_2^\ast)$ with $R_1$ a region and $R_2^\ast$ a region
of the dual lattice. We introduce the relation between ribbons
$\simeq_R$ as the minimal equivalence relation such that
$\rho_1\prima\simeq_R \rho_1\prima$ if the following conditions are
all true: $\rho_1=\rho_2\rho\rho_3$,
$\rho_1\prima=\rho_2\rho\prima\rho_3$, $(\rho, \rho\prima)_{=}$,
$F_{\rho,\rho\prima}\subset R_1$ and $V^\ast_{\rho,
\rho\prima}\subset R_2^\ast$. Thus, two ribbons are equivalent in the sense of $\simeq_R$ if they can be transformed one into the other through simple deformations within $R$. Given a state $\ket\psi\in \Hilb_G$
and subgroups $H,H\prima\subset G$, $H$ normal, we define
$R_\psi^{H.H\prima}=(R_1,R_2^\ast)$ with $R_1$ the region such that
$f \in R_1$ iff $B_f^H=_\psi 1$ and $R^\ast_2$ the dual region such
that $v^\ast\in R_2^\ast$ iff $A_v^{H\prima}=_\psi 1$. Then we write
$\simeq_\psi^{HH\prima}$ for $\simeq_{R_\psi^{HH\prima}}$.
\begin{prop}\label{thm_ribbon_deformation}
Let $\ket\psi\in \Hilb_G$ and $H, H\prima\subset G$ normal subgroups
with $hh\prima =h\prima h$ for any $h\in H$, $h\prima\in H\prima$.
If $\rho$, $\rho\prima$ are ribbons with $\rho
\simeq_\psi^{H,H\prima} \rho\prima$ then
\begin{equation}\label{teorema_ribbon_deformation}
F_{\rho}^{h\prima,S}=_\psi F_{\rho\prima}^{h\prima,S} ,
\end{equation}
where $h\prima\in H\prima$, $S\in G/H$.
\end{prop}

\begin{proof}
Using \eqref{definicion_ribb_recursiva} for
$\rho=\rho_1\rho_2\rho_3$ we can write
\begin{equation}
F_\rho^{h\prima,g H} = \sum_{l,m\in G}
F_{\rho_1}^{h\prima,l}F_{\rho_2}^{\inver l h\prima l,\inver lg m H}
F_{\rho_3}^{\inver m \inver g h\prima g m, \inver m},
\end{equation}
and thus it is enough to consider simple deformations $(\rho,
\rho\prima)_{=}$. In that case, we can set
$\rho=\tau_1\rho_1\tau_1\prima$ with $\tau_1,\tau_1\prima$ dual
triangles and there exists a ribbon $\rho_2$ such that
$\rho_1\rho_2$ is a block and the conditions of lemma
\ref{lema_bloques} (ii) are satisfied, so that
\eqref{coroloario_bloques_direct} applies. But
$(\rho_2,\rho\prima)_\bigtriangledown$, so that using
\eqref{lema_inversion_direct} we get $T_\rho^S= T_{\rho_1}^S=_\psi
T_{\rho_2}^{\inver S}= T_{\rho\prima}^S$. We can write
$\rho_2=\inver\tau_2\prima\rho_2\prima\inver\tau_2$ and
$\rho\prima=\tau_2\rho_3\tau_2\prima$ with
$\tau_2,\tau_2\prima$ direct triangles, and set
$\rho_2\primas = \comp \tau_1 \rho_2\prima\comp\tau_1\prima$. Then
$(\rho,\rho_2\primas)_{\bigtriangleup}$ and
\eqref{lema_inversion_dual} applies. Also, $\rho_2\primas\rho_3$ is
a block and the conditions of lemma \ref{lema_bloques} (i) are
satisfied, so that \eqref{coroloario_bloques_dual} applies (for
$H\prima$). Putting everything together we get
$L_\rho^{h\prima}=_\psi\sum_{k,l\in G}F_{\rho_2\primas}^{l\inver k
\inver h\prima k \inver l} T_{\tau_2}^k=_\psi\sum_{k\in G}
L_{\rho_2}^{\inver k h\prima k} T_{\tau_2}^k=_\psi
L_{\rho\prima}^{h\prima}$, where we have used also (\ref{definicion_ribb_recursiva}, \ref{unidades}).
\end{proof}

\subsubsection{Extensions, contractions and rotations}

We also want to consider deformations in which the ends of ribbons
are not fixed. Let $Q=(Q_1,Q_2^\ast)$ with $Q_1, Q_2\subset
\alledges$. We introduce the relation between ribbons $\asymp_{Q}$
as the minimal equivalence relation such that $\rho\asymp_{Q}
\rho\prima$ if $\rho=\rho_1\rho\prima\rho_2$, $\Edirect_{\rho_i}\subset
Q_1$ and $\Edual_{\rho_i}\subset Q_2$. Thus, two ribbons are equivalent in the sense of $\asymp_Q$ if they can be transformed one into the other through extensions or contractions within $Q$. We also introduce an
equivalence relation $\eqcirc_Q$ for closed ribbons, the minimal
such that $\sigma\eqcirc_Q \sigma\prima$ if
$(\sigma,\sigma\prima)_\circ$, $\Edirect_{\sigma \triangleright
\sigma\prima}\subset Q_1$ and $\Edual_{\sigma \triangleright
\sigma\prima}\subset Q_2$. Thus, two closed ribbons are equivalent in the sense of $\eqcirc_Q$ if they can be transformed one into the other through rotations within $Q$. Given a state $\ket\psi\in \Hilb_G$ and
subgroups $H,H\prima\subset G$, $H\prima$ normal, we set
$Q_\psi^{HH\prima}:=(Q_1,Q^\ast_2)$ with $Q_1$ the collection of
edges $e$ with $T_e^H\ket\psi=\ket\psi$ and $Q_2$ the collection of
edges $e\prima$ with $L_{e\prima}^{H\prima}\ket\psi=\ket\psi$. Then
we write $\asymp_\psi^{H,H\prima}$ for $\asymp_{Q_\psi^{HH\prima}}$
and similarly for $\eqcirc$.

\begin{prop}
Let $\ket\psi\in \Hilb_G$ and $H, H\prima\subset G$ subgroups with
$H\prima$ normal.

\noindent (i) If $\rho$, $\rho\prima$ are ribbons with $\rho
\asymp_\psi^{HH\prima} \rho\prima$ then
\begin{equation}\label{teorema_ribbon_extension} %
\sum_{k\in H}F_{\rho}^{k h\prima \inver k, k g H}
=_\psi\sum_{k\in H}F_{\rho\prima}^{k h\prima \inver k, k g H},
\end{equation}
where $h\prima\in H\prima$, $g\in G$.

\noindent (ii) If $\sigma$, $\sigma\prima$ are closed ribbons with
$\sigma \eqcirc_\psi^{HH\prima} \sigma\prima$ then
\begin{equation}\label{teorema_ribbon_rotation}
\sum_{k \in H}F_{\sigma}^{k h \inver k, k g \inver k}
=_\psi \sum_{k \in H}F_{\sigma\prima}^{k h \inver k, k g \inver
k}.
\end{equation}
where $h,g\in G$, $hg\inver h \inver g\in H\prima$.
\end{prop}

\begin{proof}
\noindent (i) It is enough to consider $\rho = \rho\prima\tau$ or
$\rho = \tau\rho\prima$ with $\tau$ a triangle and then apply
\eqref{definicion_ribb_recursiva}.

\noindent (ii) It is enough to consider that
$\sigma\triangleright\sigma\prima = \tau$ and then apply
\eqref{efecto_rotacion}.
\end{proof}

\subsubsection{Inversions}

We finally consider other kind of ribbon transformations in which
basically ribbons are reversed. As in the other cases, we start
introducing suitable relations. For open ribbons $\rho$,
$\rho\prima$ and $R=(R_1,R_2^\ast)$ as above, we write $\rho
\risingdotseq_R \rho\prima$ if
$\partial_0\rho=\partial_1\rho\prima$,
$\partial_1\rho=\partial_0\rho\prima$ and either \emph{(i.a)}
$(\rho,\rho\prima)_{\bigtriangledown}$ and $V_\rho^\ast\subset
R_2^\ast$ or \emph{(i.b)} $(\rho,\rho\prima)_{\bigtriangleup}$ and
$F_\rho\subset R_1$. For closed ribbons $\sigma$, $\sigma\prima$ and
a triangle $\tau$, we write $\sigma \doteqdot_{R,\tau} \sigma\prima$
if either \emph{(ii.a)} $(\sigma,\sigma\prima)_{\bigtriangledown}$,
$\tau$ is dual, $(\sigma,\tau,\sigma\prima)_\triangleright$ and
$V_\rho^\ast\subset R_2^\ast$ or \emph{(ii.b)}
$(\sigma,\sigma\prima)_{\bigtriangleup}$, $\tau$ is direct,
$(\sigma,\tau,\sigma\prima)_\triangleleft$ and $F_\rho \subset R_1$.
For $\ket\psi\in\Hilb_G$, we write $\risingdotseq_\psi^{HH\prima}$
for $\risingdotseq_{R_\psi^{HH\prima}}$ and also
$\doteqdot_\psi^{HH\prima}$ for $\doteqdot_{R_\psi^{HH\prima},\tau}$
if either $\tau$ is dual and $L_\tau^{H} =_\psi 1$ or $\tau$ is
direct and $T_\tau^{H\prima} =_\psi 1$.

\begin{prop}
Let $\ket\psi\in \Hilb_G$ and $H, H\prima \subset G$ normal
subgroups with $hh\prima =h\prima h$ for any $h\in H$, $h\prima\in
H\prima$.

\noindent (i) If $\rho$, $\rho\prima$ are open ribbons with $\rho
\risingdotseq_\psi^{H,H\prima} \rho\prima$ then
\begin{equation}\label{teorema_ribbon_inversion_open}
F_{\rho}^{h\prima, S}=_{\psi} F_{\rho\prima}^{\inver s \inver
h\prima s, \inver S},
\end{equation}
where $h\prima\in H\prima$, $s\in S\in G/H$.

\noindent (ii) If $\sigma$, $\sigma\prima$ are closed ribbons with
$\sigma \doteqdot_\psi^{HH\prima} \sigma\prima$ then
\begin{equation}\label{teorema_ribbon_inversion_closed}
\sum_{k\in H\prima} F_{\rho}^{\inver k h\prima k, \inver k S
k}=_{\psi} \sum_{k\in H\prima} F_{\rho\prima}^{\inver k \inver s
\inver h\prima s k, \inver k \inver S k},
\end{equation}
where $h\prima\in H\prima$ and $s\in S\in G/H$ with $sg\inver
s\inver g\in H$.
\end{prop}

\begin{proof}
\noindent \emph{(i.a)} This case follows from
(\ref{lema_inversion_direct},\ref{coroloario_bloques_dual}).

\noindent \emph{(i.b)} There exists ribbons $\rho_i$ and direct
ribbons $\rho_i\prima,\rho_i\primas$, $i=1,2$,  such that
$\rho=\rho_1\prima\rho_1\rho_1\primas$ and
$\rho=\rho_2\prima\rho_2\rho_2\primas$. Then there exists a direct
triangle $\tau$ so that \eqref{lema_inversion_dual} applies to $\rho_1,\rho_2$. Moreover, for $s=\partial_0 \rho_1$
we have $\beta_{s} = \tau \rho_2\primas \rho_1\prima$ and
$B_s^{H}=_\psi 1$. Then using also \eqref{definicion_ribb_recursiva}
we have $L_{\rho}^{h\prima}=_{\psi} \sum_{m\in G}T_{\rho_1\prima}^m
L_{\rho_1}^{\inver m h\prima m}=_{\psi} \sum_{k,l\in G}
F_{\rho_2}^{l\inver k \inver h \prima k \inver l,l}
T_{\rho_1\prima\tau}^k =_{\psi}\sum_{l\in
G}F_{\rho_2\rho_2\primas}^{l \inver h \prima \inver l,l}
T_{\rho_1\prima\tau\rho_2\primas}^H =_{\psi}\sum_{l\in
G}F_{\rho\prima}^{l \inver h \prima \inver l,l} $. Together with
\eqref{coroloario_bloques_dual}, this gives
\eqref{teorema_ribbon_inversion_open}.

\noindent \emph{(ii.a)} From \eqref{lema_inversion_direct} we have
$T_\sigma^g=_\psi T_{\sigma\prima}^{\inver g}$. We can set
$\sigma=\tau\prima\rho$, $\sigma\prima=\rho\prima\comp\tau\prima$
with $\tau\prima$ a direct triangle. The strip $\sigma_0=\comp \tau
\rho\tau\rho\prima$ is a nice closed strip, and indeed a block. Then
(\ref{definicion_ribb_recursiva},\ref{coroloario_bloques_dual}) give
$L_{\sigma}^{h\prima}=_\psi \sum_{k\in G} T_{\sigma}^k
L_{\tau\rho\prima\comp\tau}^{\inver k \inver h\prima k}=_\psi
\sum_{k,l\in G} T_{\sigma}^k L_{\tau}^{\inver k \inver h\prima k }
F_{\rho\prima}^{\inver k \inver h\prima k ,l} L_{\comp\tau}^{\inver
l\inver k \inver h\prima k l}$. But \eqref{lema_inversion_dual}
implies $L_\tau^g=\sum_{k\in G} T_{\tau\prima}^k L_{\comp
\tau}^{\inver k \inver g k}$ for any $g\in G$, and then
$L_{\sigma}^{h\prima}=_\psi \sum_{k\in G} T_{\sigma}^k
L_{\sigma\prima}^{\inver k \inver h\prima k}$. The result follows.

\noindent \emph{(ii.b)} From \eqref{lema_inversion_dual} we have
$L_{\rho}^{h\prima}=_\psi \sum_{k,l\in G} F_{\rho\prima}^{l\inver k
\inver h\prima k \inver l,l} T_\tau^k$. We can set
$\sigma=\tau\prima\rho$, $\sigma\prima=\rho\prima\comp\tau\prima$
with $\tau\prima$ a dual triangle. The strip $\sigma_0=\comp \tau
\rho\tau\rho\prima$ is a nice closed strip, and indeed a block. Then
(\ref{definicion_ribb_recursiva},\ref{coroloario_bloques_direct})
give $T_{\sigma}^{gH}=_\psi T_{\rho}^{gH}=_\psi
T_{\tau\rho\prima\inver\tau}^{\inver gH} =_\psi \sum_{k\in H\prima}
T_{\tau}^k T_{\sigma\prima}^{\inver k\inver g k H}$. Using \eqref{F_segun_TL} the result
follows.
\end{proof}

\subsection{Deformations in $\Ribb_\rho$, $\Ribbclosed_\sigma$, $\Ribbopen_\rho$ and
$\Ribbclosed\prima_\sigma$.}\label{apendice_deformaciones_algebras}

We are now in position to discuss the transformation properties of
the ribbon operator algebras introduced in appendix
\ref{appendix_B}. We distinguish three cases, which depend on the
values of the subgroups $N,M$ that label the Hamiltonian
\eqref{Hamiltoniano_NM}.

\subsubsection{The original Kitaev model: $N=1$, $M=G$.}

In this case we are interested in the algebras $\Ribb_\rho$ and
$\Ribbclosed_\sigma$. As for the first, open ribbons can be deformed
so that if $\rho \simeq_\psi^{1G} \rho\prima$ then
$F_\rho^{h,g}=_\psi F_{\rho\prima}^{h,g}$. That is, the action of
$\Ribb_\rho$ is invariant as long as ribbons are deformed without
crossing any excitation. They can also be reversed: if $\rho
\risingdotseq_\psi^{1G} \rho\prima$ then $F_\rho^{h,g} =_\psi
F_{\rho\prima}^{\inver g \inver h g,\inver g}$. Regarding closed
ribbons, the action of $\Ribbclosed_\sigma$ is invariant under
deformations ($\simeq_\psi^{1G}$) or rotations
($\eqcirc_\psi^{G1}$). Closed ribbon inversions give charge
inversion: if $\sigma \doteqdot_\psi^{1G} \sigma\prima$ then
$K_\sigma^{RC} =_\psi K_{\sigma\prima}^{\inver R^C \inver C}$, where
$R^C:= R^{g}$ (as defined in section
\ref{apendice_induced_representations}) with $r_{\inver C}= \inver g
\inver r_C g$ for some $g\in G$.

\subsubsection{String tension: $N$ and $M$ normal, $N$ central in $M$.}

In this case we are interested in the algebras $\Ribbopen_\rho$,
which gives the domain wall fluxes, and $\Ribbclosed_\sigma\prima$,
which gives the charges. The action of $\Ribbopen_\rho$ is invariant
under deformations which do not cross confined excitations
($\simeq_\psi^{MN}$), even if ends change as long as they do not
cross a domain wall ($\asymp_\psi^{MN}$). Inversions
($\risingdotseq_\psi^{MN}$) give domain flux inversion: $(\chi,t)$
goes to $(\inver \chi^{t}, \inver t)$. The action of
$\Ribbclosed_\sigma\prima$ is invariant under deformations
($\simeq_\psi^{NM}$) or rotations in which the end of $\sigma$ does
not cross domain walls ($\eqcirc_\psi^{MN}$). Charge inversion
($\doteqdot_\psi^{NM}$) is as follows: $(R,C)$ goes to $(\inver
R^C,\inver C)$ where $R^C:= R^{m}$ with $r_{\inver C}= \inver m
\inver r_C m$ for some $m\in M$.

\subsubsection{Domain walls: $N$ normal and abelian.}

In this case we are only interested in domain wall fluxes, that is,
in $\Ribbopen_\rho$. Its action is invariant under deformations
(those allowed by $\simeq_\psi^{NN}$), even if ends change as long
as they do not cross a domain wall ($\asymp_\psi^{MN}$). Domain wall
flux inversion ($\risingdotseq_\psi^{NN}$) is as follows: $(R,T)$
goes to $(\inver R^{T}, \inver T)$, where $R^{T}:=R^{r_T m}$ with
$r_{\inver T}M=\inver m\inver r_TM$ for some $m\in M$, so that
$\Norm {\inver T}=\inver m\inver r_T \Norm T r_T m$.

\subsection{Charge types}
\label{apendice_charge_types}

The previous results about closed ribbon transformations must be
complemented with the following one, which relates proper closed
ribbon operators with local vertex and face operators. Let $N,
M\subset G$ be normal subgroups in $G$ with $N$ central in $M$, and
define for $R\in\Irr{\Norm C\prima}$ and $C\in \Conj{G/N}$
\begin{equation}\label{definicion_D_RC}
D_s^{RC}:= \frac {n_R} {|\Norm C\prima|} \sum_{D} \sum_{q\in Q_C}
\sum_{d\in D} \inver \chi_R(d) A_s^{qd\inver q}B_s^{q r_C \inver q}
\end{equation}
where $D$ runs over $\Conj{\Norm C\prima}$.

\begin{prop}\label{prop_local_global_closed}
Let $s$ be a site, $\sigma$ a closed ribbon and $\tau$ a dual
triangle with $\beta_s \doteqdot_{V_\sigma,\tau} \sigma$. If
$\ket\psi\in\Hilb_G$ is such that $A_v^M=_\psi 1$ for any vertex
$v\neq v_s$ in $f_s$ and $L_{\tau}^{N} =_\psi 1$ then
\begin{equation}
K_\sigma^{RC} = D_s^{RC},
\end{equation}
where $R\in \Irr{\Norm C\prima}$, $C\in \Conj {G/N}$.
\end{prop}

\begin{proof}
Let $s\prima$ be the second site of $\sigma$, so that
$v_{s\prima}=v_s$ and $e_\tau$ does not belong to $f_{s\prima}$. The
states $\ket{\psi^g}:=A_v^M B_{s\prima}^g$, $g\in G$, are such that
$\beta_s \doteqdot_{\psi^g}^{NM} \sigma$. Then with the notation of
\eqref{teorema_ribbon_inversion_closed} we have $\sum_{k\in H\prima}
F_{\sigma}^{\inver k h\prima k, \inver k S k} B_{s\prima}^g =
\sum_{k\in H\prima} B^{\inver k h\prima k g \inver k \inver h\prima
k}_{s\prima} F_{\sigma}^{\inver k h\prima k, \inver k S k}=_{\psi^g}
\sum_{k\in H\prima} B^{\inver k h\prima k g \inver k \inver h\prima
k}_{s\prima} B_{s}^{\inver k S k}=_{\psi^g} \sum_{k\in H\prima}
A_s^{\inver k h\prima k}B_{s}^{\inver k S k}B_{s\prima}^g$. Since
$\ket\psi=\sum_g B_{s\prima}^g \ket{\psi^g}$, the result follows.
\end{proof}

\section{Local degrees of freedom}
\label{appendix_D}

In this appendix we give the details of the results indicated in
section \eqref{seccion_kitaev_cargas}. Choose any $C\in\Conj G$ and
two indices $i,i\prima$ and define
\begin{equation}
\ket {n} := \ket {n; i,i\prima} := F_\rho^{\inver c_i, \inver q_i n
q_{i\prima}} \ket{\psi_G},
\end{equation}
where $\ket{\psi_G}$ is a ground state of $H_G$ \eqref{Hamiltoniano_kitaev}.
Let $V$ be the space with basis $\ket n$, $n\in \Norm C$. Then there
exists an evident isomorphism $\funcion p {\C[\Norm C]} {V}$. For
$n,n\prima\in \Norm C$ and $s=\partial_0 \rho$,
$s\prima=\partial_1\rho$, consider the operators
\begin{equation}
a_{n,n\prima} := A_s^{\inver q_i n q_i}A_{s\prima}^{\inver
q_{i\prima} n\prima q_{i\prima}}.
\end{equation}
They give a representation $\funcion {a} {\Norm C\times \Norm C}{\GL
V}$ because
\begin{equation}
a_{n_1,n_2} \ket n = \ket{n_1 n \inver n_2}
\end{equation}
so that if $\funcion {\mathcal R} {\Norm C\times \Norm C}{\GL
{\C[\Norm C]}}$ is the representation of appendix A, we have
$a_{n_1,n_2} \spc p = p \spc \mathcal R_{n_1,n_2}$. This has several
consequences. First, if we define in accordance with
\eqref{definicion_E_R_ij} a basis for $V$ with elements
\begin{equation}
\ket {R;j j\prima} := \sum_{n\in \Norm C} \bar \Gamma_R^{jj\prima}
(n) \ket{n}
\end{equation}
then in the new basis
\begin{equation}
a_{n,n\prima} \ket {R;j j\prima} = \sum_{k,k\prima = 1}^{n_R}
\Gamma_R^{kj}(n) \comp \Gamma_{R}^{k\prima j\prima}(n\prima) \ket
{R; k k\prima}.
\end{equation}
In $\C[\Norm C]$ from \eqref{propiedades_e_R_ij} we get
\begin{align}
e_R^{uv} e_{R\prima}^{jj\prima} e_R^{v\prima u\prima} =
\delta_{R,R\prima} \delta_{v,j}\delta_{v\prima,j\prima}
e_R^{uu\prima}, \\%
e_R e_{R\prima}^{jj\prima} = e_{R\prima}^{jj\prima} e_R
=\delta_{R,R\prima} e_R^{jj\prima}
\end{align}
which through the isomorphism $p$ give
\begin{align}\label{cambio_arbitrario_1}
a_R^{uv} {a\prima_{\bar R}}^{u\prima v\prima} \ket {R\prima;j
j\prima} = \delta_{R,R\prima} \delta_{v,j} \delta_{v\prima,j\prima}
\ket {R;uu\prima}, \\%
a_R \ket {R\prima;j j\prima} = a\prima_{\bar R} \ket {R\prima;j
j\prima} = \delta_{R,R\prima} \ket {R;jj\prima},
\end{align}
where
\begin{align}
a_R^{uv} := \frac {n_R}{|\Norm C|} \sum_{n\in \Norm C} \bar
\Gamma_R^{uv} (n) \spc a_{n,1}, \\%
{a\prima_{R}}^{uv} := \frac {n_R}{|\Norm C|} \sum_{n\in \Norm C}
\bar \Gamma_R^{uv} (n) \spc a_{1,n}\\%
a_R = \sum_{u=1}^{n_R} a_R^{uu}, \qquad a\prima_R = \sum_{u=1}^{n_R}
{a\prima_R}^{uu}.
\end{align}
Note that $a_R^{uv},a_R\in \D_s$, ${a\prima_R}^{uv},a\prima_R\in
\D_{s\prima}$.

Finally, from \eqref{conmutacion_extremos_B} we have
\begin{equation}\label{cambio_arbitrario_2}
B_s^{c_k} B_{s\prima}^{\inver c_{k\prima}} \ket {n; i,i\prima} =
\delta_{k,i} \delta_{k\prima,i\prima} \ket {n; i,i\prima}
\end{equation}
and from \eqref{conmutacion_extremos_A}
\begin{equation}\label{cambio_arbitrario_3}
A_s^{\inver q_k q_i} A_{s\prima}^{\inver q_{k\prima} q_{i\prima}}
\ket {n; i,i\prima} = \ket {n; k,k\prima}.
\end{equation}
Note that $\ket{R; jj\prima}$ is just a shorthand for
\eqref{definicion_estado_RCiijj}. Finally, these results must be
complemented with proposition \ref{prop_local_global_closed}.

\section{Single-quasiparticle states}
\label{apendice_excitaciones_solitarias}

\begin{figure}
\includegraphics[width=5 cm]{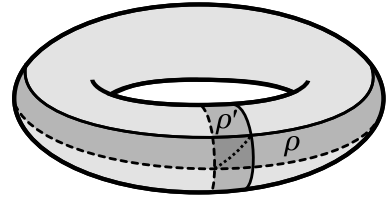}
\caption{
In a torus we can find a pair of closed ribbons $\sigma,\sigma\prima$ such that they form a crossed joint, $(\sigma,\sigma\prima)_+$. This is not possible in a sphere.}\label{figura_crossed_joint}
\end{figure}

Only in a surface of non-trivial topology can we find two closed
ribbons $\sigma,\sigma\prima$ such that $(\sigma,\sigma\prima)_+$, see Fig.~\ref{figura_crossed_joint}.
When such ribbons exist, we can construct for any $h,g\in G$ the
state
\begin{equation}
\ket{\psi_{hg}} := F_{\sigma}^{hg}L_{\sigma\prima}^{\inver g}\prod_v
A_v \vket 1.
\end{equation}
The state $\ket\psi$ is not zero, because (\ref{conmutacion
crossed_joint},\ref{conmutacion_extremos_closed})
\begin{equation}
L_{\sigma\prima}^{g}L^{\inver h}_\sigma \ket{\psi_{hg}} =\prod_v
A_v\vket 1.
\end{equation}
At most, it can have an excitation at $(v,f)=\partial
\sigma=\partial \sigma\prima$. In fact
\eqref{conmutacion_extremos_closed}
\begin{equation}
B_f \ket{\psi_{hg}} = \delta_{gh,hg} \ket{\psi_{hg}},
\end{equation}
showing that for non-abelian groups single-quasiparticle excitations
exist.

\section{Condensation}
\label{apendice_condensacion}

In this appendix we give the details of the calculations of certain
expected values for ribbon operators $\langle F \rangle$ for a
ground state of the Hamiltonian \eqref{Hamiltoniano_NM} for
$N\subset M\subset G$ subgroups of $G$, $N$ normal. Such ground states are characterized by the conditions \eqref{condiciones_GS_locales}.
For $S\subset G$, $g\in G$ we introduce the notation
\begin{equation}
\delta_{g,S}:=\delta_{gS,S}.
\end{equation}

\begin{prop}
Let $h,g\in G$, $n\in N$ and $\ket\psi,\ket\psi\prima\in \Hilb_G$ satisfy \eqref{condiciones_GS_locales}.

\noindent(i) For an arbitrary ribbon $\rho$
\begin{align}
    F_\rho^{h,g}\ket\psi &= \delta_{g,M}F_\rho^{hn,g}\ket\psi,\label{GS_any_ribbon} \\
    \bra{\psi\prima} F_\rho^{h,g} \ket\psi &= \delta_{h,M} \bra{\psi\prima} F_\rho^{h,gn}\ket\psi.\label{GS_xpctd_any_ribbon}
\end{align}

\noindent(ii) If $\rho$ is an open ribbon
\begin{align}
    F_\rho^{NM}\ket\psi &= \ket\psi,\label{GS_open_ribbon} \\
    \bra{\psi\prima} F_\rho^{h,g}\ket\psi&=\delta_{h,N}\delta_{g,M} \frac 1 {|M|} \braket {\psi\prima}{\psi}.\label{GS_xpctd_open_ribbon}
\end{align}

\noindent(iii) If $\sigma$ is a boundary ribbon
\begin{align}
    F_\sigma^{MN}\ket\psi &= \ket\psi,\label{GS_boundary_ribbon} \\
    \bra{\psi\prima} F_\sigma^{h,g}\ket\psi &= \delta_{h,M}\delta_{g,N} \frac 1 {|N|} \braket {\psi\prima}{\psi}.\label{GS_xpctd_boundary_ribbon}
\end{align}

\end{prop}

\begin{proof}
\noindent(i) If $\rho$ is a triangle this is a direct consequence of
the identities $L_\tau^nL_\tau^N=L_\tau^N$, $T_\tau^g T_\tau^M =
\delta_{g,M} T_\tau^g T_\tau^M$, $T_\tau^M L_\tau^h
T_\tau^M=\delta_{h,M} T_\tau^M L_\tau^h T_\tau^M$ and $L_\tau^N
T_\tau^g L_\tau^N= L_\tau^N T_\tau^{gn} L_\tau^N$. For general
ribbons, just apply \eqref{definicion_ribb_recursiva}.

\noindent(ii) From (i) we get $F_\sigma^{NM}\ket\psi =
F_{\sigma}^{1G} \ket\psi = \ket\psi$ using \eqref{unidades}. Let $s_i =
\partial_i \rho$ and set $\xpctd \cdot:=\bra{\psi\prima}\cdot \ket\psi$. Then from (\ref{conmutacion_extremos_B}, \ref{condiciones_GS_locales}) we have $\xpctd {F_\rho^{h,g}} = \xpctd {B_{s_0}^N
F_\rho^{h,g}} = \xpctd {F_\rho^{h,gm}B_{s_0}^{Nh}} =\delta_{h,N}
\xpctd {F_\rho^{h,gm}}$ and for $m\in M$ from (\ref{conmutacion_extremos_A}, \ref{condiciones_GS_locales}) we have $\xpctd
{F_\rho^{h,g}} = \xpctd {F_\rho^{h,g} A_{s_1}^m} =  \xpctd
{A_{s_1}^m F_\rho^{h,gm}} = \xpctd {F_\rho^{h,gm}}$. Thus $ \xpctd
{F_\sigma^{h,g}} = \delta_{g,M}\delta_{h,N} \xpctd {F_\sigma^{1,1}}$
and the result follows since $\xpctd {F_\rho^{1,M}}=\xpctd
{F_\rho^{1,G}}=\xpctd 1$.

\noindent(iii) Using the notation of appendix \ref{appendix_C}, $p_\sigma$ encloses a disc $R\subset R_\psi^N$ so
that $F_\sigma^{MN}\ket\psi = \ket\psi$. Also,
$L_\sigma^m\ket\psi=\ket\psi$ for any $m\in M$. To check this,
suppose for example that the edges $\Edual_\rho$ lie outside $R$ and
choose for each vertex $v$ in $R$ a ribbon $\rho_v$ with
$p_{\rho_v}$ a path inside $R$ from $v_0=v_{\partial_0\rho}$ to $v$.
If we set $A^m_\rho= A_{v_0}^m \prod_{v\neq v_0} \sum_k
 T_{\rho_v}^k A_v^{\inver kmk}$, with the product running over
all vertices in $R$, one can check that $\ket\psi=A^m \ket\psi =
L_\sigma^m\ket\psi$. The other case is similar. Thus, for $m\in M$
we get $\xpctd {F_\sigma^{h,g}} = \delta_{g,N} \xpctd
{F_\sigma^{hm,g}}$, so that $ \xpctd {F_\sigma^{h,g}} =
\delta_{g,N}\delta_{h,M} \xpctd {F_\sigma^{1,1}}$. The result
follows since $\xpctd {F_\rho^{1,N}}=\xpctd {F_\rho^{1,G}}=\xpctd
1$.
\end{proof}

A state satisfying (\ref{GS_open_ribbon},\ref{GS_boundary_ribbon})
for all open ribbons $\rho$ and boundary ribbons $\sigma$ cannot
contain vertex, face or edge excitations. Therefore, these
conditions characterize ground states.

We proceed to check \eqref{condensacion}, the derivation of
\eqref{condensacion2} is similar. From (\ref{base_ribbclosed},
\ref{GS_xpctd_boundary_ribbon}) we get
\begin{equation}\label{esperado1}
    \xpctd {K_\sigma^{DC}} = \frac{|C\cap N|}{|N||G|} \sum_{g\in G} |D\cap \inver gMg|,
\end{equation}
where $D\in\Conj {\Norm C}$, $C\in \Conj G$. This together with
\eqref{base_ribbclosed_2} gives \eqref{condensacion} because if
${e_M \!\uparrow}$ is the induced representation in $G$ of the
identity representation in $M$
\begin{equation}
    \chi_{e_M \!\uparrow} (g)= \frac 1 {|M|} \sum_{k\in G} \delta_{g,\inver k M k}.
\end{equation}

As for \eqref{condensacion2auxiliar}, from
\eqref{GS_xpctd_open_ribbon} we have
\begin{equation}\label{subcond1}
|M|\xpctd {\sum_{n\in \Norm C} \Gamma_R^{jj\prima} (n)F_\rho^{\inver
c_i, \inver q_i n q_i\prima}} = \delta_{c_i,N} \sum_{n\in
M_c^{i,i\prima}} \Gamma_R^{jj\prima}(n),
\end{equation}
where $M_c^{i,i\prima}:=\Norm C\cap q_i M \inver q_i\prima$. If
$M_c^{i,i\prima}$ is empty, we are done. Else,
$M_c^{i,i\prima}=M_c^{i,i}s$ for some $s\in \Norm C$, so that
$\sum_{n\in M_c^{i,i\prima}}\Gamma_R (n) = \sum_{n\in
M_c^{i,i}}\Gamma_R (n) \Gamma_R(s)$. But\cite{serre} $\sum_{n\in
M_c^{i,i}}\Gamma_R (n)=0$ if $(\chi_R,1)_{M_C^{i,i}}=0$.

\end{document}